\documentclass[%
 reprint,
nofootinbib,
 amsmath,amssymb,
 aps,
prapplied,
]{revtex4-2}
\usepackage{graphicx, color}
\usepackage{subfigure}
\usepackage{bm}
\usepackage[colorlinks=true,citecolor=blue,linkcolor=magenta]{hyperref}

\usepackage{amsthm}
\theoremstyle{plain}
\newtheorem{theorem}{Theorem}
\newtheorem{lemma}[theorem]{Lemma}
\newtheorem{corollary}[theorem]{Corollary}
\theoremstyle{definition}

\newtheorem{assumption}[theorem]{Assumption}

\begin{document}

\preprint{APS/123-QED}

\title{Trainability Enhancement of Parameterized Quantum Circuits \\ via Reduced-Domain Parameter Initialization}

\author{Yabo Wang}
\affiliation{Key Laboratory of Systems and Control, Academy of Mathematics and Systems Science, Chinese Academy of Sciences, Beijing 100190, P. R. China}
\affiliation{University of Chinese Academy of Sciences, Beijing 100049, P. R. China}
\author{Bo Qi}
\email{qibo@amss.ac.cn}
\affiliation{Key Laboratory of Systems and Control, Academy of Mathematics and Systems Science, Chinese Academy of Sciences, Beijing 100190, P. R. China}
\affiliation{University of Chinese Academy of Sciences, Beijing 100049, P. R. China}
\author{Chris Ferrie}
\affiliation{Center for Quantum Software and Information, University of Technology Sydney, Ultimo NSW 2007, Australia}
\author{Daoyi Dong}
\affiliation{CIICADA Lab, School of Engineering, Australian National University, Canberra ACT 2601, Australia}

\date{\today}
\hyphenpenalty=5000
\tolerance=1000

\begin{abstract}
Parameterized quantum circuits (PQCs) have been widely used as a machine learning model to explore the potential of achieving quantum advantages for various tasks. However, training PQCs is notoriously challenging owing to the phenomenon of plateaus and/or the existence of (exponentially) many spurious local minima. To enhance trainability, in this work we propose an efficient parameter initialization strategy with theoretical guarantees. We prove that by reducing the initial domain of each parameter inversely proportional to the square root of circuit depth, the magnitude of the cost gradient decays at most polynomially with respect to qubit count and circuit depth. Our theoretical results are substantiated through numerical simulations of variational quantum eigensolver tasks. Moreover, we demonstrate  that the reduced-domain initialization strategy can protect specific quantum neural networks from exponentially many spurious local minima. Our results highlight the significance of an appropriate parameter initialization strategy, offering insights to enhance the trainability and convergence of variational quantum algorithms.
\end{abstract}

\maketitle


\section{Introduction}

Variational quantum algorithms (VQAs) with parameterized quantum circuits (PQCs)~\cite{McClean2016,Cerezo2021v,Jerbi2023,liu2021a,Benedetti2019}  have emerged as a leading strategy
to explore the power of quantum computing in the current noisy intermediate-scale quantum era.  VQAs are optimization-based quantum-classical hybrid algorithms. The core idea is exploiting PQCs to delegate classically difficult computation and updating the variational parameters in PQCs with a classical optimizer \cite{Sweke2020,Dong2022,Basheer2022,Lavrijsen2020}.
VQAs have been applied across various practical fields, including machine learning \cite{Biamonte2017,Mitarai2018,Farhi2018,Wu2020,Perrier2020}, error correction \cite{Xu2021vari,Cao2022quantumv,johnson2017qvector}, dynamical simulations \cite{Yuan2019,Altman2021,Bharti2021q,Wen2019,Tabares2023}, combinatorial optimization \cite{Svensson2023,Farhi2014,Crooks2018,Vikst2020}, among others \cite{Wang2021variational,Jiang2018,Liang2022}.


Despite the success in many fields, training VQAs still faces many challenges when scaling up the size of PQCs. One of the key bottlenecks is the magnitude of the cost gradient decaying exponentially with the number of qubits, also known as barren plateaus (BPs)~\cite{McClean2018}. Once BPs occur, training  PQCs is prohibited  unless exponentially many copies of quantum states are provided to identify the update direction through measurement.  BPs have been found arising in many scenarios, such as randomly initialized PQCs \cite{McClean2018,Haug2021c}, deep and problem-agnostic PQCs \cite{McClean2018,Zhao2021}, global cost function \cite{Cerezo2021c,Liuzidu2021}, large expressibility \cite{McClean2018,Holmes2022}, entanglement \cite{Marrero2021,Sharma2022,Patti2021}, and noise \cite{Wangsamon2021}.  In addition to BPs, there is another bottleneck for training. For a wide class of under-parameterized variational quantum models, the landscapes of the cost function are generally swamped with spurious local minima and traps \cite{Anschuetz2022b}. These quantities may be exponentially many in relation to the number of parameters \cite{You2021}, making it extremely hard to find a global minimum~\cite{Bittel2021,Wierichs2020,Ge2022}.

To address the trainability issue, various strategies have been proposed. For instance, certain constructed PQCs, including tree tensor network \cite{Shi2006,Grant2018} and quantum convolutional neural network \cite{Cong2019}, can help avoid BPs \cite{Zhao2021,Zhang2020,Pesah2021,barthel2023}. However, due to their specific architectures, where the circuit depth is logarithmic with respect to the number of qubits, the gradient usually decays exponentially with the circuit depth instead. Ansatzes constructed in a problem-inspired manner \cite{Hadfield2019,Wiersema2020,Lyu2023symmetry,Sharma2022} have been put forward to mitigate BPs by constraining the explored space of unitaries.  Ansatzes can also be trained layer-by-layer \cite{Skolik2021,Campos2021}, or designed using classical reinforcement learning \cite{Lockwood2021}. Since BPs are cost function-dependent, it was proposed to encode the problem into local cost function \cite{Cerezo2021c}.
In addition, over-parametrization has been investigated to address the training issue caused by the presence of many spurious local minima or traps in the cost landscapes far away from the global minimum \cite{Liu2022a,You2022conv, Anschuetz2022c}. However, efficient over-parametrization, i.e., with polynomially many parameters, is achievable only under some conditions for quantum neural networks (QNNs) \cite{Larocca2021}.  Generally, the number of parameters required for over-parametrization is comparable to the system dimension \cite{You2022conv}. This exponential overhead  may incur BPs, especially when considering noises on quantum hardware.


From a practical viewpoint, an effective parameter initialization strategy is crucial for the trainability of PQCs and the convergence of VQAs~\cite{Zhou2020,McClean2018,Holmes2022,Bittel2021,Wiersema2020}.
An identity block initialization strategy has been introduced with the purpose of limiting the effective depth of circuits~\cite{Grant2019}. It was demonstrated that ansatzes with correlated parameters can effectively restrict the optimized parameter space and result in large gradients of cost function~\cite{Volkoff2021}.  A Bayesian learning initialization was employed  to identify a promising region in the parameter space before performing local search~\cite{Rad2022}. Simulations conducted in \cite{Haug2021c} revealed an initial region of parameters that can improve the trainability while maintaining a high level of expressibility. A Gaussian initialization was proposed with theoretical guarantees to escape from BPs in general deep circuits \cite{Zhang2022g}, where the lower bound of the norm of cost gradient decays at most polynomially with the circuit depth. Besides, for correlated parameterized gates, Gaussian initialization guarantees the update direction when training global observables~\cite{Zhang2022g}. However, the theoretical or empirical results of the aforementioned strategies are confined to single Pauli strings.  For complex objectives corresponding to a summation of local Pauli strings,  a  Floquet initialization strategy was proposed to mitigate BPs~\cite{cao2024exploiting}. They demonstrated that by initializing PQCs within the many-body localization phase, the gradient does not vanish as the system size grows.  Nevertheless, the phase transitions between the many-body localization and thermalization (where BPs occur) in Floquet systems need to be characterized through additional simulations.

In this work, we propose an efficient reduced-domain initialization strategy to enhance trainability of PQCs. By focusing on the variational quantum eigensolver (VQE), we theoretically prove that for  an $L$-depth circuit, if the initial domain of each parameter is reduced in proportion to $\frac{1}{\sqrt{L}}$, the norm of the cost gradient decays at most polynomially with $L$ as the number of qubits and the depth increase. Our result is quite different from the belief that BPs occur when the circuit depth is polynomial of the number of qubits. This is because the reduced-domain initialization can help maintain a balance between the trainability and expressibility of PQCs. Via  simulations, we first validate our theoretical results. Then we show that  our reduced-domain initialization can protect QNNs from the phenomenon of exponentially many spurious local minima in cost function landscapes, and help achieve improved performance.
	
The paper is organized as follows. In Section~\ref{notation}, we  describe the generic setup of VQAs. In Section~\ref{main}, we  present the theoretical results of enhancing trainability of PQCs through the reduced-domain parameter initialization strategy. In Section~\ref{experiments}, via simulations we demonstrate the enhanced trainability in VQE tasks, as well as the improved convergence in both VQE and QNN instances. Section~\ref{conclusion} concludes the paper.

\section{\label{notation} Setup \lowercase{of} VQAs}

VQAs provide a general framework for solving a wide range of problems. A typical VQA is comprised of basic elements: cost, ansatz, measurement, optimizer as well as  parameter initialization.
	
\textbf{Cost:} The first task of VQAs is to identify an appropriate cost function which effectively encodes the problem to be solved. Without loss of generality, we consider the expectation value of an observable $O$ as cost function in the form of
	\begin{equation}\label{cost}
	C\bm{(} \boldsymbol{\theta} \bm{)}= \mathrm{Tr}\left[ OU( \boldsymbol{\theta} )\rho U^{\dagger}( \boldsymbol{\theta} )\right].
	\end{equation}
Here, $\rho$ is the input state, and $U( \boldsymbol{\theta} )$ denotes a parameterized quantum unitary  with the parameter vector $\boldsymbol{\theta}$ to be optimized. In VQAs, an implicit assumption  is that the cost function should be classically computationally intensive and the value can be obtained by performing quantum measurements on a quantum circuit with possible classical postprocessing.

In this work, our aim is finding the least eigenvalue of a given Hamiltonian $H$. In the following, given a positive integer $n$, denote by $[n]$ the set $\{1, 2, \cdots, n\}$.  Let
	\begin{equation}\label{paulitensor}
	\boldsymbol{\sigma}_{\boldsymbol{i}}=\boldsymbol{\sigma}_{(i_{1}, i_{2}, \cdots, i_{N})}=\sigma_{i_{1}}\otimes\sigma_{i_{2}}\otimes\cdots\otimes\sigma_{i_{N}}
	\end{equation}
	denote an $N$-qubit Pauli string, where $i_{j}\in \left\lbrace 0,1,2,3\right\rbrace $ for $ j \in [N] $, and $\sigma_0=I$, $\sigma_1=X$, $\sigma_2=Y$, $\sigma_3=Z$ are Pauli matrices. For instance, we denote
	\begin{align}
	\nonumber
	&\boldsymbol{\sigma}_{(3, 0, 0, \cdots, 0)}=Z\otimes I\otimes I \otimes\cdots\otimes I=Z_{1},\
	\\
	\nonumber
	&\boldsymbol{\sigma}_{(3, 3, 0, \cdots, 0)}=Z\otimes Z\otimes I \otimes\cdots\otimes I=Z_{1}Z_{2}.
	\end{align}
Moreover, $\boldsymbol{\sigma}_{\boldsymbol{i}}$ is said to be $S$-local if  there are  $S$ non-identity components in the tensor product of  Eq.~\eqref{paulitensor}.

We choose $\rho=\left| 0\rangle\langle 0\right|$ and set $O=H$ in Eq.~\eqref{cost}. We consider $H$ in the form as
	\begin{equation}\label{generalH}
	H=\sum_{\boldsymbol{i}\in \mathcal{N}}{\boldsymbol{\sigma}_{\boldsymbol{i}}}=\sum_{\boldsymbol{i}\in \mathcal{N}}{\boldsymbol{\sigma}_{(i_{1}, i_{2}, \cdots, i_{N})}},
	\end{equation}
where $\mathcal{N}$ is a subset of  $\mathcal{N}_S=\{\boldsymbol{i}=(i_{1}, i_{2}, \cdots, i_{N}): 1\leq \sum_{j=1}^{N} \mathrm{1}_{i_j\neq 0}\leq S\}$, and $\mathrm{1}_{\omega}$ is the indicator function of the event $\omega$. That is, $H$ is a sum of  $N$-qubit Pauli strings with locality of at most $S$. This model encompasses ground state determination problems for various quantum many-body systems and chemical molecules \cite{Alam2022pra,lewis2024improved,McArdle2020,cao2019quantum}. In addition, the Hamiltonian Eq.~(\ref{generalH}) can be employed to tackle many NP-complete problems~\cite{Lucas2014}.

\begin{figure}[h]
	\flushleft
 \includegraphics[height=4.9cm]{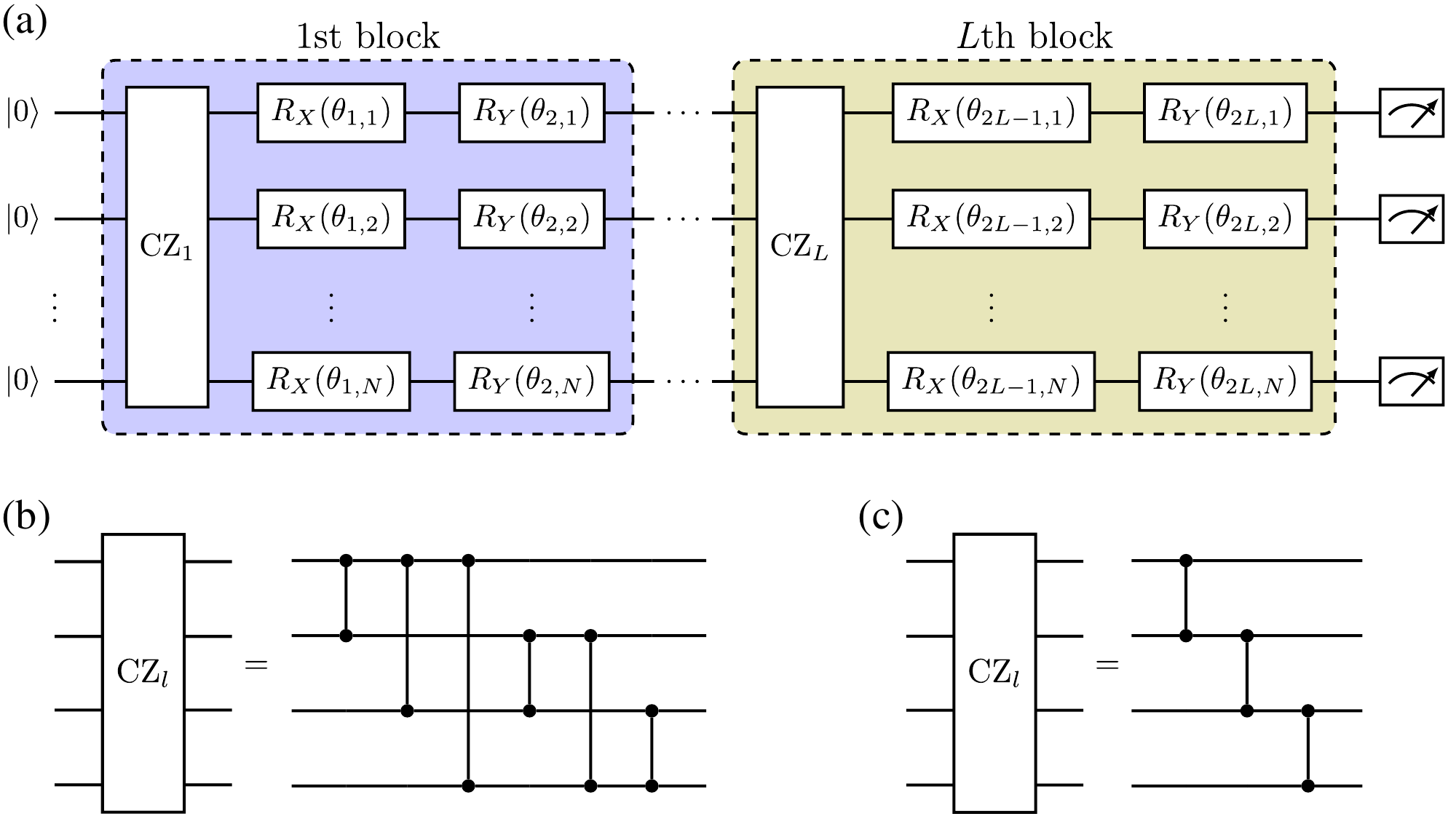}
		\caption{\label{fig:setup}Setup of HEA. (a) PQC for VQE. It consists of $L$ blocks. In the $l$th block, we first perform the entanglement layer composed of  $\mathrm{CZ}$ gates, and then successively perform $R_X$ gate and $R_Y$ gate on each qubit. (b) $\mathrm{CZ}_{l}$ layer with the fully connected topology. (c) $\mathrm{CZ}_{l}$ layer with the nearest-neighbor pairs topology and open boundary condition.}
 \end{figure}

\textbf{Ansatz:} Another key element of VQAs is ansatz, i.e., PQC whose parameters are optimized to minimize the corresponding cost function. In this work, we consider a PQC similar to that in \cite{Zhang2022g}, called the hardware-efficient ansatz (HEA). As illustrated in Fig.~\ref{fig:setup}(a), it is comprised of $L$ blocks, each of which is composed of an entanglement layer $\mathrm{CZ}_l$ including $\mathrm{CZ}$ gates on arbitrarily many qubit pairs and two layers of $R_X$ and $R_Y$ gates on all qubits. There are two widely used structures of the entanglement layer: the fully connected topology as shown in Fig.~\ref{fig:setup}(b), and the nearest-neighbor pairs topology with open boundary condition as shown in Fig.~\ref{fig:setup}(c).  Notably, our results do not rely on the specific topology of $\mathrm{CZ}_l$ and can also be generalized to other structures of PQC\@.

\textbf{Measurement:} In our VQE task, the Hamiltonian  is formulated in terms of Pauli strings. The expectation values of the $\sigma_3$ terms can be directly measured by reading out the probabilities of the computational basis, while to obtain the expectation values of the  $\sigma_1$ and $\sigma_2$ terms, we should rotate them into the $\sigma_3$ basis first.
To minimize the number of measurements needed, we should group different Pauli strings that can be measured simultaneously. There is  still an urgent need for more efficient compilation of measurement circuits.

\textbf{Optimizer:} The core task of the classical processor is updating the parameters to minimize the cost. In our VQE task, the cost gradient can be obtained by utilizing the parameter-shift rule \cite{Mitarai2018}, which also suffers from the measurement efficiency issue. In this work, we aim to analyze the norm of the gradient which greatly impacts the trainability of VQAs.

\textbf{Parameter initialization:} Due to the possibility of occurrence of BPs and/or the existence of exponentially many spurious local minima when training VQAs, it is critical to find an appropriate parameter initialization strategy to enhance trainability.
	
For the PQC  illustrated in Fig.~\ref{fig:setup}(a), there are $2L$ layers of parameterized rotation gates with $2NL$ variational parameters in total, denoted by
	\begin{equation}\nonumber
	\boldsymbol{\theta}=\boldsymbol{\theta}_{[2L]}=(\boldsymbol{\theta}_{1}, \boldsymbol{\theta}_{2}, \cdots, \boldsymbol{\theta}_{2L}),
	\end{equation}
	where $\boldsymbol{\theta}_{q}=(\theta_{q,1}, \theta_{q,2}, \cdots, \theta_{q,N})$ denotes the vector of rotation angles  in the $q$th layer for $q\in \left[2L \right]$. To be specific,  for the $l$th block, $l\in [L]$, and the $j$th qubit, $j\in [N]$, $\theta_{2l-1, j}$ denotes the rotation angle around $X$, while $\theta_{2l,j}$ denotes that around $Y$.
	
\section{\label{main}Trainability Enhancement}
\subsection{Lower bounds of gradient}
In previous works, the initial parameters are usually drawn independently and identically from a uniform distribution over an entire period $[-\pi, \pi]$, and the $N$-qubit local observable is always chosen as a simple Pauli string, for instance, $Z_{1}Z_{2}$. Under a PQC similar to ours,  Ref.~\cite{McClean2018} demonstrated that a BP occurs even when the depth $L$ is a modest linear function, that is, $L=10N$, and the exponential decay cannot be circumvented for any polynomial sum of these local operators. For the simple observable $Z_{1}Z_{2}$, Ref.~\cite{Haug2021c} numerically showed that the trainability can be improved if each parameter is randomly chosen from $[-a\pi, a\pi]$. The satisfying value of $a$ was identified by searching from $10^{-10}$ to $1$ via simulations.
	
In this work, we consider a general Hamiltonian in the form of Eq.~\eqref{generalH}, where the number of individual terms in the summation can be arbitrary. For the complex objective function, we propose a reduced-domain parameter initialization strategy, in which each parameter is independently and uniformly chosen from a reduced interval $[-a\pi,a\pi]$, where the hyperparameter $a$ is depth-dependent and will be given in the following theorems. Importantly, there is no constraint on the parameters' domain during the optimization process. That is, the variational parameters are free to leave their initial reduced-domains as guided by the negative gradient to minimize the objective function.
	
In view of the setup in Section~\ref{notation}, the decrease  of the cost function value after  an update of the parameters  can be quantified in terms of  the magnitude of $\Vert\nabla_{\boldsymbol{\theta}}{C}\Vert_{2}^{2}$ (see Appendix~\ref{intuition}). Thus, we adopt its average value $\mathop{\mathbb{E}}\limits_{\boldsymbol{\theta}}\Vert\nabla_{\boldsymbol{\theta}}{C}\Vert_{2}^{2}$ as the  trainability figure of merit.
	
Now we provide a theoretical guarantee for the reduced-domain initialization by establishing a lower bound on the norm of cost gradient at the first step. For the ease of notation, denote $\alpha=\frac{1}{4}\left(  2+\frac{\sin {2a\pi}}{a\pi}\right)$, $\beta=\frac{1}{4}\left(  2-\frac{\sin {2a\pi}}{a\pi}\right)$, and $\gamma=\frac{\sin {a\pi}}{a\pi}.$

   	\begin{theorem}\label{TH2}
		Consider the VQE problem, where the cost function is Eq.~\eqref{cost} with the Hamiltonian being Eq.~\eqref{generalH}, the ansatz illustrated as in Fig.~\ref{fig:setup}(a), the initial state chosen as  $\rho=|0\rangle\langle0|$, and each parameter in $\boldsymbol{\theta}$ chosen from $\left[ -a\pi,a\pi\right]$  independently  and uniformly.  Let $\nabla_{\boldsymbol{\theta}}{C}$ denote the gradient of the cost function with respect to $\boldsymbol{\theta}$. Then the average of the squared norm of the gradient is lower bounded by
		\begin{equation}\label{TH21}
		\mathop{\mathbb{E}}\limits_{\boldsymbol{\theta}}\Vert\nabla_{\boldsymbol{\theta}}{C}\Vert_{2}^{2}\geq 2|\mathcal{N}|(L-1)\alpha^{2SL-1}\beta^{S+1},
		\end{equation}
		where $|\mathcal{N}|$ denotes the cardinality of the set $\mathcal{N}$.
		
		Furthermore, if the hyperparameter $a$ is selected as the solution of the following equation
		\begin{equation}\label{TH2a}
		\frac{\sin {2a\pi}}{2a\pi}=\frac{S(2L-1)-2}{S(2L+1)},
		\end{equation}
		where $a\in \mathrm{\Theta}\left(\frac{1}{\sqrt{L}}\right)$,
		then the gradient norm can be lower bounded by
		\begin{equation}\label{key}
		\mathop{\mathbb{E}}\limits_{\boldsymbol{\theta}}\Vert\nabla_{\boldsymbol{\theta}}{C}\Vert_{2}^{2}\geq 2 |\mathcal{N}| \left(\frac{S+1}{eS}\right)^{S+1}\frac{L-1}{(2L+1)^{S+1}}.
		\end{equation}
	\end{theorem}
	
From Theorem~\ref{TH2}, for VQE problems with Hamiltonian being a summation of a series of local Pauli strings, the trainability can be greatly enhanced in the initial phase by appropriately reducing the initial domain of parameters. Specifically, the lower bound of the gradient norm decays at most polynomially with the circuit depth. In addition, the lower bound increases linearly with the number of individual Pauli strings in the Hamiltonian, which guarantees the efficiency of the reduced-domain strategy for complex objective functions.

An intuitive interpretation of the effectiveness of the reduced-domain initialization is as follows. It is well known that there is a conflict between the trainability  and  expressibility of PQCs.  Over-expressed PQCs may suffer from phenomena like BPs and/or the landscape being swamped with spurious local minima, making the training extremely difficult. Note that the expressibility of a PQC depends on the structure of the circuit as well as the size of the initial domains of the  parameters. It generally grows fast as the circuit depth $L$ increases and the range of the initial domains expands. Therefore, to  balance the conflict between trainability and expressibility of PQCs, we should appropriately reduce the initial domains of the parameters. As we have stated, to enhance the trainability of a PQC with depth $L$, the domain of each parameter should be reduced in proportion to $\frac{1}{\sqrt{L}}$. We will provide more detailed sensitivity analysis on the hyperparameter $a$ in Sec.~\ref{experiments}; see Table~\ref{tab:grad},
Fig.~\ref{fig:heg}(a), and Fig.~\ref{fig:hva}(a).

To further illustrate the result and explain the idea of proof of Theorem~\ref{TH2}, let us consider an example where the $N$-qubit Hamiltonian reads
	\begin{equation}\label{toyH}
	H=\sum_{j=1}^{N-1} Z_jZ_{j+1}.
	\end{equation}
	
\begin{lemma}{\label{TH1}}
		Consider the VQE problem, where the cost function is Eq.~\eqref{cost} with the Hamiltonian being Eq.~\eqref{toyH}, the ansatz illustrated as in Fig.~\ref{fig:setup}(a), the initial state chosen as  $\rho=|0\rangle\langle0|$, and each parameter in $\boldsymbol{\theta}$ chosen from $\left[ -a\pi,a\pi\right]$ independently  and uniformly.  Let $\nabla_{\boldsymbol{\theta}}{C}$ denote the gradient of the cost function with respect to $\boldsymbol{\theta}$. Then the average of the squared norm of the gradient is lower bounded by
		\begin{equation}\label{TH11}
		\mathop{\mathbb{E}}\limits_{\boldsymbol{\theta}}\Vert\nabla_{\boldsymbol{\theta}}{C}\Vert_{2}^{2}\geq 4(2N-3)L\gamma^{8L-2}\beta.
		\end{equation}
		
		Furthermore, if the hyperparameter $a$ is selected as 	
\begin{equation}
a=\frac{1}{4\pi}\sqrt{\frac{40L+7-\sqrt{1600L^{2}-400L+49}}{L}}\in \mathrm{\Theta}\left(\frac{1}{\sqrt{L}}\right),\label{TH1a}
\end{equation}
		then the gradient norm can be lower bounded by
		\begin{equation}\label{key}
		\mathop{\mathbb{E}}\limits_{\boldsymbol{\theta}}\Vert\nabla_{\boldsymbol{\theta}}{C}\Vert_{2}^{2}\geq e^{-1}(2N-3),
		\end{equation}
		where $e$ is the natural constant.
	\end{lemma}

It is clear that for the  Hamiltonian of Eq.~\eqref{toyH}, under the reduced-domain initialization, the trainability can be remarkably enhanced in the initial phase. Not only  the lower bound of the gradient norm does not decrease when the circuit depth increases, but also it increases as the number of qubits increases. In particular, for the simple $N$-qubit observable $Z_{1}Z_{2}$,  the gradient norm can be lower bounded by a constant under the reduced-domain initialization. This is quite different from the result in Ref.~\cite{McClean2018} with the domain being $[-\pi, \pi]$, where the gradient vanishes exponentially with respect to both the qubit number and the circuit depth.

The proofs of Lemma~\ref{TH1} and Theorem~\ref{TH2} are given in Appendix~\ref{proofTH1} and Appendix~\ref{proofTH2}, respectively. The basic idea is inspired by the proof of Theorem 4.1 in \cite{Zhang2022g}, where the observable considered is restricted to  a single local Pauli string. However, in our work, the observables are in the form of a summation of local Pauli strings, so that additionally we need to deal with  the cross terms generated from different Pauli strings when analyzing the norm of gradient. Specifically, we prove that under the selected initial state $\rho=|0\rangle\langle0|$, these cross terms are non-negative by analyzing the transformations of the involved Pauli strings after taking the expectations over the trainable parameters and applying the entanglement layers.

\subsection{Addressing  BPs}
We now prove that the reduced-domain  initialization strategy can help escape from BPs by considering  the variance of the partial derivative of the cost function with respect to specific variational parameters.

In general, an  $N$-qubit cost function $C$ is said to exhibit a BP in the $\theta_v$ direction,  if $\mathop{\mathbb{E}}\limits_{\boldsymbol{\theta}}\partial_{\theta_v}{C}=0$, and the variance decreases exponentially with $N$, that is,
	\begin{equation}\nonumber
	\mathop{\mathbb{V}\mathrm{ar}}\limits_{\boldsymbol{\theta}}\partial_{\theta_v}{C}\in \mathrm{O}\left(b^{N}\right),\ \text{for some}\  b \in \left( 0, 1\right).
	\end{equation}
	On the other hand, it is said to be trainable, if the variance decreases polynomially with $N$, that is,
	\begin{equation}\nonumber
	\mathop{\mathbb{V}\mathrm{ar}}\limits_{\boldsymbol{\theta}}\partial_{\theta_v}{C}\in \mathrm{\Omega} \left(\frac{1}{\mathrm{poly}(N)}\right).
	\end{equation}

Similarly, we can define the circuit-depth-dependent BP, as the depth is also a key figure of merit in VQAs. The number of qubits and the circuit depth are unified as the size of PQC\@.
	
Recall that once a BP occurs, the training of PQCs is prohibited. The following theorem states that the reduced-domain initialization strategy with a selected hyperparameter in the scaling of $\frac{1}{\sqrt{L}}$ can help PQCs avoid BPs. Remarkably, it holds for  PQCs whose depths  are polynomial functions of the number of qubits, while Ref.~\cite{McClean2018} demonstrated that for a wide class of randomly initialized PQCs of reasonably linear depth, BPs occur even with simple local observables.

	\begin{theorem}\label{TH4}
		Under the same assumptions as in Theorem~\ref{TH2}, let $\partial_{\theta_{q, n}}{C}$ denote the partial derivative of cost function with respect to $\theta_{q, n}$.	For $q\in \left[ 2L-2\right] $ and $n\in \left[ N \right] $,  the average of the squared norm of the partial derivative is lower bounded by
		\begin{equation}\label{TH41}
		\mathop{\mathbb{E}}\limits_{\boldsymbol{\theta}}\left( \partial_{\theta_{q, n}}{C}\right) ^{2}\geq \alpha^{2SL-1}\beta^{S+1}.
		\end{equation}
		
In addition, if for each  $\boldsymbol{i}\in \mathcal{N}$, the number of its elements belonging to the set $\left\lbrace 1, 2  \right\rbrace$  is not equal to one, then we have
		\begin{equation}\label{TH42}
		\mathop{\mathbb{E}}\limits_{\boldsymbol{\theta}}\partial_{\theta_{q, n}}{C}=0,
		\end{equation}
		and the variance is lower bounded by
		\begin{equation}\label{TH43}
		\mathop{\mathbb{V}\mathrm{ar}}\limits_{\boldsymbol{\theta}}\partial_{\theta_{q, n}}{C}\geq \alpha^{2SL-1}\beta^{S+1}.
		\end{equation}
		
Furthermore, if the  hyperparameter $a$ is selected according to Eq.~\eqref{TH2a} in Theorem~\ref{TH2},	where $a\in \mathrm{\Theta}\left(\frac{1}{\sqrt{L}}\right)$,  we have
		\begin{equation}\label{TH44}
		\mathop{\mathbb{E}}\limits_{\boldsymbol{\theta}}\left( \partial_{\theta_{q, n}}{C}\right) ^{2}\geq \left(\frac{S+1}{e S}\right)^{S+1}\frac{1}{(2L+1)^{S+1}}.
		\end{equation}
	\end{theorem}

Similar to Lemma~\ref{TH1}, to further illustrate the result and explain the idea of proof of Theorem~\ref{TH4}, we present the following lemma on the basis of Hamiltonian Eq.~(\ref{toyH}).

	\begin{lemma}\label{TH3}
		Under the same assumptions as in Lemma~\ref{TH1}, let $\partial_{\theta_{q, n}}{C}$ denote the partial derivative of cost function with respect to $\theta_{q, n}$. For $q\in \left[ 2L\right] $ and $n\in \left[ N \right] $, we have
		\begin{equation}\label{expectation1}
		\mathop{\mathbb{E}}\limits_{\boldsymbol{\theta}}\partial_{\theta_{q, n}}{C}=0,
		\end{equation}
		and for $q\in \left[ 2L\right] $, the variance is lower bounded by
		\begin{equation}\label{variance}
		\mathop{\mathbb{V}\mathrm{ar}}\limits_{\boldsymbol{\theta}}\partial_{\theta_{q, n}}{C}\geq f_{n}\bm{(}L, a\bm{)}.
		\end{equation}
		\begin{itemize}
			\item[$(\mathrm{i})$ ] For $n\in \left\lbrace 1, N\right\rbrace  $,
			\begin{equation}\label{TH31}
			f_{n}\bm{(}L, a\bm{)}=\alpha^{4L-1}\beta.
			\end{equation}
			Furthermore, if the  hyperparameter $a$ is selected as the solution of  the following equation
			\begin{equation}\label{key}
			\frac{\sin {a\pi}}{a\pi}=\frac{2L-1}{L},		
			\end{equation}
			where $a\in \mathrm{\Theta}\left(\frac{1}{\sqrt{L}}\right)$, then we have
			\begin{equation}\label{key}
			\mathop{\mathbb{V}\mathrm{ar}}\limits_{\boldsymbol{\theta}}\partial_{\theta_{q, n}}{C}\geq \frac{1}{4e L}.
			\end{equation}
			\item[$(\mathrm{ii})$ ] For $n\in \left\lbrace 2, 3, \cdots, N-1 \right\rbrace  $,
			\begin{equation}\label{key}
			f_{n}\bm{(}L, a\bm{)}=2\alpha^{4L-1}\beta+2\gamma^{4L}\alpha^{2L-1}\beta.
			\end{equation}
			Furthermore, if the hyperparameter $a$ is selected according to Eq.~\eqref{TH1a} in Lemma~\ref{TH1}, where $a\in \mathrm{\Theta}\left(\frac{1}{\sqrt{L}}\right)$, then we have
			\begin{equation}\label{TH38}
			\mathop{\mathbb{V}\mathrm{ar}}\limits_{\boldsymbol{\theta}}\partial_{\theta_{q, n}}{C}\geq \frac{1}{eL}.
			\end{equation}
		\end{itemize}
	\end{lemma}

The proofs of Lemma~\ref{TH3} and Theorem~\ref{TH4} are provided in Appendix~\ref{proofTH3} and Appendix~\ref{proofTH4}, respectively.

\section{\label{experiments}Numerical Simulations}
In this section, we demonstrate the efficacy of the reduced-domain initialization strategy in enhancing trainability and convergence via VQE and QNN. Unless otherwise stated, we conduct classical simulations of demonstrations using the analytical form of the cost function and its gradient, simulated through tensor operations in TensorFlow \cite{tensorflow2016} for VQE tasks and in PyTorch \cite{pytorch2019} for QNN tasks. Additionally, we utilize TensorCircuit \cite{tensorcircuit2023}, which supports just-in-time compilation to accelerate the simulation processes of VQE. 

\subsection{\label{experimentVQE}VQE}
For VQE tasks, we focus on the $N$-qubit Heisenberg model
	\begin{equation}\label{Heisenberg}
	H=\sum_{i=1}^{N-1}{(X_{i}X_{i+1}+Y_{i}Y_{i+1}+Z_{i}Z_{i+1})},
	\end{equation}
where the locality  $S=2$. It becomes challenging to obtain the ground energy of $H$  as $N$ increases.

\begin{table*}
    \caption{\label{tab:grad}Gradient magnitudes of the $12$-qubit Heisenberg model at the initial step across different initializations. }
	\begin{ruledtabular}
	\begin{tabular}{cccccccc}
    & reduced-domain-$0.07$ & reduced-domain-$0.007$ & reduced-domain-$0.7$ & zero/$\pi$-initialization & uniform & Floquet  & Gaussian \\
	\hline
	HEA& 502.8      & 35.3    & 3.6 & 0.0/ N/A & 3.5   & 236.1      & 266.9    \\
    HVA& 4803.0    & 277.6   & 61.9 & 0.0/0.00037549 & 61.4   & 1322.1    & 1947.4  \\
	\end{tabular}
	\end{ruledtabular}
\end{table*}

Unless otherwise stated, we adopt the HEA  illustrated in Fig.~\ref{fig:setup}(a) and fix the entanglement layer in each block arranged in the nearest-neighbor pairs topology with open boundary as shown in Fig.~\ref{fig:setup}(c). For trainability and convergence, we compare {\color{blue}five} different  initialization strategies, including zero-initialization by setting all parameters to be zero, uniform $\mathcal{U}\left[ -\pi,\pi\right]$, Floquet $\mathcal{U}\left[ -W,W\right]$ with $W=0.4$ for intermediate-depth circuits and $W=0.2$ for deep circuits \cite{cao2024exploiting} \footnote{In Ref.~\cite{cao2024exploiting}, the Floquet-initialized PQC incorporates $X$- and $Y$-type rotation gates with angle uniformly sampled from $\left[ -W,W\right]$, and $Z$-type rotation gates with angle uniformly sampled from $\left[ -\pi,\pi\right]$.  In our work, the $W$s employed are consistent with those utilized in Ref.~\cite{cao2024exploiting}, given the similar settings. Specifically, in the HEA featuring only $X$- and $Y$-type rotation gates (see Fig.~\ref{fig:setup}(a)), all the parameters are initialized from $\mathcal{U}\left[ -W, W \right]$. In the Hamiltonian variational ansatz including all three types of gates (see Fig.~\ref{fig:hvapqc}), the parameters are initialized accordingly.}, Gaussian $\mathcal{N}\left(0, \sigma^{2} \right)$ with the variance $\sigma^{2}=\frac{1}{8SL}=\frac{1}{16L}$ \cite{Zhang2022g}, and our reduced-domain  $\mathcal{U}\left[ -a\pi,a\pi\right]$ with the hyperparameter $a$ selected according to Eq.~\eqref{TH2a}. Recall that after the initialization, the variational parameters are updated by the classical optimizer to minimize the cost function, and are not constrained to the initial domains.

\begin{figure}[htp]
    \centering
    \includegraphics[height=6cm]{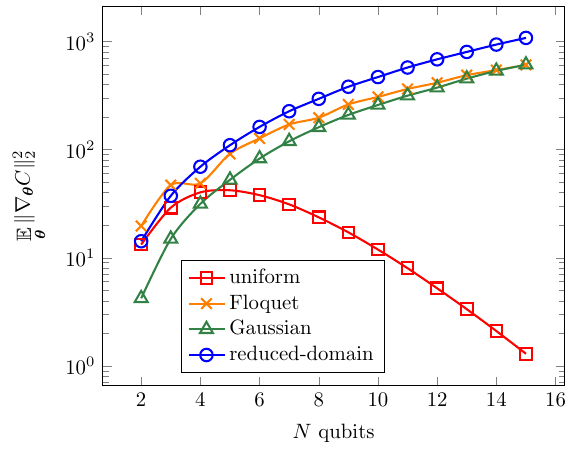}
    \caption{\label{fig:decay}Gradient magnitudes of the Heisenberg model at the initial step versus the number of qubits under  five initialization strategies: zero-initialization (not displayed where the gradient values are consistently zero), uniform (square), Floquet (cross), Gaussian (triangle), and reduced-domain (circle). Here, $L=5N$, {\color{blue} $W=0.2$}, $\sigma^2=\frac{1}{8SL}=\frac{1}{80N}$, and $a$ is determined  according to Eq.~\eqref{TH2a}. There exists a BP under the uniform initialization, whereas the other three strategies help PQCs avoid BPs.  Under the three advanced initialization strategies, the gradient norms increase exponentially with the qubit number,  among which the reduced-domain  exhibits the most favorable  growth.}
\end{figure}

In Fig.~\ref{fig:decay}, we compare the trainability of VQE under the five initialization strategies by evaluating the mean squared norm of the gradient at the first step. Here,  we set  $L=5N$ so that the  HEA in Fig.~\ref{fig:setup}(a) includes the same number of rotation layers $10N$  as that in Figure 3 of Ref.~\cite{McClean2018}. We plot the average points over 1000 rounds of demonstrations on a semi-log plot. The gradient values under zero-initialization are always zero and we do not display them. From Fig.~\ref{fig:decay}, under the uniform initialization, the magnitude of the cost gradient decays exponentially as the number of qubits $N$ increases, which coincides with both the theoretical and empirical results concerning BPs in Ref.~\cite{McClean2018}. In contrast, under the other three advanced initializations, the magnitudes of the gradient increase exponentially with the number of qubits. It is clear that the reduced-domain exhibits the most favorable growth among the five initialization strategies.

\begin{figure*}[tp]
	\centering
	\includegraphics[width=\textwidth]{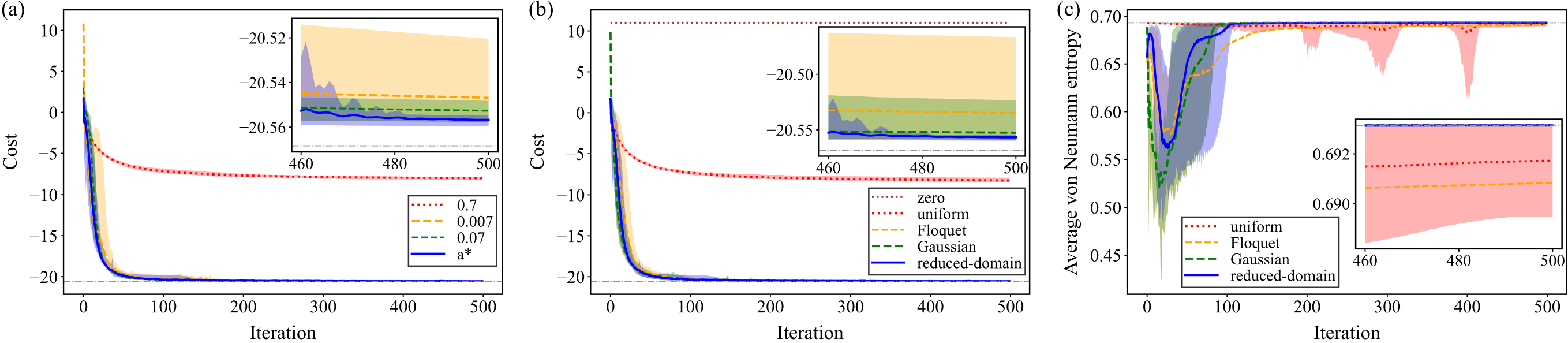}
 \caption{\label{fig:heg}  Convergence behavior of the $12$-qubit Heisenberg model under
different initializations in the HEA context.~The dash-dotted grey line corresponds to that of the global minimum. Here,  $L=3.5N=42$.~(a) Convergence of the cost values under different  hyperparameters of the reduced-domain initialization:  $a^{*}\approx0.0736328$ (solid blue), $a=0.7$ (dotted red), $a=0.07$ (thin dashed green), and $a=0.007$ (dashed orange). (b) Convergence of the cost values under different initializations. (c) Convergence  of the average von Neumann entropy under different initializations. Here, zero-initialization: dotted dark red, uniform: dotted red, Floquet: thin dashed orange, Gaussian: dashed green, and reduced-domain: solid blue. The corresponding hyperparameters are $W=0.4$, $\sigma^2=\frac{1}{8SL}=\frac{1}{672}\approx 0.001488$, and  $a^{*}=0.0736328$, respectively.
  }
\end{figure*}

Next, we demonstrate the effect of initialization  strategies on the convergence behavior of finding the ground state of the $12$-qubit Heisenberg model Eq.~(\ref{Heisenberg}).  In our demonstrations, VQEs are optimized by Adam \cite{Kingma2014} with a learning rate multiplier schedule, which has been shown to substantially improve Adam's performance \cite{Loshchilov2019}. Specifically, we utilize a linearly decreasing learning rate $\eta$ ranging from $0.1$ to $0.05$ during the $500$ training iterations. For each initialization strategy, we sample $1000$ sets of initial parameters to evaluate the initial gradient, and sample $10$ sets to track the optimization process, presenting the average performance.  The filled area in the figures of the rest of this subsection represents the smallest area that encompasses the behaviors of these 10 samples. To ensure the expressibility of the HEA in Fig.~\ref{fig:setup}(a), we set $L = 3.5N = 42$. Accordingly, the hyperparameter $a^{*}$  selected according to Eq.~\eqref{TH2a} is $0.0736328$.

We first conduct a sensitivity analysis on the hyperparameter $a$ by comparing the performance under three different values of $a$. They are in  different orders of magnitudes from $a^{*}=0.0736328$, namely, $a=0.007$, $a=0.07$, and $a=0.7$. We list the mean squared norm of the gradient at the first step for these different values of $a$ in Table~\ref{tab:grad}. We find that $a=0.07$, having the same order with $a^{*}$,  provides the most favorable trainability. We illustrate the convergence behavior of cost values
under different values of $a$ in Fig.~\ref{fig:heg}(a).  We find that  the soft reduced-domain with $a=0.7$ does not outperform the uniform initialization (see Fig.~\ref{fig:heg}(b)), while the hard one with $a=0.007$ is less effective than the appropriate one with $a=0.07$. Among the four reduced-domain initializations, the one with $a^{*}$ has the best convergence behavior, which  supports the validity of our theoretical results.

We then compare the reduced-domain strategy with $a^*$ against the other four strategies. We list the mean squared norm of the gradient at the first step of the other four initializations  in Table~\ref{tab:grad}.  We find that the magnitude of the initial gradient under the reduced-domain with $a=0.07$ is significantly larger than those under the other  initializations.
We depict their convergence behaviors of cost values in Fig.~\ref{fig:heg}(b). It is clear that the zero-initialization fails to optimize due to zero gradients. Under the  uniform initialization, the optimization quickly gets stuck  far away from the global minimum. In contrast, under the reduced-domain, Floquet, and Gaussian, large gradients at the initial stage can swiftly steer the cost to the vicinity of the global minimum. At the end (see the inset of Fig.~\ref{fig:heg}(b)), the reduced-domain strategy is most robust across random parametrizations and achieves the most favorable convergence value, followed by the Gaussian, and then the Floquet.

We also track the dynamics of quantum correlations throughout the optimization process by calculating the average von Neumann entropy, defined as
\begin{equation*}
\begin{aligned}
S\left( \rho \right)  \triangleq -\frac{1}{N} \sum_{i=1}^{N} \operatorname{Tr} \left[ \rho^i \log(\rho^i) \right],
\end{aligned}
\end{equation*}
where $\rho^i= \operatorname{Tr}_{\overline{i}}\left[\rho\right]$ denotes the reduced state of $\rho$ on the $i$th qubit, obtained by tracing out all the other qubits $\overline{i}$. It helps to evaluate the efficiency of  the PQC in capturing the entanglement of the target ground state. As depicted in Fig.~\ref{fig:heg}(c), due to the imbalance between the expressibility and trainability, the primitive uniform initialization has only limited capability to adjust quantum correlations, and does not allow the optimizer to explore the cost landscape far away from the starting point. Under zero-initialization, the von Neumann entropy remains zero (not plotted), which is consistent with the cost values remaining unchanged in Fig.~\ref{fig:heg}(b). In contrast, the other three advanced initialization strategies exhibit capability for adjusting entanglement, even in the hardware-efficient circuits. Although these advanced  strategies restrict the initial domain, they effectively allow the optimizer to explore a broad range of the cost landscape and approach the global optimum. Notably, the reduced-domain and Gaussian initializations are more effective in achieving quantum correlations in the target ground state, as indicated by the metric of average von Neumann entropy. The efficiency of the Floquet one can be improved through a careful selection of initial trial states, whose entanglement structures closely align with that of the target ground state \cite{cao2024exploiting}. 

\begin{figure}
	\centering
    \includegraphics[width=0.4\textwidth]{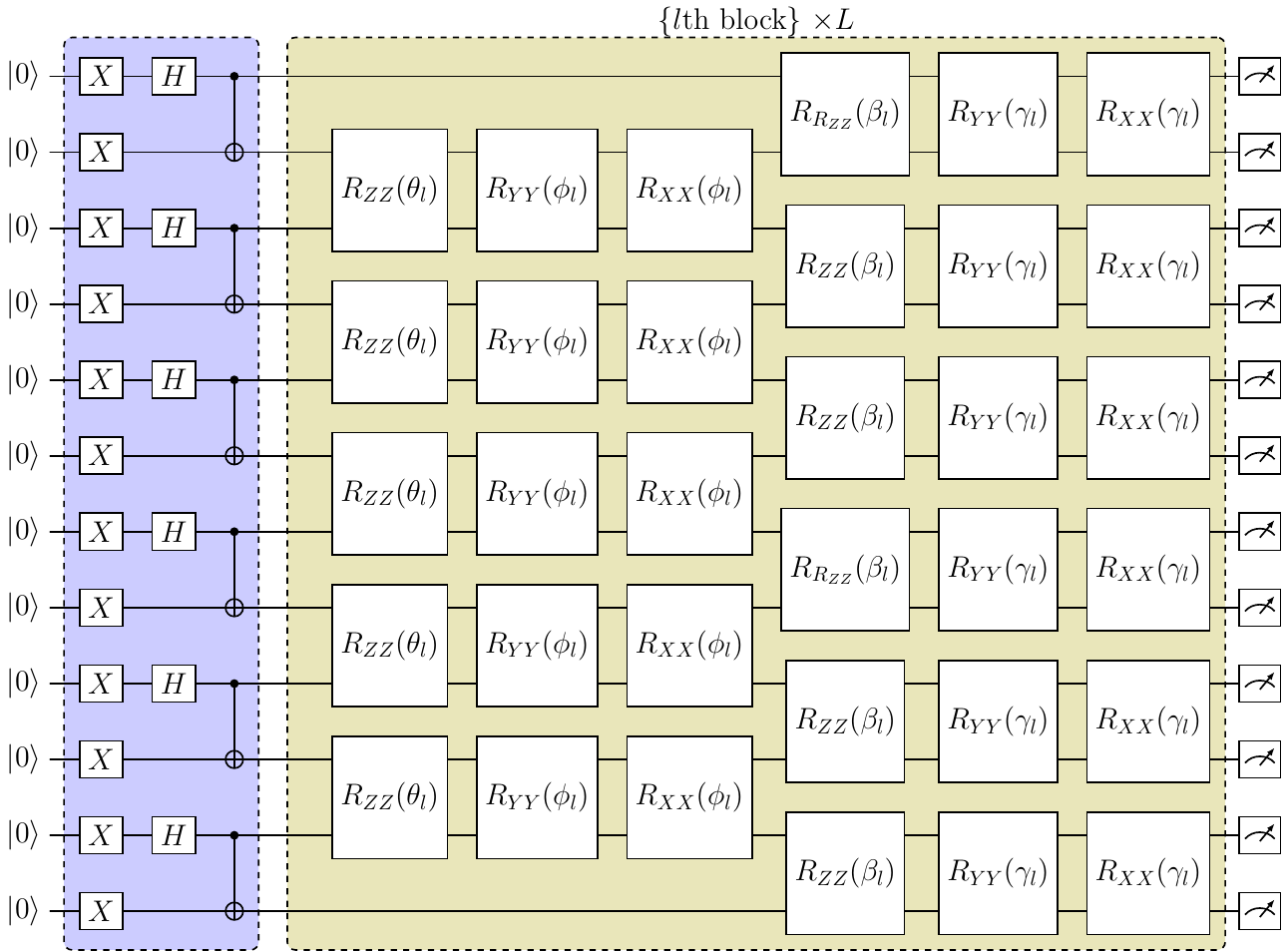}
    \caption{\label{fig:hvapqc} Setup of $L$-block HVA for the 12-qubit Heisenberg model. Here, $L=2N=24$. The purple shaded consists of Pauli-$X$ gates, Hadamard gates and CX gates. It denotes the preparation of the initial reference state $\otimes^{6}|\Psi^{-}\rangle$ with $|\Psi^{-}\rangle=(1/\sqrt{2})(01\rangle-|10\rangle)$. The two-qubit parameterized rotation gates are $R_{\Gamma\Gamma}\left(x\right)=\exp{\left\lbrace-i\frac{x}{2}\Gamma\otimes\Gamma\right\rbrace}$ with $\Gamma\Gamma\in\{XX,~YY,~ZZ\}$, and  $x\in\{\theta_{l}, \phi_{l}, \beta_{l}, \gamma_{l}\}$.
}
  \end{figure}


\begin{figure*}[htp]
	\centering
	\includegraphics[width=0.67\textwidth]{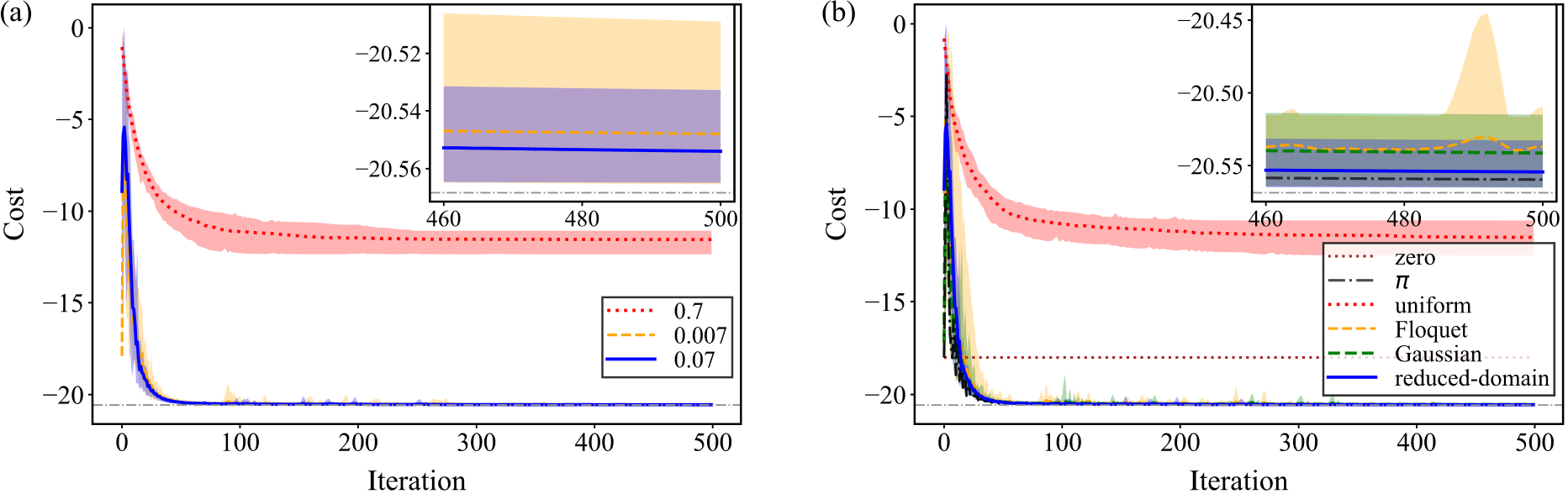}
 \caption{\label{fig:hva}Convergence behavior of the $12$-qubit Heisenberg model under  different initializations in the HVA context.~The dash-dotted grey line corresponds to that of the global minimum. Here,  $L=2N=24$.~(a) Convergence of the cost values under different  hyperparameters of the reduced-domain initialization $\mathcal{U}\left[-a\pi, a\pi\right]$:  $a=0.07$ (solid blue), $a=0.7$ (dotted red), and $a=0.007$ (dashed orange). (b) Convergence of the cost values under different initializations.  Here, zero-initialization: dotted dark red, $\pi$-initialization: dash dotted black, uniform: dotted red, Floquet: thin dashed orange, Gaussian: dashed green, and reduced-domain: solid blue. The corresponding hyperparameters are $W=0.4$, $\sigma^2= 0.001488$, and  $a=0.07$, respectively.}
\end{figure*}

We further investigate the effectiveness of the advanced initialization strategies on a problem-specific ansatz known as the Hamiltonian variational ansatz (HVA) \cite{Wiersema2020}.  In contrast to the HEA, the PQC of HVA is designed on the basis of the Hamiltonian of interest. For the 12-qubit Heisenberg model, we illustrate the PQC in Fig.~\ref{fig:hvapqc}. Note that the initial reference state plays a crucial role in HVA. In Fig.~\ref{fig:hvapqc}, the initial reference state is $\otimes^{6}|\Psi^{-}\rangle$ with $|\Psi^{-}\rangle=(1/\sqrt{2})(01\rangle-|10\rangle)$.  It is suggested in Ref.~\cite{Wiersema2020} that a suitable initialization strategy is crucial for circumventing BPs and successfully finding the ground state. In particular, they employed a simpler version of identity-block initialization than that in Ref.~\cite{Grant2019} by setting all variational parameters to be $\pi$, which we refer to as $\pi$-initialization in our paper. It enables accurate approximations to the ground state with rapid convergence.


In the context of HVA, we set $L=2N=24$, resulting in $6N(N-1)$ two-qubit parameterized rotation gates, approximately equal to the total number of $7N^2$ single-qubit parameterized rotation gates in the previous HEA; see Fig.~\ref{fig:setup}(a).  In view of the comparable number of parameterized gates, we set the hyperparameter  $a$ of the reduced-domain initialization for HVA to be $0.07$.
This choice is numerically supported by the sensitivity analysis on $a$ concerning the initial trainability and the final convergence behavior of cost values illustrated in Table~\ref{tab:grad} and Fig.~\ref{fig:hva}(a), respectively. In addition, we assess the five initialization strategies as well as the $\pi$-initialization employed in Ref.~\cite{Wiersema2020}, by examining their initial gradient and cost convergence, as shown in Table~\ref{tab:grad} and Fig.~\ref{fig:hva}(b), respectively. Similar to the numerical findings in the HEA context, the zero-initialization, uniform and soft reduced-domain ($a=0.7$) strategies  suffer from trainability issues, whereas the other initialization strategies
can help approximate the ground state, among which the appropriate reduced-domain strategy  ($a=0.07$) demonstrates superior performance. Note that the $\pi$-initialization strategy employed in Ref.~\cite{Wiersema2020} is a deterministic parameter initialization that exhibits a fixed minor gradient yet achieves better convergence performance than the average performance of the other random initializations. However, the performance of deterministic initializations is highly dependent on the Hamiltonian model and the choice of the reference state. Specifically, although the HVA reduces to the identity when all parameters are set to either $0$ or $\pi$, their performances are significantly different due to the gradient differences in their landscapes. In practical VQE tasks, deterministic initializations may not be as effective as random ones. Random samples can potentially achieve lower cost values, as illustrated by the blue shaded area in the inset of Fig.~\ref{fig:hva}(b).

Although the reduced-domain initialization is effective in addressing trainability issues, there is still a gap between the cost values and the global minimum, which is around 0.005 at best in our demonstrations. The role of reduced-domain is to balance the conflict between trainability and expressibility of PQCs. Recall that for VQAs, a key factor is designing an appropriate ansatz that is either hardware-efficient, problem-inspired, or in a learning way.Thus, to further improve the performance of VQE, we may integrate the reduced-domain strategy with an effective ansatz for PQC, e.g., the entanglement-variational hardware-efficient ansatz~\cite{wang2024eha}.

To take into account of the impact of statistical noise owing to limited measurement shots, for the $10$-qubit Heisenberg model, we consider four cases where the numbers of measurement counts per evaluation are 50, 500, 5000, and infinitely many. Here, the circuit depth is set to be $L=3.5N=35$ and the VQEs are optimized using Adam with a fixed learning rate of $0.05$ during the $200$ training iterations. These numerical simulations are conducted using the PennyLane’s built-in "default.qubit" device \cite{pennylane2018} with a finite number of shots to measure observables.

\begin{figure}[h]
        \centering
	\includegraphics[width=0.34\textwidth]{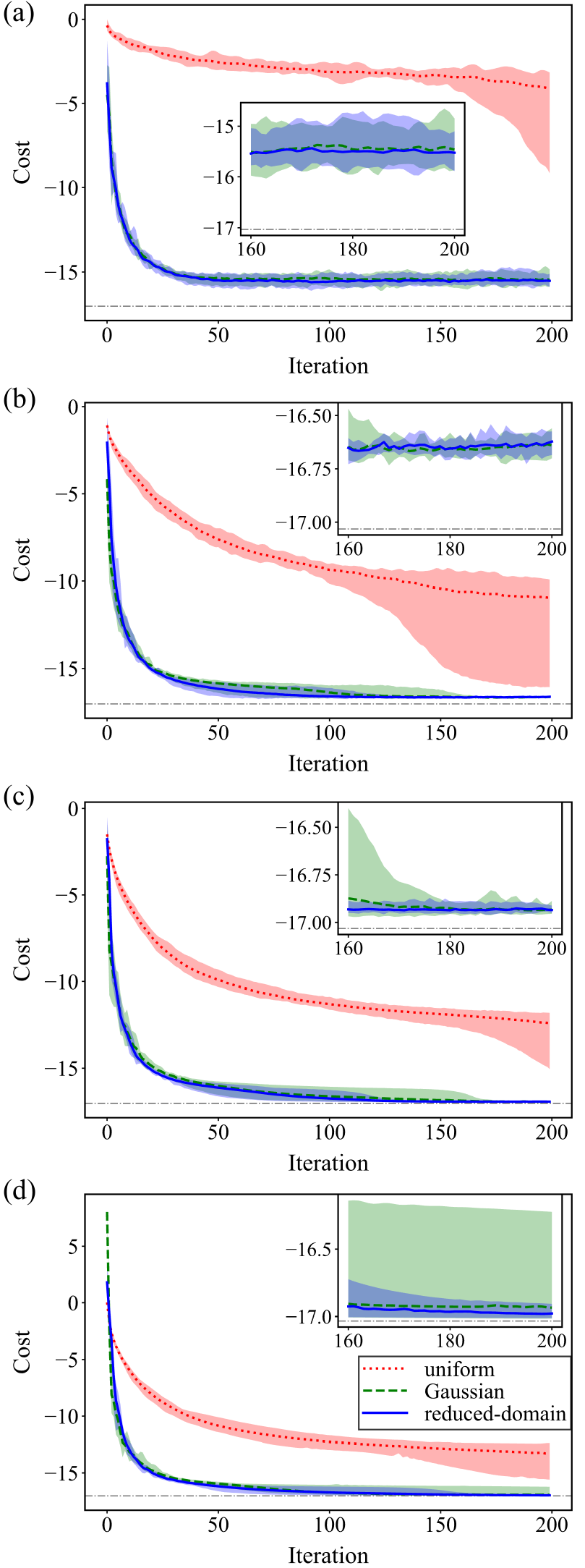}
\caption{\label{fig:finite} Convergence behavior of the $10$-qubit Heisenberg model under uniform (dotted red), Gaussian (dashed green), and reduced-domain (solid blue) initialization, with different numbers of measurement shots per evaluation:~(a) 50 shots; (b) 500 shots; (c) 5000 shots; (d) infinitely many.  Here, $S=2$,~$L=35$, the variance of Gaussian initialization is chosen as $\frac{1}{8SL}=\frac{1}{560}\approx 0.0017857$, while the hyperparameter $a$ of the reduced-domain  is selected to be $0.0806522$ according to Eq.~\eqref{TH2a}. }
\end{figure}

As illustrated in Fig.~\ref{fig:finite}, regardless of the number of measurement shots, compared to the uniform initialization, the performances under our reduced-domain and Gaussian can be greatly enhanced. Moreover, from the insets of  Fig.~\ref{fig:finite}, the performance of our reduced-domain is always more robust with respect to different realizations as compared with Gaussian.  From Fig.~\ref{fig:finite}(c) and Fig.~\ref{fig:finite}(d), we can see that for Gaussian initialization, the variance of the cost in the ideal case is larger than that with 5000-shots. This is because  the statistical noise may play a positive role. As verified in Section 9.5 in \cite{wright2022high}, gradient descent with small random noise can  implicitly compute the second-order information and exploit the direction of the negative curvature to achieve adequate local descent.

\subsection{QNN}
\begin{figure*}[htp]
    \centering
    \includegraphics[width=\textwidth]{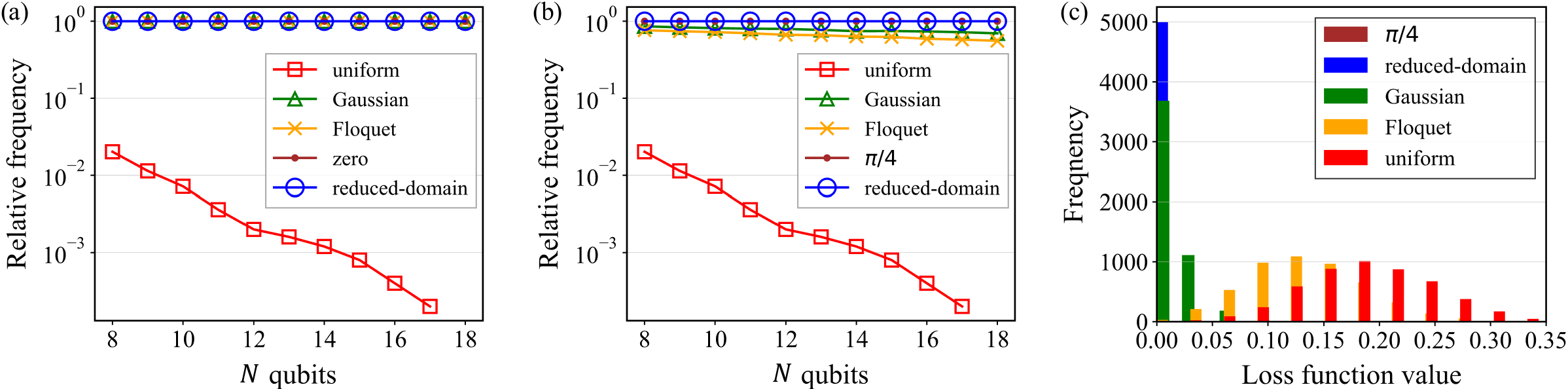}
    \caption{\label{fig:QNNdecay} Relative frequencies of achieving global convergence in 5000 rounds of demonstrations under five initializations: uniform, Gaussian,  Floquet, zero/$\frac{\pi}{4}$-initialization (parameters are set to be  zero/$\frac{\pi}{4}$), and reduced-domain. Here, $\sigma^2=0.008225$, $W=0.2$, and $a=0.1$. (a)  The center of the parameters' domain is zero for the five initializations. The probability of success decays exponentially with the number of qubits under uniform initialization, while it is always 1 under the other four strategies. (b) The center of the parameters' domain is $\frac{\pi}{4}$ for the five initializations. Both the Floquet and Gaussian initializations exhibit a slightly exponentially decreasing probability of success, while the reduced-domain and $\frac{\pi}{4}$-initialization strategies consistently succeed. (c) Detailed distribution of the loss function value at convergence for the $16$-qubit instance. The cumulative frequency of global convergence is 5000 for both $\frac{\pi}{4}$-initialization and reduced-domain ($\mathcal{U}\left[ \frac{\pi}{4}-a\frac{\pi}{2}, \frac{\pi}{4}+a\frac{\pi}{2}\right]$ with $a=0.1$), 3602 for Gaussian $(\mathcal{N}\left(\frac{\pi}{4},0.008225 \right) $), 2887 for Floquet ($\mathcal{U}\left[\frac{\pi}{4}-0.2, \frac{\pi}{4}+0.2\right])$, and only 2 for uniform ($\mathcal{U}\left[ \frac{\pi}{4}-\frac{\pi}{2}, \frac{\pi}{4}+\frac{\pi}{2}\right]$).}
\end{figure*}

We now extend the reduced-domain initialization strategy to QNN. In Proposition 5 of Ref.~\cite{You2021}, it was demonstrated that in the under-parameterized region, there exists such a dataset that can induce a loss function with exponentially many spurious local minima.  Subsequently, several simple yet extremely hard QNN instances for training were presented, where the initial parameters are uniformly drawn from the entire period. Here we focus on Example 2 in Ref.~\cite{You2021} and demonstrate that the reduced-domain parameter initialization strategy can help achieve global convergence efficiently.
	
Let us first briefly introduce the hard instance for training in Ref.~\cite{You2021}. In the under-parameterized region with depth $L=1$, constructing a $p$-parameter hard instance involves $N=p$ qubits with the loss function
	\begin{equation}\label{loss}
	\frac{1}{2p}\sum_{l=1}^{p}\left[ {\left( \sin{2\theta_{l}}-\sin{\frac{\pi}{50}}\right) }^{2} +\frac{1}{4}{\left( \cos{2\theta_{l}}-\cos{\frac{\pi}{50}}\right) }^{2}\right].
	\end{equation}
Here, the $p$-dimensional parameter vector is denoted by $\boldsymbol{\theta}=(\theta_{1}, \theta_{2}, \cdots, \theta_{p})$. It was proved that within each period, for instance, $[-\frac{\pi}{2}, \frac{\pi}{2}]$, there are $2^{p}$ minima. Among them, the global minimum is attained at $\boldsymbol{\theta}^{\star}={\left(\frac{\pi}{100},\dots,\frac{\pi}{100} \right)} $, while the other $2^{p}-1$ local minima are spurious with non-negligible suboptimality gaps.

We follow the same routine as outlined in Ref.~\cite{You2021}, where QNNs are trained  using back-propagating gradients with RMSProp. The constant learning rate is set to be 0.01, and the smoothing constant for mean-squared estimation is 0.99. We perform 5000 rounds of demonstrations with 200 iterations for each round.  We consider the convergence of the global minimum within a margin of error of $10^{-4}$.

We compare the relative frequencies of achieving global convergence under  five initialization strategies: zero-initialization by setting all parameters to be zero, uniform $\mathcal{U}\left[ -\frac{\pi}{2}, \frac{\pi}{2}\right]$ as in Ref.~\cite{You2021}, Floquet $\mathcal{U}\left[ -W,W\right]$ with $W=0.2$,  reduced-domain  $\mathcal{U}\left[ -a\frac{\pi}{2}, a\frac{\pi}{2}\right]$ with  $a=0.1$, and Gaussian $\mathcal{N}\left(0,0.008225\right) $, whose variance is set to be the same as the reduced-domain strategy. From the semi-log plot in Fig.~\ref{fig:QNNdecay}(a), it is clear that under uniform initialization, the probability of success decays exponentially with the number of qubits (or parameters). In particular, in the $18$-qubit case, there is no round achieving the global minimum within $5000$ rounds of demonstrations. In contrast, all the other advanced strategies consistently achieve global convergence.


To further compare the efficiency of these advanced initialization strategies, we  shift the parameter initialization domain away from the global minimum. Specifically, we adjust the center of the above five initialization strategies to $\frac{\pi}{4}$, instead of $0$. From the semi-log plot in Fig.~\ref{fig:QNNdecay}(b), both the Gaussian and Floquet initializations exhibit exponential growth of failure in achieving global convergence as the system size increases, although the rates of their exponential growth are significantly lower than that of uniform initialization. Notably, the reduced-domain and $\frac{\pi}{4}$-initialization strategies consistently succeed in achieving global convergence in this context of shallow circuits. Especially,  for the $16$-qubit instance, the cumulative frequency of global convergence is only 2 under uniform initialization, while it is 2887, 3602, 5000, and 5000 under the Floquet, Gaussian, $\frac{\pi}{4}$-initialization, and reduced-domain strategies, respectively. Furthermore, we illustrate a detailed distribution of the loss function's convergence for the $16$-qubit case in Fig.~\ref{fig:QNNdecay}(c).

From the above demonstrations, we conclude that when the landscape of the cost function has abundance of local minima, it is a smart choice to properly reduce the initial domain to improve the trainability by limiting the expressibility of  initial PQCs.

\section{\label{conclusion}Conclusion \lowercase{and} Discussion}
In this paper, we have presented an efficient reduced-domain parameter initialization, enhancing trainability of VQAs both theoretically and empirically. For Hamiltonian represented as a sum of local Pauli strings, we have proved that with the reduced-domain strategy, the average of the squared norm of the cost gradient decays at most polynomially with the depth of circuit and increases linearly with the number of individual terms in the summation of Hamiltonian. Thus, it clearly implies that reduced-domain initialization helps PQC avoid the phenomenon of BPs. Our empirical results have validated the theoretical results in VQE tasks. Moreover, via simulations, we have demonstrated that the reduced-domain strategy can also help address the poor convergence in presence of BPs in the VQE context of both HEA and HVA, and exponentially many spurious local minima in under-parameterized QNNs. The reduced-domain initialization strategy has the potential to significantly enhance trainability and convergence of PQCs and generate wide impact on developing algorithms with quantum advantages for various practical applications.

\begin{acknowledgments}
B.~Q. acknowledges the support of the National Natural Science Foundation of China (No.~61773370 and No.~61833010), and D.~D. acknowledges the support of the Australian Research Council's Future Fellowship funding scheme under project FT220100656.
\end{acknowledgments}

\bibliography{ref}

\begin{thebibliography}{79}%
\makeatletter
\providecommand \@ifxundefined [1]{%
 \@ifx{#1\undefined}
}%
\providecommand \@ifnum [1]{%
 \ifnum #1\expandafter \@firstoftwo
 \else \expandafter \@secondoftwo
 \fi
}%
\providecommand \@ifx [1]{%
 \ifx #1\expandafter \@firstoftwo
 \else \expandafter \@secondoftwo
 \fi
}%
\providecommand \natexlab [1]{#1}%
\providecommand \enquote  [1]{``#1''}%
\providecommand \bibnamefont  [1]{#1}%
\providecommand \bibfnamefont [1]{#1}%
\providecommand \citenamefont [1]{#1}%
\providecommand \href@noop [0]{\@secondoftwo}%
\providecommand \href [0]{\begingroup \@sanitize@url \@href}%
\providecommand \@href[1]{\@@startlink{#1}\@@href}%
\providecommand \@@href[1]{\endgroup#1\@@endlink}%
\providecommand \@sanitize@url [0]{\catcode `\\12\catcode `\$12\catcode
  `\&12\catcode `\#12\catcode `\^12\catcode `\_12\catcode `\%12\relax}%
\providecommand \@@startlink[1]{}%
\providecommand \@@endlink[0]{}%
\providecommand \url  [0]{\begingroup\@sanitize@url \@url }%
\providecommand \@url [1]{\endgroup\@href {#1}{\urlprefix }}%
\providecommand \urlprefix  [0]{URL }%
\providecommand \Eprint [0]{\href }%
\providecommand \doibase [0]{https://doi.org/}%
\providecommand \selectlanguage [0]{\@gobble}%
\providecommand \bibinfo  [0]{\@secondoftwo}%
\providecommand \bibfield  [0]{\@secondoftwo}%
\providecommand \translation [1]{[#1]}%
\providecommand \BibitemOpen [0]{}%
\providecommand \bibitemStop [0]{}%
\providecommand \bibitemNoStop [0]{.\EOS\space}%
\providecommand \EOS [0]{\spacefactor3000\relax}%
\providecommand \BibitemShut  [1]{\csname bibitem#1\endcsname}%
\let\auto@bib@innerbib\@empty
\bibitem [{\citenamefont {McClean}\ \emph {et~al.}(2016)\citenamefont
  {McClean}, \citenamefont {Romero}, \citenamefont {Babbush},\ and\
  \citenamefont {Aspuru-Guzik}}]{McClean2016}%
  \BibitemOpen
  \bibfield  {author} {\bibinfo {author} {\bibfnamefont {J.~R.}\ \bibnamefont
  {McClean}}, \bibinfo {author} {\bibfnamefont {J.}~\bibnamefont {Romero}},
  \bibinfo {author} {\bibfnamefont {R.}~\bibnamefont {Babbush}},\ and\ \bibinfo
  {author} {\bibfnamefont {A.}~\bibnamefont {Aspuru-Guzik}},\ }\bibfield
  {title} {\bibinfo {title} {The theory of variational hybrid quantum-classical
  algorithms},\ }\href {https://doi.org/10.1088/1367-2630/18/2/023023}
  {\bibfield  {journal} {\bibinfo  {journal} {New J. Phys.}\ }\textbf {\bibinfo
  {volume} {18}},\ \bibinfo {pages} {023023} (\bibinfo {year}
  {2016})}\BibitemShut {NoStop}%
\bibitem [{\citenamefont {Cerezo}\ \emph
  {et~al.}(2021{\natexlab{a}})\citenamefont {Cerezo}, \citenamefont
  {Arrasmith}, \citenamefont {Babbush}, \citenamefont {Benjamin}, \citenamefont
  {Endo}, \citenamefont {Fujii}, \citenamefont {McClean}, \citenamefont
  {Mitarai}, \citenamefont {Yuan}, \citenamefont {Cincio} \emph
  {et~al.}}]{Cerezo2021v}%
  \BibitemOpen
  \bibfield  {author} {\bibinfo {author} {\bibfnamefont {M.}~\bibnamefont
  {Cerezo}}, \bibinfo {author} {\bibfnamefont {A.}~\bibnamefont {Arrasmith}},
  \bibinfo {author} {\bibfnamefont {R.}~\bibnamefont {Babbush}}, \bibinfo
  {author} {\bibfnamefont {S.~C.}\ \bibnamefont {Benjamin}}, \bibinfo {author}
  {\bibfnamefont {S.}~\bibnamefont {Endo}}, \bibinfo {author} {\bibfnamefont
  {K.}~\bibnamefont {Fujii}}, \bibinfo {author} {\bibfnamefont {J.~R.}\
  \bibnamefont {McClean}}, \bibinfo {author} {\bibfnamefont {K.}~\bibnamefont
  {Mitarai}}, \bibinfo {author} {\bibfnamefont {X.}~\bibnamefont {Yuan}},
  \bibinfo {author} {\bibfnamefont {L.}~\bibnamefont {Cincio}}, \emph
  {et~al.},\ }\bibfield  {title} {\bibinfo {title} {Variational quantum
  algorithms},\ }\href {https://doi.org/10.1038/s42254-021-00348-9} {\bibfield
  {journal} {\bibinfo  {journal} {Nat. Rev. Phys.}\ }\textbf {\bibinfo {volume}
  {3}},\ \bibinfo {pages} {625} (\bibinfo {year}
  {2021}{\natexlab{a}})}\BibitemShut {NoStop}%
\bibitem [{\citenamefont {Jerbi}\ \emph {et~al.}(2023)\citenamefont {Jerbi},
  \citenamefont {Fiderer}, \citenamefont {Poulsen~Nautrup}, \citenamefont
  {K{\"u}bler}, \citenamefont {Briegel},\ and\ \citenamefont
  {Dunjko}}]{Jerbi2023}%
  \BibitemOpen
  \bibfield  {author} {\bibinfo {author} {\bibfnamefont {S.}~\bibnamefont
  {Jerbi}}, \bibinfo {author} {\bibfnamefont {L.~J.}\ \bibnamefont {Fiderer}},
  \bibinfo {author} {\bibfnamefont {H.}~\bibnamefont {Poulsen~Nautrup}},
  \bibinfo {author} {\bibfnamefont {J.~M.}\ \bibnamefont {K{\"u}bler}},
  \bibinfo {author} {\bibfnamefont {H.~J.}\ \bibnamefont {Briegel}},\ and\
  \bibinfo {author} {\bibfnamefont {V.}~\bibnamefont {Dunjko}},\ }\bibfield
  {title} {\bibinfo {title} {Quantum machine learning beyond kernel methods},\
  }\href {https://doi.org/10.1038/s41467-023-36159-y} {\bibfield  {journal}
  {\bibinfo  {journal} {Nat. Commun.}\ }\textbf {\bibinfo {volume} {14}},\
  \bibinfo {pages} {517} (\bibinfo {year} {2023})}\BibitemShut {NoStop}%
\bibitem [{\citenamefont {Liu}\ \emph {et~al.}(2021)\citenamefont {Liu},
  \citenamefont {Arunachalam},\ and\ \citenamefont {Temme}}]{liu2021a}%
  \BibitemOpen
  \bibfield  {author} {\bibinfo {author} {\bibfnamefont {Y.}~\bibnamefont
  {Liu}}, \bibinfo {author} {\bibfnamefont {S.}~\bibnamefont {Arunachalam}},\
  and\ \bibinfo {author} {\bibfnamefont {K.}~\bibnamefont {Temme}},\ }\bibfield
   {title} {\bibinfo {title} {A rigorous and robust quantum speed-up in
  supervised machine learning},\ }\href
  {https://doi.org/10.1038/s41567-021-01287-z} {\bibfield  {journal} {\bibinfo
  {journal} {Nat. Phys.}\ }\textbf {\bibinfo {volume} {17}},\ \bibinfo {pages}
  {1013} (\bibinfo {year} {2021})}\BibitemShut {NoStop}%
\bibitem [{\citenamefont {Benedetti}\ \emph {et~al.}(2019)\citenamefont
  {Benedetti}, \citenamefont {Lloyd}, \citenamefont {Sack},\ and\ \citenamefont
  {Fiorentini}}]{Benedetti2019}%
  \BibitemOpen
  \bibfield  {author} {\bibinfo {author} {\bibfnamefont {M.}~\bibnamefont
  {Benedetti}}, \bibinfo {author} {\bibfnamefont {E.}~\bibnamefont {Lloyd}},
  \bibinfo {author} {\bibfnamefont {S.}~\bibnamefont {Sack}},\ and\ \bibinfo
  {author} {\bibfnamefont {M.}~\bibnamefont {Fiorentini}},\ }\bibfield  {title}
  {\bibinfo {title} {Parameterized quantum circuits as machine learning
  models},\ }\href {https://doi.org/10.1088/2058-9565/ab4eb5} {\bibfield
  {journal} {\bibinfo  {journal} {Quantum Sci. Technol.}\ }\textbf {\bibinfo
  {volume} {4}},\ \bibinfo {pages} {043001} (\bibinfo {year}
  {2019})}\BibitemShut {NoStop}%
\bibitem [{\citenamefont {Sweke}\ \emph {et~al.}(2020)\citenamefont {Sweke},
  \citenamefont {Wilde}, \citenamefont {Meyer}, \citenamefont {Schuld},
  \citenamefont {F{\"a}hrmann}, \citenamefont {Meynard-Piganeau},\ and\
  \citenamefont {Eisert}}]{Sweke2020}%
  \BibitemOpen
  \bibfield  {author} {\bibinfo {author} {\bibfnamefont {R.}~\bibnamefont
  {Sweke}}, \bibinfo {author} {\bibfnamefont {F.}~\bibnamefont {Wilde}},
  \bibinfo {author} {\bibfnamefont {J.}~\bibnamefont {Meyer}}, \bibinfo
  {author} {\bibfnamefont {M.}~\bibnamefont {Schuld}}, \bibinfo {author}
  {\bibfnamefont {P.~K.}\ \bibnamefont {F{\"a}hrmann}}, \bibinfo {author}
  {\bibfnamefont {B.}~\bibnamefont {Meynard-Piganeau}},\ and\ \bibinfo {author}
  {\bibfnamefont {J.}~\bibnamefont {Eisert}},\ }\bibfield  {title} {\bibinfo
  {title} {Stochastic gradient descent for hybrid quantum-classical
  optimization},\ }\href {https://doi.org/10.22331/q-2020-08-31-314} {\bibfield
   {journal} {\bibinfo  {journal} {Quantum}\ }\textbf {\bibinfo {volume} {4}},\
  \bibinfo {pages} {314} (\bibinfo {year} {2020})}\BibitemShut {NoStop}%
\bibitem [{\citenamefont {Dong}\ and\ \citenamefont
  {Petersen}(2022)}]{Dong2022}%
  \BibitemOpen
  \bibfield  {author} {\bibinfo {author} {\bibfnamefont {D.}~\bibnamefont
  {Dong}}\ and\ \bibinfo {author} {\bibfnamefont {I.~R.}\ \bibnamefont
  {Petersen}},\ }\bibfield  {title} {\bibinfo {title} {Quantum estimation,
  control and learning: Opportunities and challenges},\ }\href
  {https://doi.org/10.1016/j.arcontrol.2022.04.011} {\bibfield  {journal}
  {\bibinfo  {journal} {Annual Reviews in Control}\ }\textbf {\bibinfo {volume}
  {54}},\ \bibinfo {pages} {243} (\bibinfo {year} {2022})}\BibitemShut
  {NoStop}%
\bibitem [{\citenamefont {Basheer}\ \emph {et~al.}(2022)\citenamefont
  {Basheer}, \citenamefont {Feng}, \citenamefont {Ferrie},\ and\ \citenamefont
  {Li}}]{Basheer2022}%
  \BibitemOpen
  \bibfield  {author} {\bibinfo {author} {\bibfnamefont {A.}~\bibnamefont
  {Basheer}}, \bibinfo {author} {\bibfnamefont {Y.}~\bibnamefont {Feng}},
  \bibinfo {author} {\bibfnamefont {C.}~\bibnamefont {Ferrie}},\ and\ \bibinfo
  {author} {\bibfnamefont {S.}~\bibnamefont {Li}},\ }\bibfield  {title}
  {\bibinfo {title} {Alternating layered variational quantum circuits can be
  classically optimized efficiently using classical shadows},\ }\href
  {https://doi.org/10.48550/arXiv.2208.11623} {\bibfield  {journal} {\bibinfo
  {journal} {arXiv preprint arXiv:2208.11623}\ } (\bibinfo {year}
  {2022})}\BibitemShut {NoStop}%
\bibitem [{\citenamefont {Lavrijsen}\ \emph {et~al.}(2020)\citenamefont
  {Lavrijsen}, \citenamefont {Tudor}, \citenamefont {M{\"u}ller}, \citenamefont
  {Iancu},\ and\ \citenamefont {De~Jong}}]{Lavrijsen2020}%
  \BibitemOpen
  \bibfield  {author} {\bibinfo {author} {\bibfnamefont {W.}~\bibnamefont
  {Lavrijsen}}, \bibinfo {author} {\bibfnamefont {A.}~\bibnamefont {Tudor}},
  \bibinfo {author} {\bibfnamefont {J.}~\bibnamefont {M{\"u}ller}}, \bibinfo
  {author} {\bibfnamefont {C.}~\bibnamefont {Iancu}},\ and\ \bibinfo {author}
  {\bibfnamefont {W.}~\bibnamefont {De~Jong}},\ }\bibfield  {title} {\bibinfo
  {title} {Classical optimizers for noisy intermediate-scale quantum devices},\
  }in\ \href {https://ieeexplore.ieee.org/document/9259985/} {\emph {\bibinfo
  {booktitle} {2020 IEEE International Conference on Quantum Computing and
  Engineering (QCE)}}}\ (\bibinfo {year} {2020})\ pp.\ \bibinfo {pages}
  {267--277}\BibitemShut {NoStop}%
\bibitem [{\citenamefont {Biamonte}\ \emph {et~al.}(2017)\citenamefont
  {Biamonte}, \citenamefont {Wittek}, \citenamefont {Pancotti}, \citenamefont
  {Rebentrost}, \citenamefont {Wiebe},\ and\ \citenamefont
  {Lloyd}}]{Biamonte2017}%
  \BibitemOpen
  \bibfield  {author} {\bibinfo {author} {\bibfnamefont {J.}~\bibnamefont
  {Biamonte}}, \bibinfo {author} {\bibfnamefont {P.}~\bibnamefont {Wittek}},
  \bibinfo {author} {\bibfnamefont {N.}~\bibnamefont {Pancotti}}, \bibinfo
  {author} {\bibfnamefont {P.}~\bibnamefont {Rebentrost}}, \bibinfo {author}
  {\bibfnamefont {N.}~\bibnamefont {Wiebe}},\ and\ \bibinfo {author}
  {\bibfnamefont {S.}~\bibnamefont {Lloyd}},\ }\bibfield  {title} {\bibinfo
  {title} {Quantum machine learning},\ }\href
  {https://doi.org/10.1038/nature23474} {\bibfield  {journal} {\bibinfo
  {journal} {Nature}\ }\textbf {\bibinfo {volume} {549}},\ \bibinfo {pages}
  {195} (\bibinfo {year} {2017})}\BibitemShut {NoStop}%
\bibitem [{\citenamefont {Mitarai}\ \emph {et~al.}(2018)\citenamefont
  {Mitarai}, \citenamefont {Negoro}, \citenamefont {Kitagawa},\ and\
  \citenamefont {Fujii}}]{Mitarai2018}%
  \BibitemOpen
  \bibfield  {author} {\bibinfo {author} {\bibfnamefont {K.}~\bibnamefont
  {Mitarai}}, \bibinfo {author} {\bibfnamefont {M.}~\bibnamefont {Negoro}},
  \bibinfo {author} {\bibfnamefont {M.}~\bibnamefont {Kitagawa}},\ and\
  \bibinfo {author} {\bibfnamefont {K.}~\bibnamefont {Fujii}},\ }\bibfield
  {title} {\bibinfo {title} {Quantum circuit learning},\ }\href
  {https://doi.org/10.1103/PhysRevA.98.032309} {\bibfield  {journal} {\bibinfo
  {journal} {Phys. Rev. A}\ }\textbf {\bibinfo {volume} {98}},\ \bibinfo
  {pages} {032309} (\bibinfo {year} {2018})}\BibitemShut {NoStop}%
\bibitem [{\citenamefont {Farhi}\ and\ \citenamefont
  {Neven}(2018)}]{Farhi2018}%
  \BibitemOpen
  \bibfield  {author} {\bibinfo {author} {\bibfnamefont {E.}~\bibnamefont
  {Farhi}}\ and\ \bibinfo {author} {\bibfnamefont {H.}~\bibnamefont {Neven}},\
  }\bibfield  {title} {\bibinfo {title} {Classification with quantum neural
  networks on near term processors},\ }\href
  {https://doi.org/10.48550/arXiv.1802.06002} {\bibfield  {journal} {\bibinfo
  {journal} {arXiv preprint arXiv:1802.06002}\ } (\bibinfo {year}
  {2018})}\BibitemShut {NoStop}%
\bibitem [{\citenamefont {Wu}\ \emph {et~al.}(2020)\citenamefont {Wu},
  \citenamefont {Cao}, \citenamefont {Xie},\ and\ \citenamefont
  {Liu}}]{Wu2020}%
  \BibitemOpen
  \bibfield  {author} {\bibinfo {author} {\bibfnamefont {R.-B.}\ \bibnamefont
  {Wu}}, \bibinfo {author} {\bibfnamefont {X.}~\bibnamefont {Cao}}, \bibinfo
  {author} {\bibfnamefont {P.}~\bibnamefont {Xie}},\ and\ \bibinfo {author}
  {\bibfnamefont {Y.-x.}\ \bibnamefont {Liu}},\ }\bibfield  {title} {\bibinfo
  {title} {End-to-end quantum machine learning implemented with controlled
  quantum dynamics},\ }\href {https://doi.org/10.1103/PhysRevApplied.14.064020}
  {\bibfield  {journal} {\bibinfo  {journal} {Phys. Rev. Appl.}\ }\textbf
  {\bibinfo {volume} {14}},\ \bibinfo {pages} {064020} (\bibinfo {year}
  {2020})}\BibitemShut {NoStop}%
\bibitem [{\citenamefont {Perrier}\ \emph {et~al.}(2020)\citenamefont
  {Perrier}, \citenamefont {Tao},\ and\ \citenamefont {Ferrie}}]{Perrier2020}%
  \BibitemOpen
  \bibfield  {author} {\bibinfo {author} {\bibfnamefont {E.}~\bibnamefont
  {Perrier}}, \bibinfo {author} {\bibfnamefont {D.}~\bibnamefont {Tao}},\ and\
  \bibinfo {author} {\bibfnamefont {C.}~\bibnamefont {Ferrie}},\ }\bibfield
  {title} {\bibinfo {title} {Quantum geometric machine learning for quantum
  circuits and control},\ }\href {https://doi.org/10.1088/1367-2630/abbf6b}
  {\bibfield  {journal} {\bibinfo  {journal} {New J. Phys.}\ }\textbf {\bibinfo
  {volume} {22}},\ \bibinfo {pages} {103056} (\bibinfo {year}
  {2020})}\BibitemShut {NoStop}%
\bibitem [{\citenamefont {Xu}\ \emph {et~al.}(2021)\citenamefont {Xu},
  \citenamefont {Benjamin},\ and\ \citenamefont {Yuan}}]{Xu2021vari}%
  \BibitemOpen
  \bibfield  {author} {\bibinfo {author} {\bibfnamefont {X.}~\bibnamefont
  {Xu}}, \bibinfo {author} {\bibfnamefont {S.~C.}\ \bibnamefont {Benjamin}},\
  and\ \bibinfo {author} {\bibfnamefont {X.}~\bibnamefont {Yuan}},\ }\bibfield
  {title} {\bibinfo {title} {Variational circuit compiler for quantum error
  correction},\ }\href {https://doi.org/10.1103/PhysRevApplied.15.034068}
  {\bibfield  {journal} {\bibinfo  {journal} {Phys. Rev. Appl.}\ }\textbf
  {\bibinfo {volume} {15}},\ \bibinfo {pages} {034068} (\bibinfo {year}
  {2021})}\BibitemShut {NoStop}%
\bibitem [{\citenamefont {Cao}\ \emph {et~al.}(2022)\citenamefont {Cao},
  \citenamefont {Zhang}, \citenamefont {Wu}, \citenamefont {Grassl},\ and\
  \citenamefont {Zeng}}]{Cao2022quantumv}%
  \BibitemOpen
  \bibfield  {author} {\bibinfo {author} {\bibfnamefont {C.}~\bibnamefont
  {Cao}}, \bibinfo {author} {\bibfnamefont {C.}~\bibnamefont {Zhang}}, \bibinfo
  {author} {\bibfnamefont {Z.}~\bibnamefont {Wu}}, \bibinfo {author}
  {\bibfnamefont {M.}~\bibnamefont {Grassl}},\ and\ \bibinfo {author}
  {\bibfnamefont {B.}~\bibnamefont {Zeng}},\ }\bibfield  {title} {\bibinfo
  {title} {Quantum variational learning for quantum error-correcting codes},\
  }\href {https://doi.org/10.22331/q-2022-10-06-828} {\bibfield  {journal}
  {\bibinfo  {journal} {{Quantum}}\ }\textbf {\bibinfo {volume} {6}},\ \bibinfo
  {pages} {828} (\bibinfo {year} {2022})}\BibitemShut {NoStop}%
\bibitem [{\citenamefont {Johnson}\ \emph {et~al.}(2017)\citenamefont
  {Johnson}, \citenamefont {Romero}, \citenamefont {Olson}, \citenamefont
  {Cao},\ and\ \citenamefont {Aspuru-Guzik}}]{johnson2017qvector}%
  \BibitemOpen
  \bibfield  {author} {\bibinfo {author} {\bibfnamefont {P.~D.}\ \bibnamefont
  {Johnson}}, \bibinfo {author} {\bibfnamefont {J.}~\bibnamefont {Romero}},
  \bibinfo {author} {\bibfnamefont {J.}~\bibnamefont {Olson}}, \bibinfo
  {author} {\bibfnamefont {Y.}~\bibnamefont {Cao}},\ and\ \bibinfo {author}
  {\bibfnamefont {A.}~\bibnamefont {Aspuru-Guzik}},\ }\bibfield  {title}
  {\bibinfo {title} {{QVECTOR}: an algorithm for device-tailored quantum error
  correction},\ }\href {https://arxiv.org/abs/1711.02249} {\bibfield  {journal}
  {\bibinfo  {journal} {arXiv preprint arXiv:1711.02249}\ } (\bibinfo {year}
  {2017})}\BibitemShut {NoStop}%
\bibitem [{\citenamefont {Yuan}\ \emph {et~al.}(2019)\citenamefont {Yuan},
  \citenamefont {Endo}, \citenamefont {Zhao}, \citenamefont {Li},\ and\
  \citenamefont {Benjamin}}]{Yuan2019}%
  \BibitemOpen
  \bibfield  {author} {\bibinfo {author} {\bibfnamefont {X.}~\bibnamefont
  {Yuan}}, \bibinfo {author} {\bibfnamefont {S.}~\bibnamefont {Endo}}, \bibinfo
  {author} {\bibfnamefont {Q.}~\bibnamefont {Zhao}}, \bibinfo {author}
  {\bibfnamefont {Y.}~\bibnamefont {Li}},\ and\ \bibinfo {author}
  {\bibfnamefont {S.~C.}\ \bibnamefont {Benjamin}},\ }\bibfield  {title}
  {\bibinfo {title} {Theory of variational quantum simulation},\ }\href
  {https://doi.org/10.22331/q-2019-10-07-191} {\bibfield  {journal} {\bibinfo
  {journal} {Quantum}\ }\textbf {\bibinfo {volume} {3}},\ \bibinfo {pages}
  {191} (\bibinfo {year} {2019})}\BibitemShut {NoStop}%
\bibitem [{\citenamefont {Altman}\ \emph {et~al.}(2021)\citenamefont {Altman},
  \citenamefont {Brown}, \citenamefont {Carleo}, \citenamefont {Carr},
  \citenamefont {Demler}, \citenamefont {Chin}, \citenamefont {DeMarco},
  \citenamefont {Economou}, \citenamefont {Eriksson}, \citenamefont {Fu},
  \citenamefont {Greiner}, \citenamefont {Hazzard}, \citenamefont {Hulet},
  \citenamefont {Koll\'ar}, \citenamefont {Lev}, \citenamefont {Lukin},
  \citenamefont {Ma}, \citenamefont {Mi}, \citenamefont {Misra}, \citenamefont
  {Monroe}, \citenamefont {Murch}, \citenamefont {Nazario}, \citenamefont {Ni},
  \citenamefont {Potter}, \citenamefont {Roushan}, \citenamefont {Saffman},
  \citenamefont {Schleier-Smith}, \citenamefont {Siddiqi}, \citenamefont
  {Simmonds}, \citenamefont {Singh}, \citenamefont {Spielman}, \citenamefont
  {Temme}, \citenamefont {Weiss}, \citenamefont {Vu\ifmmode \check{c}\else
  \v{c}\fi{}kovi\ifmmode~\acute{c}\else \'{c}\fi{}}, \citenamefont
  {Vuleti\ifmmode~\acute{c}\else \'{c}\fi{}}, \citenamefont {Ye},\ and\
  \citenamefont {Zwierlein}}]{Altman2021}%
  \BibitemOpen
  \bibfield  {author} {\bibinfo {author} {\bibfnamefont {E.}~\bibnamefont
  {Altman}}, \bibinfo {author} {\bibfnamefont {K.~R.}\ \bibnamefont {Brown}},
  \bibinfo {author} {\bibfnamefont {G.}~\bibnamefont {Carleo}}, \bibinfo
  {author} {\bibfnamefont {L.~D.}\ \bibnamefont {Carr}}, \bibinfo {author}
  {\bibfnamefont {E.}~\bibnamefont {Demler}}, \bibinfo {author} {\bibfnamefont
  {C.}~\bibnamefont {Chin}}, \bibinfo {author} {\bibfnamefont {B.}~\bibnamefont
  {DeMarco}}, \bibinfo {author} {\bibfnamefont {S.~E.}\ \bibnamefont
  {Economou}}, \bibinfo {author} {\bibfnamefont {M.~A.}\ \bibnamefont
  {Eriksson}}, \bibinfo {author} {\bibfnamefont {K.-M.~C.}\ \bibnamefont {Fu}},
  \bibinfo {author} {\bibfnamefont {M.}~\bibnamefont {Greiner}}, \bibinfo
  {author} {\bibfnamefont {K.~R.~A.}\ \bibnamefont {Hazzard}}, \bibinfo
  {author} {\bibfnamefont {R.~G.}\ \bibnamefont {Hulet}}, \bibinfo {author}
  {\bibfnamefont {A.~J.}\ \bibnamefont {Koll\'ar}}, \bibinfo {author}
  {\bibfnamefont {B.~L.}\ \bibnamefont {Lev}}, \bibinfo {author} {\bibfnamefont
  {M.~D.}\ \bibnamefont {Lukin}}, \bibinfo {author} {\bibfnamefont
  {R.}~\bibnamefont {Ma}}, \bibinfo {author} {\bibfnamefont {X.}~\bibnamefont
  {Mi}}, \bibinfo {author} {\bibfnamefont {S.}~\bibnamefont {Misra}}, \bibinfo
  {author} {\bibfnamefont {C.}~\bibnamefont {Monroe}}, \bibinfo {author}
  {\bibfnamefont {K.}~\bibnamefont {Murch}}, \bibinfo {author} {\bibfnamefont
  {Z.}~\bibnamefont {Nazario}}, \bibinfo {author} {\bibfnamefont {K.-K.}\
  \bibnamefont {Ni}}, \bibinfo {author} {\bibfnamefont {A.~C.}\ \bibnamefont
  {Potter}}, \bibinfo {author} {\bibfnamefont {P.}~\bibnamefont {Roushan}},
  \bibinfo {author} {\bibfnamefont {M.}~\bibnamefont {Saffman}}, \bibinfo
  {author} {\bibfnamefont {M.}~\bibnamefont {Schleier-Smith}}, \bibinfo
  {author} {\bibfnamefont {I.}~\bibnamefont {Siddiqi}}, \bibinfo {author}
  {\bibfnamefont {R.}~\bibnamefont {Simmonds}}, \bibinfo {author}
  {\bibfnamefont {M.}~\bibnamefont {Singh}}, \bibinfo {author} {\bibfnamefont
  {I.~B.}\ \bibnamefont {Spielman}}, \bibinfo {author} {\bibfnamefont
  {K.}~\bibnamefont {Temme}}, \bibinfo {author} {\bibfnamefont {D.~S.}\
  \bibnamefont {Weiss}}, \bibinfo {author} {\bibfnamefont {J.}~\bibnamefont
  {Vu\ifmmode \check{c}\else \v{c}\fi{}kovi\ifmmode~\acute{c}\else
  \'{c}\fi{}}}, \bibinfo {author} {\bibfnamefont {V.}~\bibnamefont
  {Vuleti\ifmmode~\acute{c}\else \'{c}\fi{}}}, \bibinfo {author} {\bibfnamefont
  {J.}~\bibnamefont {Ye}},\ and\ \bibinfo {author} {\bibfnamefont
  {M.}~\bibnamefont {Zwierlein}},\ }\bibfield  {title} {\bibinfo {title}
  {Quantum simulators: Architectures and opportunities},\ }\href
  {https://doi.org/10.1103/PRXQuantum.2.017003} {\bibfield  {journal} {\bibinfo
   {journal} {PRX Quantum}\ }\textbf {\bibinfo {volume} {2}},\ \bibinfo {pages}
  {017003} (\bibinfo {year} {2021})}\BibitemShut {NoStop}%
\bibitem [{\citenamefont {Bharti}\ and\ \citenamefont
  {Haug}(2021)}]{Bharti2021q}%
  \BibitemOpen
  \bibfield  {author} {\bibinfo {author} {\bibfnamefont {K.}~\bibnamefont
  {Bharti}}\ and\ \bibinfo {author} {\bibfnamefont {T.}~\bibnamefont {Haug}},\
  }\bibfield  {title} {\bibinfo {title} {Quantum-assisted simulator},\ }\href
  {https://doi.org/10.1103/PhysRevA.104.042418} {\bibfield  {journal} {\bibinfo
   {journal} {Phys. Rev. A}\ }\textbf {\bibinfo {volume} {104}},\ \bibinfo
  {pages} {042418} (\bibinfo {year} {2021})}\BibitemShut {NoStop}%
\bibitem [{\citenamefont {Ho}\ and\ \citenamefont {Hsieh}(2019)}]{Wen2019}%
  \BibitemOpen
  \bibfield  {author} {\bibinfo {author} {\bibfnamefont {W.~W.}\ \bibnamefont
  {Ho}}\ and\ \bibinfo {author} {\bibfnamefont {T.~H.}\ \bibnamefont {Hsieh}},\
  }\bibfield  {title} {\bibinfo {title} {Efficient variational simulation of
  non-trivial quantum states},\ }\href
  {https://doi.org/10.21468/SciPostPhys.6.3.029} {\bibfield  {journal}
  {\bibinfo  {journal} {SciPost Phys.}\ }\textbf {\bibinfo {volume} {6}},\
  \bibinfo {pages} {029} (\bibinfo {year} {2019})}\BibitemShut {NoStop}%
\bibitem [{\citenamefont {Tabares}\ \emph {et~al.}(2023)\citenamefont
  {Tabares}, \citenamefont {Mu\~noz de~las Heras}, \citenamefont {Tagliacozzo},
  \citenamefont {Porras},\ and\ \citenamefont
  {Gonz\'alez-Tudela}}]{Tabares2023}%
  \BibitemOpen
  \bibfield  {author} {\bibinfo {author} {\bibfnamefont {C.}~\bibnamefont
  {Tabares}}, \bibinfo {author} {\bibfnamefont {A.}~\bibnamefont {Mu\~noz
  de~las Heras}}, \bibinfo {author} {\bibfnamefont {L.}~\bibnamefont
  {Tagliacozzo}}, \bibinfo {author} {\bibfnamefont {D.}~\bibnamefont
  {Porras}},\ and\ \bibinfo {author} {\bibfnamefont {A.}~\bibnamefont
  {Gonz\'alez-Tudela}},\ }\bibfield  {title} {\bibinfo {title} {Variational
  quantum simulators based on waveguide {QED}},\ }\href
  {https://doi.org/10.1103/PhysRevLett.131.073602} {\bibfield  {journal}
  {\bibinfo  {journal} {Phys. Rev. Lett.}\ }\textbf {\bibinfo {volume} {131}},\
  \bibinfo {pages} {073602} (\bibinfo {year} {2023})}\BibitemShut {NoStop}%
\bibitem [{\citenamefont {Svensson}\ \emph {et~al.}(2023)\citenamefont
  {Svensson}, \citenamefont {Andersson}, \citenamefont {Gr\"onkvist},
  \citenamefont {Vikst\aa{}l}, \citenamefont {Dubhashi}, \citenamefont
  {Ferrini},\ and\ \citenamefont {Johansson}}]{Svensson2023}%
  \BibitemOpen
  \bibfield  {author} {\bibinfo {author} {\bibfnamefont {M.}~\bibnamefont
  {Svensson}}, \bibinfo {author} {\bibfnamefont {M.}~\bibnamefont {Andersson}},
  \bibinfo {author} {\bibfnamefont {M.}~\bibnamefont {Gr\"onkvist}}, \bibinfo
  {author} {\bibfnamefont {P.}~\bibnamefont {Vikst\aa{}l}}, \bibinfo {author}
  {\bibfnamefont {D.}~\bibnamefont {Dubhashi}}, \bibinfo {author}
  {\bibfnamefont {G.}~\bibnamefont {Ferrini}},\ and\ \bibinfo {author}
  {\bibfnamefont {G.}~\bibnamefont {Johansson}},\ }\bibfield  {title} {\bibinfo
  {title} {Hybrid quantum-classical heuristic to solve large-scale integer
  linear programs},\ }\href {https://doi.org/10.1103/PhysRevApplied.20.034062}
  {\bibfield  {journal} {\bibinfo  {journal} {Phys. Rev. Appl.}\ }\textbf
  {\bibinfo {volume} {20}},\ \bibinfo {pages} {034062} (\bibinfo {year}
  {2023})}\BibitemShut {NoStop}%
\bibitem [{\citenamefont {Farhi}\ \emph {et~al.}(2014)\citenamefont {Farhi},
  \citenamefont {Goldstone},\ and\ \citenamefont {Gutmann}}]{Farhi2014}%
  \BibitemOpen
  \bibfield  {author} {\bibinfo {author} {\bibfnamefont {E.}~\bibnamefont
  {Farhi}}, \bibinfo {author} {\bibfnamefont {J.}~\bibnamefont {Goldstone}},\
  and\ \bibinfo {author} {\bibfnamefont {S.}~\bibnamefont {Gutmann}},\
  }\bibfield  {title} {\bibinfo {title} {A quantum approximate optimization
  algorithm},\ }\href {https://doi.org/10.48550/arXiv.1411.4028} {\bibfield
  {journal} {\bibinfo  {journal} {arXiv preprint arXiv:1411.4028}\ } (\bibinfo
  {year} {2014})}\BibitemShut {NoStop}%
\bibitem [{\citenamefont {Crooks}(2018)}]{Crooks2018}%
  \BibitemOpen
  \bibfield  {author} {\bibinfo {author} {\bibfnamefont {G.~E.}\ \bibnamefont
  {Crooks}},\ }\bibfield  {title} {\bibinfo {title} {Performance of the quantum
  approximate optimization algorithm on the maximum cut problem},\ }\href
  {https://doi.org/10.48550/arXiv.1811.08419} {\bibfield  {journal} {\bibinfo
  {journal} {arXiv preprint arXiv:1811.08419}\ } (\bibinfo {year}
  {2018})}\BibitemShut {NoStop}%
\bibitem [{\citenamefont {Vikst\aa{}l}\ \emph {et~al.}(2020)\citenamefont
  {Vikst\aa{}l}, \citenamefont {Gr\"onkvist}, \citenamefont {Svensson},
  \citenamefont {Andersson}, \citenamefont {Johansson},\ and\ \citenamefont
  {Ferrini}}]{Vikst2020}%
  \BibitemOpen
  \bibfield  {author} {\bibinfo {author} {\bibfnamefont {P.}~\bibnamefont
  {Vikst\aa{}l}}, \bibinfo {author} {\bibfnamefont {M.}~\bibnamefont
  {Gr\"onkvist}}, \bibinfo {author} {\bibfnamefont {M.}~\bibnamefont
  {Svensson}}, \bibinfo {author} {\bibfnamefont {M.}~\bibnamefont {Andersson}},
  \bibinfo {author} {\bibfnamefont {G.}~\bibnamefont {Johansson}},\ and\
  \bibinfo {author} {\bibfnamefont {G.}~\bibnamefont {Ferrini}},\ }\bibfield
  {title} {\bibinfo {title} {Applying the quantum approximate optimization
  algorithm to the tail-assignment problem},\ }\href
  {https://doi.org/10.1103/PhysRevApplied.14.034009} {\bibfield  {journal}
  {\bibinfo  {journal} {Phys. Rev. Appl.}\ }\textbf {\bibinfo {volume} {14}},\
  \bibinfo {pages} {034009} (\bibinfo {year} {2020})}\BibitemShut {NoStop}%
\bibitem [{\citenamefont {Wang}\ \emph
  {et~al.}(2021{\natexlab{a}})\citenamefont {Wang}, \citenamefont {Li},\ and\
  \citenamefont {Wang}}]{Wang2021variational}%
  \BibitemOpen
  \bibfield  {author} {\bibinfo {author} {\bibfnamefont {Y.}~\bibnamefont
  {Wang}}, \bibinfo {author} {\bibfnamefont {G.}~\bibnamefont {Li}},\ and\
  \bibinfo {author} {\bibfnamefont {X.}~\bibnamefont {Wang}},\ }\bibfield
  {title} {\bibinfo {title} {Variational quantum {Gibbs} state preparation with
  a truncated taylor series},\ }\href
  {https://doi.org/10.1103/PhysRevApplied.16.054035} {\bibfield  {journal}
  {\bibinfo  {journal} {Phys. Rev. Appl.}\ }\textbf {\bibinfo {volume} {16}},\
  \bibinfo {pages} {054035} (\bibinfo {year} {2021}{\natexlab{a}})}\BibitemShut
  {NoStop}%
\bibitem [{\citenamefont {Jiang}\ \emph {et~al.}(2018)\citenamefont {Jiang},
  \citenamefont {Sung}, \citenamefont {Kechedzhi}, \citenamefont
  {Smelyanskiy},\ and\ \citenamefont {Boixo}}]{Jiang2018}%
  \BibitemOpen
  \bibfield  {author} {\bibinfo {author} {\bibfnamefont {Z.}~\bibnamefont
  {Jiang}}, \bibinfo {author} {\bibfnamefont {K.~J.}\ \bibnamefont {Sung}},
  \bibinfo {author} {\bibfnamefont {K.}~\bibnamefont {Kechedzhi}}, \bibinfo
  {author} {\bibfnamefont {V.~N.}\ \bibnamefont {Smelyanskiy}},\ and\ \bibinfo
  {author} {\bibfnamefont {S.}~\bibnamefont {Boixo}},\ }\bibfield  {title}
  {\bibinfo {title} {Quantum algorithms to simulate many-body physics of
  correlated fermions},\ }\href
  {https://doi.org/10.1103/PhysRevApplied.9.044036} {\bibfield  {journal}
  {\bibinfo  {journal} {Phys. Rev. Appl.}\ }\textbf {\bibinfo {volume} {9}},\
  \bibinfo {pages} {044036} (\bibinfo {year} {2018})}\BibitemShut {NoStop}%
\bibitem [{\citenamefont {Liang}\ \emph {et~al.}(2022)\citenamefont {Liang},
  \citenamefont {Wang}, \citenamefont {Cheng}, \citenamefont {Ding},
  \citenamefont {Ren}, \citenamefont {Gao}, \citenamefont {Hu}, \citenamefont
  {Boning}, \citenamefont {Qian}, \citenamefont {Han}, \citenamefont {Jiang},\
  and\ \citenamefont {Shi}}]{Liang2022}%
  \BibitemOpen
  \bibfield  {author} {\bibinfo {author} {\bibfnamefont {Z.}~\bibnamefont
  {Liang}}, \bibinfo {author} {\bibfnamefont {H.}~\bibnamefont {Wang}},
  \bibinfo {author} {\bibfnamefont {J.}~\bibnamefont {Cheng}}, \bibinfo
  {author} {\bibfnamefont {Y.}~\bibnamefont {Ding}}, \bibinfo {author}
  {\bibfnamefont {H.}~\bibnamefont {Ren}}, \bibinfo {author} {\bibfnamefont
  {Z.}~\bibnamefont {Gao}}, \bibinfo {author} {\bibfnamefont {Z.}~\bibnamefont
  {Hu}}, \bibinfo {author} {\bibfnamefont {D.~S.}\ \bibnamefont {Boning}},
  \bibinfo {author} {\bibfnamefont {X.}~\bibnamefont {Qian}}, \bibinfo {author}
  {\bibfnamefont {S.}~\bibnamefont {Han}}, \bibinfo {author} {\bibfnamefont
  {W.}~\bibnamefont {Jiang}},\ and\ \bibinfo {author} {\bibfnamefont
  {Y.}~\bibnamefont {Shi}},\ }\bibfield  {title} {\bibinfo {title} {Variational
  quantum pulse learning},\ }in\ \href
  {https://ieeexplore.ieee.org/document/9951187} {\emph {\bibinfo {booktitle}
  {2022 IEEE International Conference on Quantum Computing and Engineering
  (QCE)}}}\ (\bibinfo {year} {2022})\ pp.\ \bibinfo {pages}
  {556--565}\BibitemShut {NoStop}%
\bibitem [{\citenamefont {McClean}\ \emph {et~al.}(2018)\citenamefont
  {McClean}, \citenamefont {Boixo}, \citenamefont {Smelyanskiy}, \citenamefont
  {Babbush},\ and\ \citenamefont {Neven}}]{McClean2018}%
  \BibitemOpen
  \bibfield  {author} {\bibinfo {author} {\bibfnamefont {J.~R.}\ \bibnamefont
  {McClean}}, \bibinfo {author} {\bibfnamefont {S.}~\bibnamefont {Boixo}},
  \bibinfo {author} {\bibfnamefont {V.~N.}\ \bibnamefont {Smelyanskiy}},
  \bibinfo {author} {\bibfnamefont {R.}~\bibnamefont {Babbush}},\ and\ \bibinfo
  {author} {\bibfnamefont {H.}~\bibnamefont {Neven}},\ }\bibfield  {title}
  {\bibinfo {title} {Barren plateaus in quantum neural network training
  landscapes},\ }\href {https://doi.org/10.1038/s41467-018-07090-4} {\bibfield
  {journal} {\bibinfo  {journal} {Nat. Commun.}\ }\textbf {\bibinfo {volume}
  {9}},\ \bibinfo {pages} {1} (\bibinfo {year} {2018})}\BibitemShut {NoStop}%
\bibitem [{\citenamefont {Haug}\ \emph {et~al.}(2021)\citenamefont {Haug},
  \citenamefont {Bharti},\ and\ \citenamefont {Kim}}]{Haug2021c}%
  \BibitemOpen
  \bibfield  {author} {\bibinfo {author} {\bibfnamefont {T.}~\bibnamefont
  {Haug}}, \bibinfo {author} {\bibfnamefont {K.}~\bibnamefont {Bharti}},\ and\
  \bibinfo {author} {\bibfnamefont {M.~S.}\ \bibnamefont {Kim}},\ }\bibfield
  {title} {\bibinfo {title} {Capacity and quantum geometry of parametrized
  quantum circuits},\ }\href {https://doi.org/10.1103/PRXQuantum.2.040309}
  {\bibfield  {journal} {\bibinfo  {journal} {PRX Quantum}\ }\textbf {\bibinfo
  {volume} {2}},\ \bibinfo {pages} {040309} (\bibinfo {year}
  {2021})}\BibitemShut {NoStop}%
\bibitem [{\citenamefont {Zhao}\ and\ \citenamefont {Gao}(2021)}]{Zhao2021}%
  \BibitemOpen
  \bibfield  {author} {\bibinfo {author} {\bibfnamefont {C.}~\bibnamefont
  {Zhao}}\ and\ \bibinfo {author} {\bibfnamefont {X.-S.}\ \bibnamefont {Gao}},\
  }\bibfield  {title} {\bibinfo {title} {Analyzing the barren plateau
  phenomenon in training quantum neural networks with the {ZX}-calculus},\
  }\href {https://doi.org/10.22331/q-2021-06-04-466} {\bibfield  {journal}
  {\bibinfo  {journal} {Quantum}\ }\textbf {\bibinfo {volume} {5}},\ \bibinfo
  {pages} {466} (\bibinfo {year} {2021})}\BibitemShut {NoStop}%
\bibitem [{\citenamefont {Cerezo}\ \emph
  {et~al.}(2021{\natexlab{b}})\citenamefont {Cerezo}, \citenamefont {Sone},
  \citenamefont {Volkoff}, \citenamefont {Cincio},\ and\ \citenamefont
  {Coles}}]{Cerezo2021c}%
  \BibitemOpen
  \bibfield  {author} {\bibinfo {author} {\bibfnamefont {M.}~\bibnamefont
  {Cerezo}}, \bibinfo {author} {\bibfnamefont {A.}~\bibnamefont {Sone}},
  \bibinfo {author} {\bibfnamefont {T.}~\bibnamefont {Volkoff}}, \bibinfo
  {author} {\bibfnamefont {L.}~\bibnamefont {Cincio}},\ and\ \bibinfo {author}
  {\bibfnamefont {P.~J.}\ \bibnamefont {Coles}},\ }\bibfield  {title} {\bibinfo
  {title} {Cost function dependent barren plateaus in shallow parametrized
  quantum circuits},\ }\href {https://doi.org/10.1038/s41467-021-21728-w}
  {\bibfield  {journal} {\bibinfo  {journal} {Nat. Commun.}\ }\textbf {\bibinfo
  {volume} {12}},\ \bibinfo {pages} {1} (\bibinfo {year}
  {2021}{\natexlab{b}})}\BibitemShut {NoStop}%
\bibitem [{\citenamefont {Liu}\ \emph {et~al.}(2022)\citenamefont {Liu},
  \citenamefont {Yu}, \citenamefont {Duan},\ and\ \citenamefont
  {Deng}}]{Liuzidu2021}%
  \BibitemOpen
  \bibfield  {author} {\bibinfo {author} {\bibfnamefont {Z.}~\bibnamefont
  {Liu}}, \bibinfo {author} {\bibfnamefont {L.-W.}\ \bibnamefont {Yu}},
  \bibinfo {author} {\bibfnamefont {L.-M.}\ \bibnamefont {Duan}},\ and\
  \bibinfo {author} {\bibfnamefont {D.-L.}\ \bibnamefont {Deng}},\ }\bibfield
  {title} {\bibinfo {title} {Presence and absence of barren plateaus in
  tensor-network based machine learning},\ }\href
  {https://doi.org/10.1103/PhysRevLett.129.270501} {\bibfield  {journal}
  {\bibinfo  {journal} {Phys. Rev. Lett.}\ }\textbf {\bibinfo {volume} {129}},\
  \bibinfo {pages} {270501} (\bibinfo {year} {2022})}\BibitemShut {NoStop}%
\bibitem [{\citenamefont {Holmes}\ \emph {et~al.}(2022)\citenamefont {Holmes},
  \citenamefont {Sharma}, \citenamefont {Cerezo},\ and\ \citenamefont
  {Coles}}]{Holmes2022}%
  \BibitemOpen
  \bibfield  {author} {\bibinfo {author} {\bibfnamefont {Z.}~\bibnamefont
  {Holmes}}, \bibinfo {author} {\bibfnamefont {K.}~\bibnamefont {Sharma}},
  \bibinfo {author} {\bibfnamefont {M.}~\bibnamefont {Cerezo}},\ and\ \bibinfo
  {author} {\bibfnamefont {P.~J.}\ \bibnamefont {Coles}},\ }\bibfield  {title}
  {\bibinfo {title} {Connecting ansatz expressibility to gradient magnitudes
  and barren plateaus},\ }\href {https://doi.org/10.1103/PRXQuantum.3.010313}
  {\bibfield  {journal} {\bibinfo  {journal} {PRX Quantum}\ }\textbf {\bibinfo
  {volume} {3}},\ \bibinfo {pages} {010313} (\bibinfo {year}
  {2022})}\BibitemShut {NoStop}%
\bibitem [{\citenamefont {Ortiz~Marrero}\ \emph {et~al.}(2021)\citenamefont
  {Ortiz~Marrero}, \citenamefont {Kieferov\'a},\ and\ \citenamefont
  {Wiebe}}]{Marrero2021}%
  \BibitemOpen
  \bibfield  {author} {\bibinfo {author} {\bibfnamefont {C.}~\bibnamefont
  {Ortiz~Marrero}}, \bibinfo {author} {\bibfnamefont {M.}~\bibnamefont
  {Kieferov\'a}},\ and\ \bibinfo {author} {\bibfnamefont {N.}~\bibnamefont
  {Wiebe}},\ }\bibfield  {title} {\bibinfo {title} {Entanglement-induced barren
  plateaus},\ }\href {https://doi.org/10.1103/PRXQuantum.2.040316} {\bibfield
  {journal} {\bibinfo  {journal} {PRX Quantum}\ }\textbf {\bibinfo {volume}
  {2}},\ \bibinfo {pages} {040316} (\bibinfo {year} {2021})}\BibitemShut
  {NoStop}%
\bibitem [{\citenamefont {Sharma}\ \emph {et~al.}(2022)\citenamefont {Sharma},
  \citenamefont {Cerezo}, \citenamefont {Cincio},\ and\ \citenamefont
  {Coles}}]{Sharma2022}%
  \BibitemOpen
  \bibfield  {author} {\bibinfo {author} {\bibfnamefont {K.}~\bibnamefont
  {Sharma}}, \bibinfo {author} {\bibfnamefont {M.}~\bibnamefont {Cerezo}},
  \bibinfo {author} {\bibfnamefont {L.}~\bibnamefont {Cincio}},\ and\ \bibinfo
  {author} {\bibfnamefont {P.~J.}\ \bibnamefont {Coles}},\ }\bibfield  {title}
  {\bibinfo {title} {Trainability of dissipative perceptron-based quantum
  neural networks},\ }\href {https://doi.org/10.1103/PhysRevLett.128.180505}
  {\bibfield  {journal} {\bibinfo  {journal} {Phys. Rev. Lett.}\ }\textbf
  {\bibinfo {volume} {128}},\ \bibinfo {pages} {180505} (\bibinfo {year}
  {2022})}\BibitemShut {NoStop}%
\bibitem [{\citenamefont {Patti}\ \emph {et~al.}(2021)\citenamefont {Patti},
  \citenamefont {Najafi}, \citenamefont {Gao},\ and\ \citenamefont
  {Yelin}}]{Patti2021}%
  \BibitemOpen
  \bibfield  {author} {\bibinfo {author} {\bibfnamefont {T.~L.}\ \bibnamefont
  {Patti}}, \bibinfo {author} {\bibfnamefont {K.}~\bibnamefont {Najafi}},
  \bibinfo {author} {\bibfnamefont {X.}~\bibnamefont {Gao}},\ and\ \bibinfo
  {author} {\bibfnamefont {S.~F.}\ \bibnamefont {Yelin}},\ }\bibfield  {title}
  {\bibinfo {title} {Entanglement devised barren plateau mitigation},\ }\href
  {https://doi.org/10.1103/PhysRevResearch.3.033090} {\bibfield  {journal}
  {\bibinfo  {journal} {Phys. Rev. Res.}\ }\textbf {\bibinfo {volume} {3}},\
  \bibinfo {pages} {033090} (\bibinfo {year} {2021})}\BibitemShut {NoStop}%
\bibitem [{\citenamefont {Wang}\ \emph
  {et~al.}(2021{\natexlab{b}})\citenamefont {Wang}, \citenamefont {Fontana},
  \citenamefont {Cerezo}, \citenamefont {Sharma}, \citenamefont {Sone},
  \citenamefont {Cincio},\ and\ \citenamefont {Coles}}]{Wangsamon2021}%
  \BibitemOpen
  \bibfield  {author} {\bibinfo {author} {\bibfnamefont {S.}~\bibnamefont
  {Wang}}, \bibinfo {author} {\bibfnamefont {E.}~\bibnamefont {Fontana}},
  \bibinfo {author} {\bibfnamefont {M.}~\bibnamefont {Cerezo}}, \bibinfo
  {author} {\bibfnamefont {K.}~\bibnamefont {Sharma}}, \bibinfo {author}
  {\bibfnamefont {A.}~\bibnamefont {Sone}}, \bibinfo {author} {\bibfnamefont
  {L.}~\bibnamefont {Cincio}},\ and\ \bibinfo {author} {\bibfnamefont {P.~J.}\
  \bibnamefont {Coles}},\ }\bibfield  {title} {\bibinfo {title} {Noise-induced
  barren plateaus in variational quantum algorithms},\ }\href
  {https://doi.org/10.1038/s41467-021-27045-6} {\bibfield  {journal} {\bibinfo
  {journal} {Nat. Commun.}\ }\textbf {\bibinfo {volume} {12}},\ \bibinfo
  {pages} {1} (\bibinfo {year} {2021}{\natexlab{b}})}\BibitemShut {NoStop}%
\bibitem [{\citenamefont {Anschuetz}\ and\ \citenamefont
  {Kiani}(2022)}]{Anschuetz2022b}%
  \BibitemOpen
  \bibfield  {author} {\bibinfo {author} {\bibfnamefont {E.~R.}\ \bibnamefont
  {Anschuetz}}\ and\ \bibinfo {author} {\bibfnamefont {B.~T.}\ \bibnamefont
  {Kiani}},\ }\bibfield  {title} {\bibinfo {title} {Quantum variational
  algorithms are swamped with traps},\ }\href
  {https://doi.org/10.1038/s41467-022-35364-5} {\bibfield  {journal} {\bibinfo
  {journal} {Nat. Commun.}\ }\textbf {\bibinfo {volume} {13}},\ \bibinfo
  {pages} {7760} (\bibinfo {year} {2022})}\BibitemShut {NoStop}%
\bibitem [{\citenamefont {You}\ and\ \citenamefont {Wu}(2021)}]{You2021}%
  \BibitemOpen
  \bibfield  {author} {\bibinfo {author} {\bibfnamefont {X.}~\bibnamefont
  {You}}\ and\ \bibinfo {author} {\bibfnamefont {X.}~\bibnamefont {Wu}},\
  }\bibfield  {title} {\bibinfo {title} {Exponentially many local minima in
  quantum neural networks},\ }in\ \href
  {http://proceedings.mlr.press/v139/you21c.html} {\emph {\bibinfo {booktitle}
  {Proceedings of the 38th International Conference on Machine Learning (ICML
  2021)}}}\ (\bibinfo {year} {2021})\ pp.\ \bibinfo {pages}
  {12144--12155}\BibitemShut {NoStop}%
\bibitem [{\citenamefont {Bittel}\ and\ \citenamefont
  {Kliesch}(2021)}]{Bittel2021}%
  \BibitemOpen
  \bibfield  {author} {\bibinfo {author} {\bibfnamefont {L.}~\bibnamefont
  {Bittel}}\ and\ \bibinfo {author} {\bibfnamefont {M.}~\bibnamefont
  {Kliesch}},\ }\bibfield  {title} {\bibinfo {title} {Training variational
  quantum algorithms is {NP}-hard},\ }\href
  {https://doi.org/10.1103/PhysRevLett.127.120502} {\bibfield  {journal}
  {\bibinfo  {journal} {Phys. Rev. Lett.}\ }\textbf {\bibinfo {volume} {127}},\
  \bibinfo {pages} {120502} (\bibinfo {year} {2021})}\BibitemShut {NoStop}%
\bibitem [{\citenamefont {Wierichs}\ \emph {et~al.}(2020)\citenamefont
  {Wierichs}, \citenamefont {Gogolin},\ and\ \citenamefont
  {Kastoryano}}]{Wierichs2020}%
  \BibitemOpen
  \bibfield  {author} {\bibinfo {author} {\bibfnamefont {D.}~\bibnamefont
  {Wierichs}}, \bibinfo {author} {\bibfnamefont {C.}~\bibnamefont {Gogolin}},\
  and\ \bibinfo {author} {\bibfnamefont {M.}~\bibnamefont {Kastoryano}},\
  }\bibfield  {title} {\bibinfo {title} {Avoiding local minima in variational
  quantum eigensolvers with the natural gradient optimizer},\ }\href
  {https://doi.org/10.1103/PhysRevResearch.2.043246} {\bibfield  {journal}
  {\bibinfo  {journal} {Phys. Rev. Res.}\ }\textbf {\bibinfo {volume} {2}},\
  \bibinfo {pages} {043246} (\bibinfo {year} {2020})}\BibitemShut {NoStop}%
\bibitem [{\citenamefont {Ge}\ \emph {et~al.}(2022)\citenamefont {Ge},
  \citenamefont {Wu},\ and\ \citenamefont {Rabitz}}]{Ge2022}%
  \BibitemOpen
  \bibfield  {author} {\bibinfo {author} {\bibfnamefont {X.}~\bibnamefont
  {Ge}}, \bibinfo {author} {\bibfnamefont {R.-B.}\ \bibnamefont {Wu}},\ and\
  \bibinfo {author} {\bibfnamefont {H.}~\bibnamefont {Rabitz}},\ }\bibfield
  {title} {\bibinfo {title} {The optimization landscape of hybrid
  quantum--classical algorithms: From quantum control to {NISQ} applications},\
  }\href {https://doi.org/10.1016/j.arcontrol.2022.06.001} {\bibfield
  {journal} {\bibinfo  {journal} {Annual Reviews in Control}\ }\textbf
  {\bibinfo {volume} {54}},\ \bibinfo {pages} {314} (\bibinfo {year}
  {2022})}\BibitemShut {NoStop}%
\bibitem [{\citenamefont {Shi}\ \emph {et~al.}(2006)\citenamefont {Shi},
  \citenamefont {Duan},\ and\ \citenamefont {Vidal}}]{Shi2006}%
  \BibitemOpen
  \bibfield  {author} {\bibinfo {author} {\bibfnamefont {Y.-Y.}\ \bibnamefont
  {Shi}}, \bibinfo {author} {\bibfnamefont {L.-M.}\ \bibnamefont {Duan}},\ and\
  \bibinfo {author} {\bibfnamefont {G.}~\bibnamefont {Vidal}},\ }\bibfield
  {title} {\bibinfo {title} {Classical simulation of quantum many-body systems
  with a tree tensor network},\ }\href
  {https://doi.org/10.1103/PhysRevA.74.022320} {\bibfield  {journal} {\bibinfo
  {journal} {Phys. Rev. A}\ }\textbf {\bibinfo {volume} {74}},\ \bibinfo
  {pages} {022320} (\bibinfo {year} {2006})}\BibitemShut {NoStop}%
\bibitem [{\citenamefont {Grant}\ \emph {et~al.}(2018)\citenamefont {Grant},
  \citenamefont {Benedetti}, \citenamefont {Cao}, \citenamefont {Hallam},
  \citenamefont {Lockhart}, \citenamefont {Stojevic}, \citenamefont {Green},\
  and\ \citenamefont {Severini}}]{Grant2018}%
  \BibitemOpen
  \bibfield  {author} {\bibinfo {author} {\bibfnamefont {E.}~\bibnamefont
  {Grant}}, \bibinfo {author} {\bibfnamefont {M.}~\bibnamefont {Benedetti}},
  \bibinfo {author} {\bibfnamefont {S.}~\bibnamefont {Cao}}, \bibinfo {author}
  {\bibfnamefont {A.}~\bibnamefont {Hallam}}, \bibinfo {author} {\bibfnamefont
  {J.}~\bibnamefont {Lockhart}}, \bibinfo {author} {\bibfnamefont
  {V.}~\bibnamefont {Stojevic}}, \bibinfo {author} {\bibfnamefont {A.~G.}\
  \bibnamefont {Green}},\ and\ \bibinfo {author} {\bibfnamefont
  {S.}~\bibnamefont {Severini}},\ }\bibfield  {title} {\bibinfo {title}
  {Hierarchical quantum classifiers},\ }\href
  {https://doi.org/10.1038/s41534-018-0116-9} {\bibfield  {journal} {\bibinfo
  {journal} {npj Quantum Inf.}\ }\textbf {\bibinfo {volume} {4}},\ \bibinfo
  {pages} {1} (\bibinfo {year} {2018})}\BibitemShut {NoStop}%
\bibitem [{\citenamefont {Cong}\ \emph {et~al.}(2019)\citenamefont {Cong},
  \citenamefont {Choi},\ and\ \citenamefont {Lukin}}]{Cong2019}%
  \BibitemOpen
  \bibfield  {author} {\bibinfo {author} {\bibfnamefont {I.}~\bibnamefont
  {Cong}}, \bibinfo {author} {\bibfnamefont {S.}~\bibnamefont {Choi}},\ and\
  \bibinfo {author} {\bibfnamefont {M.~D.}\ \bibnamefont {Lukin}},\ }\bibfield
  {title} {\bibinfo {title} {Quantum convolutional neural networks},\ }\href
  {https://doi.org/10.1038/s41567-019-0648-8} {\bibfield  {journal} {\bibinfo
  {journal} {Nat. Phys.}\ }\textbf {\bibinfo {volume} {15}},\ \bibinfo {pages}
  {1273} (\bibinfo {year} {2019})}\BibitemShut {NoStop}%
\bibitem [{\citenamefont {Zhang}\ \emph {et~al.}(2020)\citenamefont {Zhang},
  \citenamefont {Hsieh}, \citenamefont {Liu},\ and\ \citenamefont
  {Tao}}]{Zhang2020}%
  \BibitemOpen
  \bibfield  {author} {\bibinfo {author} {\bibfnamefont {K.}~\bibnamefont
  {Zhang}}, \bibinfo {author} {\bibfnamefont {M.-H.}\ \bibnamefont {Hsieh}},
  \bibinfo {author} {\bibfnamefont {L.}~\bibnamefont {Liu}},\ and\ \bibinfo
  {author} {\bibfnamefont {D.}~\bibnamefont {Tao}},\ }\bibfield  {title}
  {\bibinfo {title} {Toward trainability of quantum neural networks},\ }\href
  {https://doi.org/10.48550/arXiv.2011.06258} {\bibfield  {journal} {\bibinfo
  {journal} {arXiv preprint arXiv:2011.06258}\ } (\bibinfo {year}
  {2020})}\BibitemShut {NoStop}%
\bibitem [{\citenamefont {Pesah}\ \emph {et~al.}(2021)\citenamefont {Pesah},
  \citenamefont {Cerezo}, \citenamefont {Wang}, \citenamefont {Volkoff},
  \citenamefont {Sornborger},\ and\ \citenamefont {Coles}}]{Pesah2021}%
  \BibitemOpen
  \bibfield  {author} {\bibinfo {author} {\bibfnamefont {A.}~\bibnamefont
  {Pesah}}, \bibinfo {author} {\bibfnamefont {M.}~\bibnamefont {Cerezo}},
  \bibinfo {author} {\bibfnamefont {S.}~\bibnamefont {Wang}}, \bibinfo {author}
  {\bibfnamefont {T.}~\bibnamefont {Volkoff}}, \bibinfo {author} {\bibfnamefont
  {A.~T.}\ \bibnamefont {Sornborger}},\ and\ \bibinfo {author} {\bibfnamefont
  {P.~J.}\ \bibnamefont {Coles}},\ }\bibfield  {title} {\bibinfo {title}
  {Absence of barren plateaus in quantum convolutional neural networks},\
  }\href {https://doi.org/10.1103/PhysRevX.11.041011} {\bibfield  {journal}
  {\bibinfo  {journal} {Phys. Rev. X}\ }\textbf {\bibinfo {volume} {11}},\
  \bibinfo {pages} {041011} (\bibinfo {year} {2021})}\BibitemShut {NoStop}%
\bibitem [{\citenamefont {Barthel}\ and\ \citenamefont
  {Miao}(2023)}]{barthel2023}%
  \BibitemOpen
  \bibfield  {author} {\bibinfo {author} {\bibfnamefont {T.}~\bibnamefont
  {Barthel}}\ and\ \bibinfo {author} {\bibfnamefont {Q.}~\bibnamefont {Miao}},\
  }\bibfield  {title} {\bibinfo {title} {Absence of barren plateaus and scaling
  of gradients in the energy optimization of isometric tensor network states},\
  }\href {https://doi.org/10.48550/arXiv.2304.00161} {\bibfield  {journal}
  {\bibinfo  {journal} {arXiv preprint arXiv:2304.00161}\ } (\bibinfo {year}
  {2023})}\BibitemShut {NoStop}%
\bibitem [{\citenamefont {Hadfield}\ \emph {et~al.}(2019)\citenamefont
  {Hadfield}, \citenamefont {Wang}, \citenamefont {O’gorman}, \citenamefont
  {Rieffel}, \citenamefont {Venturelli},\ and\ \citenamefont
  {Biswas}}]{Hadfield2019}%
  \BibitemOpen
  \bibfield  {author} {\bibinfo {author} {\bibfnamefont {S.}~\bibnamefont
  {Hadfield}}, \bibinfo {author} {\bibfnamefont {Z.}~\bibnamefont {Wang}},
  \bibinfo {author} {\bibfnamefont {B.}~\bibnamefont {O’gorman}}, \bibinfo
  {author} {\bibfnamefont {E.~G.}\ \bibnamefont {Rieffel}}, \bibinfo {author}
  {\bibfnamefont {D.}~\bibnamefont {Venturelli}},\ and\ \bibinfo {author}
  {\bibfnamefont {R.}~\bibnamefont {Biswas}},\ }\bibfield  {title} {\bibinfo
  {title} {From the quantum approximate optimization algorithm to a quantum
  alternating operator ansatz},\ }\href {https://doi.org/10.3390/a12020034}
  {\bibfield  {journal} {\bibinfo  {journal} {Algorithms}\ }\textbf {\bibinfo
  {volume} {12}},\ \bibinfo {pages} {34} (\bibinfo {year} {2019})}\BibitemShut
  {NoStop}%
\bibitem [{\citenamefont {Wiersema}\ \emph {et~al.}(2020)\citenamefont
  {Wiersema}, \citenamefont {Zhou}, \citenamefont {de~Sereville}, \citenamefont
  {Carrasquilla}, \citenamefont {Kim},\ and\ \citenamefont
  {Yuen}}]{Wiersema2020}%
  \BibitemOpen
  \bibfield  {author} {\bibinfo {author} {\bibfnamefont {R.}~\bibnamefont
  {Wiersema}}, \bibinfo {author} {\bibfnamefont {C.}~\bibnamefont {Zhou}},
  \bibinfo {author} {\bibfnamefont {Y.}~\bibnamefont {de~Sereville}}, \bibinfo
  {author} {\bibfnamefont {J.~F.}\ \bibnamefont {Carrasquilla}}, \bibinfo
  {author} {\bibfnamefont {Y.~B.}\ \bibnamefont {Kim}},\ and\ \bibinfo {author}
  {\bibfnamefont {H.}~\bibnamefont {Yuen}},\ }\bibfield  {title} {\bibinfo
  {title} {Exploring entanglement and optimization within the {H}amiltonian
  variational ansatz},\ }\href {https://doi.org/10.1103/PRXQuantum.1.020319}
  {\bibfield  {journal} {\bibinfo  {journal} {PRX Quantum}\ }\textbf {\bibinfo
  {volume} {1}},\ \bibinfo {pages} {020319} (\bibinfo {year}
  {2020})}\BibitemShut {NoStop}%
\bibitem [{\citenamefont {Lyu}\ \emph {et~al.}(2023)\citenamefont {Lyu},
  \citenamefont {Xu}, \citenamefont {Yung},\ and\ \citenamefont
  {Bayat}}]{Lyu2023symmetry}%
  \BibitemOpen
  \bibfield  {author} {\bibinfo {author} {\bibfnamefont {C.}~\bibnamefont
  {Lyu}}, \bibinfo {author} {\bibfnamefont {X.}~\bibnamefont {Xu}}, \bibinfo
  {author} {\bibfnamefont {M.-H.}\ \bibnamefont {Yung}},\ and\ \bibinfo
  {author} {\bibfnamefont {A.}~\bibnamefont {Bayat}},\ }\bibfield  {title}
  {\bibinfo {title} {Symmetry enhanced variational quantum spin eigensolver},\
  }\href {https://doi.org/10.22331/q-2023-01-19-899} {\bibfield  {journal}
  {\bibinfo  {journal} {{Quantum}}\ }\textbf {\bibinfo {volume} {7}},\ \bibinfo
  {pages} {899} (\bibinfo {year} {2023})}\BibitemShut {NoStop}%
\bibitem [{\citenamefont {Skolik}\ \emph {et~al.}(2021)\citenamefont {Skolik},
  \citenamefont {McClean}, \citenamefont {Mohseni}, \citenamefont {van~der
  Smagt},\ and\ \citenamefont {Leib}}]{Skolik2021}%
  \BibitemOpen
  \bibfield  {author} {\bibinfo {author} {\bibfnamefont {A.}~\bibnamefont
  {Skolik}}, \bibinfo {author} {\bibfnamefont {J.~R.}\ \bibnamefont {McClean}},
  \bibinfo {author} {\bibfnamefont {M.}~\bibnamefont {Mohseni}}, \bibinfo
  {author} {\bibfnamefont {P.}~\bibnamefont {van~der Smagt}},\ and\ \bibinfo
  {author} {\bibfnamefont {M.}~\bibnamefont {Leib}},\ }\bibfield  {title}
  {\bibinfo {title} {Layerwise learning for quantum neural networks},\ }\href
  {https://doi.org/10.1007/s42484-020-00036-4} {\bibfield  {journal} {\bibinfo
  {journal} {Quantum Mach. Intell.}\ }\textbf {\bibinfo {volume} {3}},\
  \bibinfo {pages} {1} (\bibinfo {year} {2021})}\BibitemShut {NoStop}%
\bibitem [{\citenamefont {Campos}\ \emph {et~al.}(2021)\citenamefont {Campos},
  \citenamefont {Rabinovich}, \citenamefont {Akshay},\ and\ \citenamefont
  {Biamonte}}]{Campos2021}%
  \BibitemOpen
  \bibfield  {author} {\bibinfo {author} {\bibfnamefont {E.}~\bibnamefont
  {Campos}}, \bibinfo {author} {\bibfnamefont {D.}~\bibnamefont {Rabinovich}},
  \bibinfo {author} {\bibfnamefont {V.}~\bibnamefont {Akshay}},\ and\ \bibinfo
  {author} {\bibfnamefont {J.}~\bibnamefont {Biamonte}},\ }\bibfield  {title}
  {\bibinfo {title} {Training saturation in layerwise quantum approximate
  optimization},\ }\href {https://doi.org/10.1103/PhysRevA.104.L030401}
  {\bibfield  {journal} {\bibinfo  {journal} {Phys. Rev. A}\ }\textbf {\bibinfo
  {volume} {104}},\ \bibinfo {pages} {L030401} (\bibinfo {year}
  {2021})}\BibitemShut {NoStop}%
\bibitem [{\citenamefont {Lockwood}(2021)}]{Lockwood2021}%
  \BibitemOpen
  \bibfield  {author} {\bibinfo {author} {\bibfnamefont {O.}~\bibnamefont
  {Lockwood}},\ }\bibfield  {title} {\bibinfo {title} {Optimizing quantum
  variational circuits with deep reinforcement learning},\ }\href
  {https://doi.org/10.48550/arXiv.2109.03188} {\bibfield  {journal} {\bibinfo
  {journal} {arXiv preprint arXiv:2109.03188}\ } (\bibinfo {year}
  {2021})}\BibitemShut {NoStop}%
\bibitem [{\citenamefont {Liu}\ \emph {et~al.}(2023)\citenamefont {Liu},
  \citenamefont {Najafi}, \citenamefont {Sharma}, \citenamefont {Tacchino},
  \citenamefont {Jiang},\ and\ \citenamefont {Mezzacapo}}]{Liu2022a}%
  \BibitemOpen
  \bibfield  {author} {\bibinfo {author} {\bibfnamefont {J.}~\bibnamefont
  {Liu}}, \bibinfo {author} {\bibfnamefont {K.}~\bibnamefont {Najafi}},
  \bibinfo {author} {\bibfnamefont {K.}~\bibnamefont {Sharma}}, \bibinfo
  {author} {\bibfnamefont {F.}~\bibnamefont {Tacchino}}, \bibinfo {author}
  {\bibfnamefont {L.}~\bibnamefont {Jiang}},\ and\ \bibinfo {author}
  {\bibfnamefont {A.}~\bibnamefont {Mezzacapo}},\ }\bibfield  {title} {\bibinfo
  {title} {Analytic theory for the dynamics of wide quantum neural networks},\
  }\href {https://doi.org/10.1103/PhysRevLett.130.150601} {\bibfield  {journal}
  {\bibinfo  {journal} {Phys. Rev. Lett.}\ }\textbf {\bibinfo {volume} {130}},\
  \bibinfo {pages} {150601} (\bibinfo {year} {2023})}\BibitemShut {NoStop}%
\bibitem [{\citenamefont {You}\ \emph {et~al.}(2022)\citenamefont {You},
  \citenamefont {Chakrabarti},\ and\ \citenamefont {Wu}}]{You2022conv}%
  \BibitemOpen
  \bibfield  {author} {\bibinfo {author} {\bibfnamefont {X.}~\bibnamefont
  {You}}, \bibinfo {author} {\bibfnamefont {S.}~\bibnamefont {Chakrabarti}},\
  and\ \bibinfo {author} {\bibfnamefont {X.}~\bibnamefont {Wu}},\ }\bibfield
  {title} {\bibinfo {title} {A convergence theory for over-parameterized
  variational quantum eigensolvers},\ }\href
  {https://doi.org/10.48550/arXiv.2205.12481} {\bibfield  {journal} {\bibinfo
  {journal} {arXiv preprint arXiv:2205.12481}\ } (\bibinfo {year}
  {2022})}\BibitemShut {NoStop}%
\bibitem [{\citenamefont {Anschuetz}(2022)}]{Anschuetz2022c}%
  \BibitemOpen
  \bibfield  {author} {\bibinfo {author} {\bibfnamefont {E.~R.}\ \bibnamefont
  {Anschuetz}},\ }\bibfield  {title} {\bibinfo {title} {Critical points in
  quantum generative models},\ }in\ \href
  {https://openreview.net/forum?id=2f1z55GVQN} {\emph {\bibinfo {booktitle}
  {International Conference on Learning Representations}}}\ (\bibinfo {year}
  {2022})\BibitemShut {NoStop}%
\bibitem [{\citenamefont {Larocca}\ \emph {et~al.}(2023)\citenamefont
  {Larocca}, \citenamefont {Ju}, \citenamefont {Garc{\'\i}a-Mart{\'\i}n},
  \citenamefont {Coles},\ and\ \citenamefont {Cerezo}}]{Larocca2021}%
  \BibitemOpen
  \bibfield  {author} {\bibinfo {author} {\bibfnamefont {M.}~\bibnamefont
  {Larocca}}, \bibinfo {author} {\bibfnamefont {N.}~\bibnamefont {Ju}},
  \bibinfo {author} {\bibfnamefont {D.}~\bibnamefont
  {Garc{\'\i}a-Mart{\'\i}n}}, \bibinfo {author} {\bibfnamefont {P.~J.}\
  \bibnamefont {Coles}},\ and\ \bibinfo {author} {\bibfnamefont
  {M.}~\bibnamefont {Cerezo}},\ }\bibfield  {title} {\bibinfo {title} {Theory
  of overparametrization in quantum neural networks},\ }\href
  {https://doi.org/10.1038/s43588-023-00467-6} {\bibfield  {journal} {\bibinfo
  {journal} {Nature Computational Science}\ }\textbf {\bibinfo {volume} {3}},\
  \bibinfo {pages} {542} (\bibinfo {year} {2023})}\BibitemShut {NoStop}%
\bibitem [{\citenamefont {Zhou}\ \emph {et~al.}(2020)\citenamefont {Zhou},
  \citenamefont {Wang}, \citenamefont {Choi}, \citenamefont {Pichler},\ and\
  \citenamefont {Lukin}}]{Zhou2020}%
  \BibitemOpen
  \bibfield  {author} {\bibinfo {author} {\bibfnamefont {L.}~\bibnamefont
  {Zhou}}, \bibinfo {author} {\bibfnamefont {S.-T.}\ \bibnamefont {Wang}},
  \bibinfo {author} {\bibfnamefont {S.}~\bibnamefont {Choi}}, \bibinfo {author}
  {\bibfnamefont {H.}~\bibnamefont {Pichler}},\ and\ \bibinfo {author}
  {\bibfnamefont {M.~D.}\ \bibnamefont {Lukin}},\ }\bibfield  {title} {\bibinfo
  {title} {Quantum approximate optimization algorithm: Performance, mechanism,
  and implementation on near-term devices},\ }\href
  {https://doi.org/10.1103/PhysRevX.10.021067} {\bibfield  {journal} {\bibinfo
  {journal} {Phys. Rev. X}\ }\textbf {\bibinfo {volume} {10}},\ \bibinfo
  {pages} {021067} (\bibinfo {year} {2020})}\BibitemShut {NoStop}%
\bibitem [{\citenamefont {Grant}\ \emph {et~al.}(2019)\citenamefont {Grant},
  \citenamefont {Wossnig}, \citenamefont {Ostaszewski},\ and\ \citenamefont
  {Benedetti}}]{Grant2019}%
  \BibitemOpen
  \bibfield  {author} {\bibinfo {author} {\bibfnamefont {E.}~\bibnamefont
  {Grant}}, \bibinfo {author} {\bibfnamefont {L.}~\bibnamefont {Wossnig}},
  \bibinfo {author} {\bibfnamefont {M.}~\bibnamefont {Ostaszewski}},\ and\
  \bibinfo {author} {\bibfnamefont {M.}~\bibnamefont {Benedetti}},\ }\bibfield
  {title} {\bibinfo {title} {An initialization strategy for addressing barren
  plateaus in parametrized quantum circuits},\ }\href
  {https://doi.org/10.22331/q-2019-12-09-214} {\bibfield  {journal} {\bibinfo
  {journal} {Quantum}\ }\textbf {\bibinfo {volume} {3}},\ \bibinfo {pages}
  {214} (\bibinfo {year} {2019})}\BibitemShut {NoStop}%
\bibitem [{\citenamefont {Volkoff}\ and\ \citenamefont
  {Coles}(2021)}]{Volkoff2021}%
  \BibitemOpen
  \bibfield  {author} {\bibinfo {author} {\bibfnamefont {T.}~\bibnamefont
  {Volkoff}}\ and\ \bibinfo {author} {\bibfnamefont {P.~J.}\ \bibnamefont
  {Coles}},\ }\bibfield  {title} {\bibinfo {title} {Large gradients via
  correlation in random parameterized quantum circuits},\ }\href
  {https://doi.org/10.1088/2058-9565/abd891} {\bibfield  {journal} {\bibinfo
  {journal} {Quantum Sci. Technol.}\ }\textbf {\bibinfo {volume} {6}},\
  \bibinfo {pages} {025008} (\bibinfo {year} {2021})}\BibitemShut {NoStop}%
\bibitem [{\citenamefont {Rad}\ \emph {et~al.}(2022)\citenamefont {Rad},
  \citenamefont {Seif},\ and\ \citenamefont {Linke}}]{Rad2022}%
  \BibitemOpen
  \bibfield  {author} {\bibinfo {author} {\bibfnamefont {A.}~\bibnamefont
  {Rad}}, \bibinfo {author} {\bibfnamefont {A.}~\bibnamefont {Seif}},\ and\
  \bibinfo {author} {\bibfnamefont {N.~M.}\ \bibnamefont {Linke}},\ }\bibfield
  {title} {\bibinfo {title} {Surviving the barren plateau in variational
  quantum circuits with {B}ayesian learning initialization},\ }\href
  {https://doi.org/10.48550/arXiv.2203.02464} {\bibfield  {journal} {\bibinfo
  {journal} {arXiv preprint arXiv:2203.02464}\ } (\bibinfo {year}
  {2022})}\BibitemShut {NoStop}%
\bibitem [{\citenamefont {Zhang}\ \emph {et~al.}(2022)\citenamefont {Zhang},
  \citenamefont {Liu}, \citenamefont {Hsieh},\ and\ \citenamefont
  {Tao}}]{Zhang2022g}%
  \BibitemOpen
  \bibfield  {author} {\bibinfo {author} {\bibfnamefont {K.}~\bibnamefont
  {Zhang}}, \bibinfo {author} {\bibfnamefont {L.}~\bibnamefont {Liu}}, \bibinfo
  {author} {\bibfnamefont {M.-H.}\ \bibnamefont {Hsieh}},\ and\ \bibinfo
  {author} {\bibfnamefont {D.}~\bibnamefont {Tao}},\ }\bibfield  {title}
  {\bibinfo {title} {Escaping from the barren plateau via gaussian
  initializations in deep variational quantum circuits},\ }in\ \href
  {https://openreview.net/forum?id=jXgbJdQ2YIy} {\emph {\bibinfo {booktitle}
  {Advances in Neural Information Processing Systems}}},\ \bibinfo {editor}
  {edited by\ \bibinfo {editor} {\bibfnamefont {A.~H.}\ \bibnamefont {Oh}},
  \bibinfo {editor} {\bibfnamefont {A.}~\bibnamefont {Agarwal}}, \bibinfo
  {editor} {\bibfnamefont {D.}~\bibnamefont {Belgrave}},\ and\ \bibinfo
  {editor} {\bibfnamefont {K.}~\bibnamefont {Cho}}}\ (\bibinfo {year}
  {2022})\BibitemShut {NoStop}%
\bibitem [{\citenamefont {Cao}\ \emph {et~al.}(2024)\citenamefont {Cao},
  \citenamefont {Zhou}, \citenamefont {Tannu}, \citenamefont {Shannon},\ and\
  \citenamefont {Joynt}}]{cao2024exploiting}%
  \BibitemOpen
  \bibfield  {author} {\bibinfo {author} {\bibfnamefont {C.}~\bibnamefont
  {Cao}}, \bibinfo {author} {\bibfnamefont {Y.}~\bibnamefont {Zhou}}, \bibinfo
  {author} {\bibfnamefont {S.}~\bibnamefont {Tannu}}, \bibinfo {author}
  {\bibfnamefont {N.}~\bibnamefont {Shannon}},\ and\ \bibinfo {author}
  {\bibfnamefont {R.}~\bibnamefont {Joynt}},\ }\bibfield  {title} {\bibinfo
  {title} {Exploiting many-body localization for scalable variational quantum
  simulation},\ }\href {https://arxiv.org/abs/2404.17560} {\bibfield  {journal}
  {\bibinfo  {journal} {arXiv preprint arXiv:2404.17560}\ } (\bibinfo {year}
  {2024})}\BibitemShut {NoStop}%
\bibitem [{\citenamefont {Alam}\ \emph {et~al.}(2022)\citenamefont {Alam},
  \citenamefont {Wudarski}, \citenamefont {Reagor}, \citenamefont {Sud},
  \citenamefont {Grabbe}, \citenamefont {Wang}, \citenamefont {Hodson},
  \citenamefont {Lott}, \citenamefont {Rieffel},\ and\ \citenamefont
  {Venturelli}}]{Alam2022pra}%
  \BibitemOpen
  \bibfield  {author} {\bibinfo {author} {\bibfnamefont {M.~S.}\ \bibnamefont
  {Alam}}, \bibinfo {author} {\bibfnamefont {F.~A.}\ \bibnamefont {Wudarski}},
  \bibinfo {author} {\bibfnamefont {M.~J.}\ \bibnamefont {Reagor}}, \bibinfo
  {author} {\bibfnamefont {J.}~\bibnamefont {Sud}}, \bibinfo {author}
  {\bibfnamefont {S.}~\bibnamefont {Grabbe}}, \bibinfo {author} {\bibfnamefont
  {Z.}~\bibnamefont {Wang}}, \bibinfo {author} {\bibfnamefont {M.}~\bibnamefont
  {Hodson}}, \bibinfo {author} {\bibfnamefont {P.~A.}\ \bibnamefont {Lott}},
  \bibinfo {author} {\bibfnamefont {E.~G.}\ \bibnamefont {Rieffel}},\ and\
  \bibinfo {author} {\bibfnamefont {D.}~\bibnamefont {Venturelli}},\ }\bibfield
   {title} {\bibinfo {title} {Practical verification of quantum properties in
  quantum-approximate-optimization runs},\ }\href
  {https://doi.org/10.1103/PhysRevApplied.17.024026} {\bibfield  {journal}
  {\bibinfo  {journal} {Phys. Rev. Appl.}\ }\textbf {\bibinfo {volume} {17}},\
  \bibinfo {pages} {024026} (\bibinfo {year} {2022})}\BibitemShut {NoStop}%
\bibitem [{\citenamefont {Lewis}\ \emph {et~al.}(2024)\citenamefont {Lewis},
  \citenamefont {Huang}, \citenamefont {Tran}, \citenamefont {Lehner},
  \citenamefont {Kueng},\ and\ \citenamefont {Preskill}}]{lewis2024improved}%
  \BibitemOpen
  \bibfield  {author} {\bibinfo {author} {\bibfnamefont {L.}~\bibnamefont
  {Lewis}}, \bibinfo {author} {\bibfnamefont {H.-Y.}\ \bibnamefont {Huang}},
  \bibinfo {author} {\bibfnamefont {V.~T.}\ \bibnamefont {Tran}}, \bibinfo
  {author} {\bibfnamefont {S.}~\bibnamefont {Lehner}}, \bibinfo {author}
  {\bibfnamefont {R.}~\bibnamefont {Kueng}},\ and\ \bibinfo {author}
  {\bibfnamefont {J.}~\bibnamefont {Preskill}},\ }\bibfield  {title} {\bibinfo
  {title} {Improved machine learning algorithm for predicting ground state
  properties},\ }\href {https://doi.org/10.1038/s41467-024-45014-7} {\bibfield
  {journal} {\bibinfo  {journal} {Nat. Commun.}\ }\textbf {\bibinfo {volume}
  {15}},\ \bibinfo {pages} {895} (\bibinfo {year} {2024})}\BibitemShut
  {NoStop}%
\bibitem [{\citenamefont {McArdle}\ \emph {et~al.}(2020)\citenamefont
  {McArdle}, \citenamefont {Endo}, \citenamefont {Aspuru-Guzik}, \citenamefont
  {Benjamin},\ and\ \citenamefont {Yuan}}]{McArdle2020}%
  \BibitemOpen
  \bibfield  {author} {\bibinfo {author} {\bibfnamefont {S.}~\bibnamefont
  {McArdle}}, \bibinfo {author} {\bibfnamefont {S.}~\bibnamefont {Endo}},
  \bibinfo {author} {\bibfnamefont {A.}~\bibnamefont {Aspuru-Guzik}}, \bibinfo
  {author} {\bibfnamefont {S.~C.}\ \bibnamefont {Benjamin}},\ and\ \bibinfo
  {author} {\bibfnamefont {X.}~\bibnamefont {Yuan}},\ }\bibfield  {title}
  {\bibinfo {title} {Quantum computational chemistry},\ }\href
  {https://doi.org/10.1103/RevModPhys.92.015003} {\bibfield  {journal}
  {\bibinfo  {journal} {Rev. Mod. Phys.}\ }\textbf {\bibinfo {volume} {92}},\
  \bibinfo {pages} {015003} (\bibinfo {year} {2020})}\BibitemShut {NoStop}%
\bibitem [{\citenamefont {Cao}\ \emph {et~al.}(2019)\citenamefont {Cao},
  \citenamefont {Romero}, \citenamefont {Olson}, \citenamefont {Degroote},
  \citenamefont {Johnson}, \citenamefont {Kieferov{\'a}}, \citenamefont
  {Kivlichan}, \citenamefont {Menke}, \citenamefont {Peropadre}, \citenamefont
  {Sawaya} \emph {et~al.}}]{cao2019quantum}%
  \BibitemOpen
  \bibfield  {author} {\bibinfo {author} {\bibfnamefont {Y.}~\bibnamefont
  {Cao}}, \bibinfo {author} {\bibfnamefont {J.}~\bibnamefont {Romero}},
  \bibinfo {author} {\bibfnamefont {J.~P.}\ \bibnamefont {Olson}}, \bibinfo
  {author} {\bibfnamefont {M.}~\bibnamefont {Degroote}}, \bibinfo {author}
  {\bibfnamefont {P.~D.}\ \bibnamefont {Johnson}}, \bibinfo {author}
  {\bibfnamefont {M.}~\bibnamefont {Kieferov{\'a}}}, \bibinfo {author}
  {\bibfnamefont {I.~D.}\ \bibnamefont {Kivlichan}}, \bibinfo {author}
  {\bibfnamefont {T.}~\bibnamefont {Menke}}, \bibinfo {author} {\bibfnamefont
  {B.}~\bibnamefont {Peropadre}}, \bibinfo {author} {\bibfnamefont {N.~P.}\
  \bibnamefont {Sawaya}}, \emph {et~al.},\ }\bibfield  {title} {\bibinfo
  {title} {Quantum chemistry in the age of quantum computing},\ }\href
  {https://doi.org/10.1021/acs.chemrev.8b00803} {\bibfield  {journal} {\bibinfo
   {journal} {Chemical reviews}\ }\textbf {\bibinfo {volume} {119}},\ \bibinfo
  {pages} {10856} (\bibinfo {year} {2019})}\BibitemShut {NoStop}%
\bibitem [{\citenamefont {Lucas}(2014)}]{Lucas2014}%
  \BibitemOpen
  \bibfield  {author} {\bibinfo {author} {\bibfnamefont {A.}~\bibnamefont
  {Lucas}},\ }\bibfield  {title} {\bibinfo {title} {Ising formulations of many
  {NP} problems},\ }\bibfield  {journal} {\bibinfo  {journal} {Frontiers in
  Physics}\ }\textbf {\bibinfo {volume} {2}},\ \href
  {https://doi.org/10.3389/fphy.2014.00005} {10.3389/fphy.2014.00005} (\bibinfo
  {year} {2014})\BibitemShut {NoStop}%
\bibitem [{\citenamefont {Abadi}\ \emph {et~al.}(2016)\citenamefont {Abadi},
  \citenamefont {Barham}, \citenamefont {Chen}, \citenamefont {Chen},
  \citenamefont {Davis}, \citenamefont {Dean}, \citenamefont {Devin},
  \citenamefont {Ghemawat}, \citenamefont {Irving}, \citenamefont {Isard},
  \citenamefont {Kudlur}, \citenamefont {Levenberg}, \citenamefont {Monga},
  \citenamefont {Moore}, \citenamefont {Murray}, \citenamefont {Steiner},
  \citenamefont {Tucker}, \citenamefont {Vasudevan}, \citenamefont {Warden},
  \citenamefont {Wicke}, \citenamefont {Yu},\ and\ \citenamefont
  {Zheng}}]{tensorflow2016}%
  \BibitemOpen
  \bibfield  {author} {\bibinfo {author} {\bibfnamefont {M.}~\bibnamefont
  {Abadi}}, \bibinfo {author} {\bibfnamefont {P.}~\bibnamefont {Barham}},
  \bibinfo {author} {\bibfnamefont {J.}~\bibnamefont {Chen}}, \bibinfo {author}
  {\bibfnamefont {Z.}~\bibnamefont {Chen}}, \bibinfo {author} {\bibfnamefont
  {A.}~\bibnamefont {Davis}}, \bibinfo {author} {\bibfnamefont
  {J.}~\bibnamefont {Dean}}, \bibinfo {author} {\bibfnamefont {M.}~\bibnamefont
  {Devin}}, \bibinfo {author} {\bibfnamefont {S.}~\bibnamefont {Ghemawat}},
  \bibinfo {author} {\bibfnamefont {G.}~\bibnamefont {Irving}}, \bibinfo
  {author} {\bibfnamefont {M.}~\bibnamefont {Isard}}, \bibinfo {author}
  {\bibfnamefont {M.}~\bibnamefont {Kudlur}}, \bibinfo {author} {\bibfnamefont
  {J.}~\bibnamefont {Levenberg}}, \bibinfo {author} {\bibfnamefont
  {R.}~\bibnamefont {Monga}}, \bibinfo {author} {\bibfnamefont
  {S.}~\bibnamefont {Moore}}, \bibinfo {author} {\bibfnamefont {D.~G.}\
  \bibnamefont {Murray}}, \bibinfo {author} {\bibfnamefont {B.}~\bibnamefont
  {Steiner}}, \bibinfo {author} {\bibfnamefont {P.}~\bibnamefont {Tucker}},
  \bibinfo {author} {\bibfnamefont {V.}~\bibnamefont {Vasudevan}}, \bibinfo
  {author} {\bibfnamefont {P.}~\bibnamefont {Warden}}, \bibinfo {author}
  {\bibfnamefont {M.}~\bibnamefont {Wicke}}, \bibinfo {author} {\bibfnamefont
  {Y.}~\bibnamefont {Yu}},\ and\ \bibinfo {author} {\bibfnamefont
  {X.}~\bibnamefont {Zheng}},\ }\bibfield  {title} {\bibinfo {title}
  {{TensorFlow}: a system for large-scale machine learning},\ }in\ \href@noop
  {} {\emph {\bibinfo {booktitle} {Proceedings of the 12th USENIX Conference on
  Operating Systems Design and Implementation}}},\ \bibinfo {series and number}
  {OSDI'16}\ (\bibinfo  {publisher} {USENIX Association},\ \bibinfo {address}
  {USA},\ \bibinfo {year} {2016})\ p.\ \bibinfo {pages} {265–283}\BibitemShut
  {NoStop}%
\bibitem [{\citenamefont {Paszke}\ \emph {et~al.}(2019)\citenamefont {Paszke},
  \citenamefont {Gross}, \citenamefont {Massa}, \citenamefont {Lerer},
  \citenamefont {Bradbury}, \citenamefont {Chanan}, \citenamefont {Killeen},
  \citenamefont {Lin}, \citenamefont {Gimelshein}, \citenamefont {Antiga},
  \citenamefont {Desmaison}, \citenamefont {Kopf}, \citenamefont {Yang},
  \citenamefont {DeVito}, \citenamefont {Raison}, \citenamefont {Tejani},
  \citenamefont {Chilamkurthy}, \citenamefont {Steiner}, \citenamefont {Fang},
  \citenamefont {Bai},\ and\ \citenamefont {Chintala}}]{pytorch2019}%
  \BibitemOpen
  \bibfield  {author} {\bibinfo {author} {\bibfnamefont {A.}~\bibnamefont
  {Paszke}}, \bibinfo {author} {\bibfnamefont {S.}~\bibnamefont {Gross}},
  \bibinfo {author} {\bibfnamefont {F.}~\bibnamefont {Massa}}, \bibinfo
  {author} {\bibfnamefont {A.}~\bibnamefont {Lerer}}, \bibinfo {author}
  {\bibfnamefont {J.}~\bibnamefont {Bradbury}}, \bibinfo {author}
  {\bibfnamefont {G.}~\bibnamefont {Chanan}}, \bibinfo {author} {\bibfnamefont
  {T.}~\bibnamefont {Killeen}}, \bibinfo {author} {\bibfnamefont
  {Z.}~\bibnamefont {Lin}}, \bibinfo {author} {\bibfnamefont {N.}~\bibnamefont
  {Gimelshein}}, \bibinfo {author} {\bibfnamefont {L.}~\bibnamefont {Antiga}},
  \bibinfo {author} {\bibfnamefont {A.}~\bibnamefont {Desmaison}}, \bibinfo
  {author} {\bibfnamefont {A.}~\bibnamefont {Kopf}}, \bibinfo {author}
  {\bibfnamefont {E.}~\bibnamefont {Yang}}, \bibinfo {author} {\bibfnamefont
  {Z.}~\bibnamefont {DeVito}}, \bibinfo {author} {\bibfnamefont
  {M.}~\bibnamefont {Raison}}, \bibinfo {author} {\bibfnamefont
  {A.}~\bibnamefont {Tejani}}, \bibinfo {author} {\bibfnamefont
  {S.}~\bibnamefont {Chilamkurthy}}, \bibinfo {author} {\bibfnamefont
  {B.}~\bibnamefont {Steiner}}, \bibinfo {author} {\bibfnamefont
  {L.}~\bibnamefont {Fang}}, \bibinfo {author} {\bibfnamefont {J.}~\bibnamefont
  {Bai}},\ and\ \bibinfo {author} {\bibfnamefont {S.}~\bibnamefont
  {Chintala}},\ }\bibfield  {title} {\bibinfo {title} {{PyTorch}: An imperative
  style, high-performance deep learning library},\ }in\ \href
  {https://proceedings.neurips.cc/paper_files/paper/2019/file/bdbca288fee7f92f2bfa9f7012727740-Paper.pdf}
  {\emph {\bibinfo {booktitle} {Advances in Neural Information Processing
  Systems}}},\ Vol.~\bibinfo {volume} {32},\ \bibinfo {editor} {edited by\
  \bibinfo {editor} {\bibfnamefont {H.}~\bibnamefont {Wallach}}, \bibinfo
  {editor} {\bibfnamefont {H.}~\bibnamefont {Larochelle}}, \bibinfo {editor}
  {\bibfnamefont {A.}~\bibnamefont {Beygelzimer}}, \bibinfo {editor}
  {\bibfnamefont {F.}~\bibnamefont {d\textquotesingle Alch\'{e}-Buc}}, \bibinfo
  {editor} {\bibfnamefont {E.}~\bibnamefont {Fox}},\ and\ \bibinfo {editor}
  {\bibfnamefont {R.}~\bibnamefont {Garnett}}}\ (\bibinfo  {publisher} {Curran
  Associates, Inc.},\ \bibinfo {year} {2019})\BibitemShut {NoStop}%
\bibitem [{\citenamefont {Zhang}\ \emph {et~al.}(2023)\citenamefont {Zhang},
  \citenamefont {Allcock}, \citenamefont {Wan}, \citenamefont {Liu},
  \citenamefont {Sun}, \citenamefont {Yu}, \citenamefont {Yang}, \citenamefont
  {Qiu}, \citenamefont {Ye}, \citenamefont {Chen}, \citenamefont {Lee},
  \citenamefont {Zheng}, \citenamefont {Jian}, \citenamefont {Yao},
  \citenamefont {Hsieh},\ and\ \citenamefont {Zhang}}]{tensorcircuit2023}%
  \BibitemOpen
  \bibfield  {author} {\bibinfo {author} {\bibfnamefont {S.-X.}\ \bibnamefont
  {Zhang}}, \bibinfo {author} {\bibfnamefont {J.}~\bibnamefont {Allcock}},
  \bibinfo {author} {\bibfnamefont {Z.-Q.}\ \bibnamefont {Wan}}, \bibinfo
  {author} {\bibfnamefont {S.}~\bibnamefont {Liu}}, \bibinfo {author}
  {\bibfnamefont {J.}~\bibnamefont {Sun}}, \bibinfo {author} {\bibfnamefont
  {H.}~\bibnamefont {Yu}}, \bibinfo {author} {\bibfnamefont {X.-H.}\
  \bibnamefont {Yang}}, \bibinfo {author} {\bibfnamefont {J.}~\bibnamefont
  {Qiu}}, \bibinfo {author} {\bibfnamefont {Z.}~\bibnamefont {Ye}}, \bibinfo
  {author} {\bibfnamefont {Y.-Q.}\ \bibnamefont {Chen}}, \bibinfo {author}
  {\bibfnamefont {C.-K.}\ \bibnamefont {Lee}}, \bibinfo {author} {\bibfnamefont
  {Y.-C.}\ \bibnamefont {Zheng}}, \bibinfo {author} {\bibfnamefont {S.-K.}\
  \bibnamefont {Jian}}, \bibinfo {author} {\bibfnamefont {H.}~\bibnamefont
  {Yao}}, \bibinfo {author} {\bibfnamefont {C.-Y.}\ \bibnamefont {Hsieh}},\
  and\ \bibinfo {author} {\bibfnamefont {S.}~\bibnamefont {Zhang}},\ }\bibfield
   {title} {\bibinfo {title} {{TensorCircuit}: a quantum software framework for
  the {NISQ} era},\ }\href {https://doi.org/10.22331/q-2023-02-02-912}
  {\bibfield  {journal} {\bibinfo  {journal} {Quantum}\ }\textbf {\bibinfo
  {volume} {7}},\ \bibinfo {pages} {912} (\bibinfo {year} {2023})}\BibitemShut
  {NoStop}%
\bibitem [{\citenamefont {Kingma}\ and\ \citenamefont {Ba}(2014)}]{Kingma2014}%
  \BibitemOpen
  \bibfield  {author} {\bibinfo {author} {\bibfnamefont {D.~P.}\ \bibnamefont
  {Kingma}}\ and\ \bibinfo {author} {\bibfnamefont {J.}~\bibnamefont {Ba}},\
  }\bibfield  {title} {\bibinfo {title} {Adam: A method for stochastic
  optimization},\ }\href {https://doi.org/10.48550/arXiv.1412.6980} {\bibfield
  {journal} {\bibinfo  {journal} {arXiv preprint arXiv:1412.6980}\ } (\bibinfo
  {year} {2014})}\BibitemShut {NoStop}%
\bibitem [{\citenamefont {Loshchilov}\ and\ \citenamefont
  {Hutter}(2019)}]{Loshchilov2019}%
  \BibitemOpen
  \bibfield  {author} {\bibinfo {author} {\bibfnamefont {I.}~\bibnamefont
  {Loshchilov}}\ and\ \bibinfo {author} {\bibfnamefont {F.}~\bibnamefont
  {Hutter}},\ }\bibfield  {title} {\bibinfo {title} {Decoupled weight decay
  regularization},\ }in\ \href {https://openreview.net/forum?id=Bkg6RiCqY7}
  {\emph {\bibinfo {booktitle} {International Conference on Learning
  Representations}}}\ (\bibinfo {year} {2019})\BibitemShut {NoStop}%
\bibitem [{\citenamefont {Wang}\ \emph {et~al.}(2024)\citenamefont {Wang},
  \citenamefont {Qi}, \citenamefont {Wang},\ and\ \citenamefont
  {Dong}}]{wang2024eha}%
  \BibitemOpen
  \bibfield  {author} {\bibinfo {author} {\bibfnamefont {X.}~\bibnamefont
  {Wang}}, \bibinfo {author} {\bibfnamefont {B.}~\bibnamefont {Qi}}, \bibinfo
  {author} {\bibfnamefont {Y.}~\bibnamefont {Wang}},\ and\ \bibinfo {author}
  {\bibfnamefont {D.}~\bibnamefont {Dong}},\ }\bibfield  {title} {\bibinfo
  {title} {Entanglement-variational hardware-efficient ansatz for
  eigensolvers},\ }\href {https://doi.org/10.1103/PhysRevApplied.21.034059}
  {\bibfield  {journal} {\bibinfo  {journal} {Phys. Rev. Appl.}\ }\textbf
  {\bibinfo {volume} {21}},\ \bibinfo {pages} {034059} (\bibinfo {year}
  {2024})}\BibitemShut {NoStop}%
\bibitem [{\citenamefont {Bergholm}\ \emph {et~al.}(2018)\citenamefont
  {Bergholm}, \citenamefont {Izaac}, \citenamefont {Schuld}, \citenamefont
  {Gogolin}, \citenamefont {Ahmed}, \citenamefont {Ajith}, \citenamefont
  {Alam}, \citenamefont {Alonso-Linaje}, \citenamefont {AkashNarayanan},
  \citenamefont {Asadi} \emph {et~al.}}]{pennylane2018}%
  \BibitemOpen
  \bibfield  {author} {\bibinfo {author} {\bibfnamefont {V.}~\bibnamefont
  {Bergholm}}, \bibinfo {author} {\bibfnamefont {J.}~\bibnamefont {Izaac}},
  \bibinfo {author} {\bibfnamefont {M.}~\bibnamefont {Schuld}}, \bibinfo
  {author} {\bibfnamefont {C.}~\bibnamefont {Gogolin}}, \bibinfo {author}
  {\bibfnamefont {S.}~\bibnamefont {Ahmed}}, \bibinfo {author} {\bibfnamefont
  {V.}~\bibnamefont {Ajith}}, \bibinfo {author} {\bibfnamefont {M.~S.}\
  \bibnamefont {Alam}}, \bibinfo {author} {\bibfnamefont {G.}~\bibnamefont
  {Alonso-Linaje}}, \bibinfo {author} {\bibfnamefont {B.}~\bibnamefont
  {AkashNarayanan}}, \bibinfo {author} {\bibfnamefont {A.}~\bibnamefont
  {Asadi}}, \emph {et~al.},\ }\bibfield  {title} {\bibinfo {title} {Pennylane:
  Automatic differentiation of hybrid quantum-classical computations},\ }\href
  {https://doi.org/10.48550/arXiv.1811.04968} {\bibfield  {journal} {\bibinfo
  {journal} {arXiv preprint arXiv:1811.04968}\ } (\bibinfo {year}
  {2018})}\BibitemShut {NoStop}%
\bibitem [{\citenamefont {Wright}\ and\ \citenamefont
  {Ma}(2022)}]{wright2022high}%
  \BibitemOpen
  \bibfield  {author} {\bibinfo {author} {\bibfnamefont {J.}~\bibnamefont
  {Wright}}\ and\ \bibinfo {author} {\bibfnamefont {Y.}~\bibnamefont {Ma}},\
  }\href {https://doi.org/10.1017/9781108779302} {\emph {\bibinfo {title}
  {High-dimensional data analysis with low-dimensional models: Principles,
  computation, and applications}}}\ (\bibinfo  {publisher} {Cambridge
  University Press},\ \bibinfo {year} {2022})\BibitemShut {NoStop}%
\end{thebibliography}%

\onecolumngrid
\appendix
	
\section{\label{intuition}The norm of gradient}
In this section, we explain the meaning of the figure of merit $\Vert\nabla_{\boldsymbol{\theta}}{C}\Vert_{2}^{2}$.
	
In gradient-based optimizers,  the parameter vector $\boldsymbol{\theta}$ updates as
	\begin{equation}\nonumber
	\boldsymbol{\theta}_{t+1}=\boldsymbol{\theta}_{t} - \eta\nabla_{\boldsymbol{\theta}}{C\bm{(}\boldsymbol{\theta}_t\bm{)}},
	\end{equation}
where $\eta$ denotes the learning rate, which is usually small. With the update, the value of cost function decreases as
	\begin{equation}\nonumber
	C\left(\boldsymbol{\theta}_{t+1}\right) =C\left(\boldsymbol{\theta}_{t} - \eta\nabla_{\boldsymbol{\theta}}{C(\boldsymbol{\theta}_t)} \right)=C\left(\boldsymbol{\theta}_{t} \right) - \eta \Vert\nabla_{\boldsymbol{\theta}}{C(\boldsymbol{\theta}_t)}\Vert_{2}^{2} +o\left(\eta \right).
	\end{equation}
	
It is clear to see that the decrease with parameters updating is quantified in terms of  the magnitude of $\Vert\nabla_{\boldsymbol{\theta}}{C}\Vert_{2}^{2}$.  Thus, we adopt the average $\mathop{\mathbb{E}}\limits_{\boldsymbol{\theta}}\Vert\nabla_{\boldsymbol{\theta}}{C}\Vert_{2}^{2}$ as the  trainability figure of merit.

\section{\label{lemmas}Technical Lemmas}
In this section, we present several technical lemmas, which are all based on the following assumption.
	
\begin{assumption}
	A random variable $ \theta $ is drawn from  uniform distribution over $\left[ -a\pi,a\pi\right]$, where $a \in \left( 0, 1\right) $ is a hyperparameter to be further determined.
\end{assumption}
	
Under the assumption, let  $\mathop{\mathbb{E}}\limits_{\theta}$ denote the expectation over $\theta$ with respect to the uniform distribution $\mathcal{U}\left[ -a\pi,a\pi\right]$. Recall that
	\begin{equation}\nonumber
	\alpha=\frac{1}{4}\left( 2+\frac{\sin {2a\pi}}{a\pi}\right),\,\beta=\frac{1}{4}\left( 2-\frac{\sin {2a\pi}}{a\pi}\right),\, \gamma=\frac{\sin {a\pi}}{a\pi}.
	\end{equation}
	
\begin{lemma}\label{firstpowerlemma}
	Let $A$ be an arbitrary linear operator, $G$ be a Hermitian unitary and $V=e^{-i \frac{\theta}{2}G}$. Let $O$ be an arbitrary Hermitian operator that anti-commutes with $G$, and $\tilde{O}$ be an arbitrary Hermitian operator that commutes with $G$. Then,
	\begin{align}
		\mathop{\mathbb{E}}\limits_{\theta}\mathrm{Tr}\left[ OVAV^{\dagger}\right] &=\gamma\, \mathrm{Tr}\left[ OA\right] ,\label{Efirstpower1}
		\\
		\mathop{\mathbb{E}}\limits_{\theta} \frac{\partial}{\partial\theta} \mathrm{Tr}\left[ OVAV^{\dagger}\right] &=\gamma\,\mathrm{Tr}\left[ iGOA\right] ,\label{Efirstpower2}
		\\
		\mathop{\mathbb{E}}\limits_{\theta}\mathrm{Tr}\left[ \tilde{O}VAV^{\dagger}\right] &=\mathrm{Tr}\left[ \tilde{O}A\right],\label{Efirstpower3}
		\\
		\mathop{\mathbb{E}}\limits_{\theta} \frac{\partial}{\partial\theta} \mathrm{Tr}\left[ \tilde{O}VAV^{\dagger}\right] &=0.\label{Efirstpower4}
		\end{align}
\end{lemma}
\begin{proof}
From  $V=e^{-i \frac{\theta}{2}G}=I\cos{\frac{\theta}{2}}-iG\sin{\frac{\theta}{2}}$, for any operator $B$, we have
	\begin{align}
		\mathrm{Tr}\left[ BVAV^{\dagger}\right]\nonumber=&\,\mathrm{Tr}\left[B(I\cos{\frac{\theta}{2}}-iG\sin{\frac{\theta}{2}})A(I\cos{\frac{\theta}{2}}+iG\sin{\frac{\theta}{2}})\right]\nonumber
		\\
		\label{key0}
		=&\frac{1+\cos{\theta}}{2} \mathrm{Tr}\left[ BA\right] + \frac{1-\cos{\theta}}{2} \mathrm{Tr}\left[ BGAG\right] +\frac{\sin{\theta}}{2} \left(\mathrm{Tr}\left[ iBAG\right] -\mathrm{Tr}\left[ iBGA\right]  \right).
	\end{align}
		
If $O$ anti-commutes with $G$, then $iGO$ is Hermitian. In addition, with $G$ being unitary and $\{O, G\}=0$, from Eq.~\eqref{key0}  we have
	\begin{equation}\label{key01}
		\mathrm{Tr}\left[ OVAV^{\dagger}\right]=\cos{\theta} \mathrm{Tr}\left[ OA\right] + \sin{\theta} \mathrm{Tr}\left[ iGOA\right],
		\end{equation}
		and accordingly,
		\begin{equation}\label{key02}
		\frac{\partial}{\partial\theta}\mathrm{Tr}\left[ OVAV^{\dagger}\right]=-\sin{\theta} \mathrm{Tr}\left[ OA\right] + \cos{\theta} \mathrm{Tr}\left[ iGOA\right].
		\end{equation}
Thus, Eqs.~\eqref{Efirstpower1} and \eqref{Efirstpower2}  can be obtained by taking the expectations of Eqs.~\eqref{key01} and \eqref{key02} over $\theta$, respectively.
		
Moreover,  when $\tilde{O}$ commutes with the unitary $G$, following similar calculations, it is straightforward to have Eqs.~\eqref{Efirstpower3} and \eqref{Efirstpower4}.
	\end{proof}
	
From Lemma~\ref{firstpowerlemma}, we have the following conclusion.

\begin{corollary}\label{firstpowercorollary}
Let $A$ be an arbitrary linear operator,  $G$ be a Hermitian unitary and $V=e^{-i\frac{\theta}{2}G}$. Let $O$ denote an arbitrary Hermitian operator that anti-commutes with $G$, and  $\tilde{O}_{1}$ and  $\tilde{O}_{2}$ be arbitrary Hermitian operators that commute with $G$. Then,
	\begin{align}
		\mathop{\mathbb{E}}\limits_{\theta}\mathrm{Tr}\left[ OVAV^{\dagger}\right]  \mathrm{Tr}\left[ \tilde{O}_{1}VAV^{\dagger}\right]
		&=\gamma\,\mathrm{Tr}\left[ OA\right] \mathrm{Tr}\left[ \tilde{O}_{1}A\right],\label{firstpowercorollary1}
		\\
		\mathop{\mathbb{E}}\limits_{\theta}\frac{\partial}{\partial\theta} \mathrm{Tr}\left[ OVAV^{\dagger}\right] \cdot \frac{\partial}{\partial\theta}\mathrm{Tr}\left[ \tilde{O}_{1}VAV^{\dagger}\right] &=0 \label{firstpowercorollary2},
		\\
		\mathop{\mathbb{E}}\limits_{\theta}\mathrm{Tr}\left[ \tilde{O}_{1}VAV^{\dagger}\right] \mathrm{Tr}\left[ \tilde{O}_{2}VAV^{\dagger}\right]
		&=\mathrm{Tr}\left[ \tilde{O}_{1}A\right]\mathrm{Tr}\left[ \tilde{O}_{2}A\right],\label{firstpowercorollary3}
		\\
		\mathop{\mathbb{E}}\limits_{\theta}\frac{\partial}{\partial\theta} \mathrm{Tr}\left[ \tilde{O}_{1}VAV^{\dagger}\right] \cdot\frac{\partial}{\partial\theta}\mathrm{Tr}\left[ \tilde{O}_{2}VAV^{\dagger}\right] &=0\label{firstpowercorollary4}.
		\end{align}
	\end{corollary}

\begin{lemma}\label{secondpowerlemma}
Let $A$ be an arbitrary linear operator,  $G$ be a Hermitian unitary and $V=e^{-i\frac{\theta}{2}G}$. Let $O_{1}$ and $O_{2}$ be arbitrary Hermitian operators that anti-commute with $G$. Then,
	\begin{align}
		\label{secondpowerlemma1}
		\mathop{\mathbb{E}}\limits_{\theta}\mathrm{Tr}\left[ O_{1}VAV^{\dagger}\right]  \mathrm{Tr}\left[ O_{2}VAV^{\dagger}\right]&=\,\alpha\,\mathrm{Tr}\left[ O_{1}A\right] \mathrm{Tr}\left[ {O}_{2}A\right]+\beta\, \mathrm{Tr}\left[ iGO_{1}A\right] \mathrm{Tr}\left[ iG{O}_{2}A\right],
		\\
		\label{secondpowerlemma2}
		\mathop{\mathbb{E}}\limits_{\theta}\frac{\partial}{\partial\theta}\mathrm{Tr}\left[ O_{1}VAV^{\dagger}\right] \cdot \frac{\partial}{\partial\theta}\mathrm{Tr}\left[ O_{2}VAV^{\dagger}\right]
		&=\,\beta\, \mathrm{Tr}\left[ O_{1}A\right] \mathrm{Tr}\left[ {O}_{2}A\right]+\alpha\, \mathrm{Tr}\left[ iGO_{1}A\right] \mathrm{Tr}\left[ iG{O}_{2}A\right].
		\end{align}
	\end{lemma}
\begin{proof}
According to Eqs.~\eqref{key01} and \eqref{key02}, we have
		\begin{align}
		\begin{split}\label{secondpower1}
		\mathrm{Tr}\left[ O_{1}VAV^{\dagger}\right] \cdot \mathrm{Tr}\left[ O_{2}VAV^{\dagger}\right]
		=&\cos^{2}{\theta}\mathrm{Tr}\left[ O_{1}A\right] \mathrm{Tr}\left[ {O}_{2}A\right]
		+\sin^{2}{\theta} \mathrm{Tr}\left[ iGO_{1}A\right] \mathrm{Tr}\left[ iG{O}_{2}A\right]
		\\
		&+\sin{\theta}\cos{\theta}\left\lbrace\mathrm{Tr}\left[ O_{1}A\right]\mathrm{Tr}\left[ iG{O}_{2}A\right]+\mathrm{Tr}\left[ {O}_{2}A\right]\mathrm{Tr}\left[ iGO_{1}A\right]\right\rbrace,
		\end{split}
		\\
		\begin{split}\label{secondpower2}
		\frac{\partial}{\partial\theta}\mathrm{Tr}\left[ O_{1}VAV^{\dagger}\right] \cdot \frac{\partial}{\partial\theta}\mathrm{Tr}\left[ O_{2}VAV^{\dagger}\right]
		=&\sin^{2}\theta\mathrm{Tr}\left[ O_{1}A\right] \mathrm{Tr}\left[ {O}_{2}A\right]
		+\cos^{2}\theta \mathrm{Tr}\left[ iGO_{1}A\right] \mathrm{Tr}\left[ iG{O}_{2}A\right]
		\\
		&-\sin{\theta}\cos{\theta}\left\lbrace\mathrm{Tr}\left[ O_{1}A\right]\mathrm{Tr}\left[ iG{O}_{2}A\right]+\mathrm{Tr}\left[ {O}_{2}A\right]\mathrm{Tr}\left[ iGO_{1}A\right] \right\rbrace.
		\end{split}
		\end{align}
Then, Eqs.~\eqref{secondpowerlemma1} and \eqref{secondpowerlemma2} can be derived by taking the expectations of Eqs.~\eqref{secondpower1} and \eqref{secondpower2} and employing	
	\begin{equation}\nonumber
		\mathop{\mathbb{E}}\limits_{\theta}\cos^{2}{\theta}=\alpha,\
		\mathop{\mathbb{E}}\limits_{\theta}\sin^{2}{\theta}=\beta,\
		\mathop{\mathbb{E}}\limits_{\theta}\sin{\theta}\cos{\theta}=0.
		\end{equation}
	\end{proof}
It follows that the following conclusion holds.

\begin{corollary}\label{secondpowercorollary}
Let $A$ be an arbitrary linear operator,  $G$ be a Hermitian unitary and $V=e^{-i\frac{\theta}{2}G}$. Let $O$ be an arbitrary Hermitian operator that anti-commutes with $G$. Then,
	\begin{align}
		\mathop{\mathbb{E}}\limits_{\theta}{\mathrm{Tr}^{2}\left[ OVAV^{\dagger}\right]}&=\alpha\, {\mathrm{Tr}^{2}\left[ OA\right]}+\beta\, {\mathrm{Tr}^{2}\left[ iGOA\right]},\label{secondpowercorollary1}
		\\
		\mathop{\mathbb{E}}\limits_{\theta}\left( \frac{\partial}{\partial\theta}{\mathrm{Tr}\left[ OVAV^{\dagger}\right]}\right) ^{2}&=\beta\, \mathrm{Tr}^{2}\left[ OA\right]+\alpha\, {\mathrm{Tr}^{2}\left[ iGOA\right]},\label{secondpowercorollary2}
		\end{align}
		where $\mathrm{Tr}^{2}[\cdot]=(\mathrm{Tr}[\cdot])^2.$
	\end{corollary}
	
The following lemma concerns the transformations of two-qubit Pauli tensor products after a controlled-$Z$ gate.
	\begin{lemma}\label{CZlemma}
		Let ${CZ}$ denote a controlled-$Z$ gate, and $\sigma_i\otimes\sigma_j$ be a two-qubit Pauli tensor product, where $\sigma_i$ and $\sigma_j$ are Pauli matrices.  If $\sigma_{i^{\prime}}\otimes\sigma_{j^{\prime}}=CZ^{\dagger} (\sigma_i\otimes \sigma_j)CZ$, we denote the transformation by
		$ \sigma_i\otimes \sigma_j \rightarrow \sigma_{i^{\prime}}\otimes\sigma_{j^{\prime}}$. All of the specific transformations can be summarized as follows:
		\begin{align}
			&Z \otimes I \leftrightarrow Z \otimes I,\ \  Z \otimes Z \leftrightarrow Z \otimes Z, \nonumber\\
			&Y \otimes I \leftrightarrow Y \otimes Z,\ \  Y \otimes Y \leftrightarrow X \otimes X, \nonumber\\
			&X \otimes I \leftrightarrow X \otimes Z,\ \  X \otimes Y \leftrightarrow -Y \otimes X.\nonumber
		\end{align}
	\end{lemma}
	
\section{\label{proofTH1}Proof of Lemma~\ref{TH1}}
\begin{proof}
For the  Hamiltonian of Eq.~\eqref{toyH}, the average of the squared norm of the gradient can be described as
	\begin{align}\nonumber
		\mathop{\mathbb{E}}\limits_{\boldsymbol{\theta}}\Vert\nabla_{\boldsymbol{\theta}}{C}\Vert_{2}^{2}
		&=\sum_{q=1}^{2L}{\sum_{n=1}^{N}{ \mathop{\mathbb{E}}\limits_{\boldsymbol{\theta}}\left(\frac{\partial C}{\partial{\theta_{q, n}}}\right)^{2}}}
		\\
		&=\sum_{q=1}^{2L} {\sum_{n=1}^{N} {\mathop{\mathbb{E}}\limits_{\boldsymbol{\theta}}\left(\sum_{i=1}^{N-1}{\frac{\partial C_{i}}{\partial{\theta_{q, n}}}}\right)^{2}}}\nonumber
		\\
		\label{totalsum}
		&=\sum_{i=1}^{N-1} {\sum_{q=1}^{2L} {\sum_{n=1}^{N} {\mathop{\mathbb{E}}\limits_{\boldsymbol{\theta}}\left(\frac{\partial C_{i}}{\partial{\theta_{q, n}}}\right)^{2}}}}+\sum_{i\ne j=1}^{N-1} {\sum_{q=1}^{2L} {\sum_{n=1}^{N} {\mathop{\mathbb{E}}\limits_{\boldsymbol{\theta}}\left(\frac{\partial C_{i}}{\partial{\theta_{q, n}}}\frac{\partial C_{j}}{\partial{\theta_{q, n}}}\right)}}},
		\end{align}
where $C_{i}= \mathrm{Tr}\left[ Z_{i}Z_{i+1}U( \boldsymbol{\theta} )\rho U^{\dagger}( \boldsymbol{\theta} )\right]$.
Thus, we turn to analyze the lower bound of  Eq.~\eqref{totalsum}.
		
For clarity, we first introduce notation. As illustrated in Fig.~\ref{fig:setup}(a), the $L$-layer parameterized unitary $U( \boldsymbol{\theta} )$ can be expressed as the product of $L$ sequential unitaries
	\begin{equation}\nonumber
		U( \boldsymbol{\theta} )=U_{L}( \boldsymbol{\theta}_{2L},\boldsymbol{\theta}_{2L-1}) \cdots U_{1}( \boldsymbol{\theta}_{2},\boldsymbol{\theta}_{1} ),
		\end{equation}
with
		\begin{align*}
		U_{l}( \boldsymbol{\theta}_{2l},\boldsymbol{\theta}_{2l-1})&=R_{2l}(\boldsymbol{\theta}_{2l})R_{2l-1}(\boldsymbol{\theta}_{2l-1})\mathrm{CZ}_{l},
		\\
		R_{2l}(\boldsymbol{\theta}_{2l})&=e^{- \frac{i}{2}\theta_{2l,1}Y}\otimes\cdots\otimes e^{- \frac{i}{2}\theta_{2l,N}Y},
		\\
		R_{2l-1}(\boldsymbol{\theta}_{2l-1})&=e^{- \frac{i}{2}\theta_{2l-1,1}X}\otimes\cdots\otimes e^{- \frac{i}{2}\theta_{2l-1,N}X}.
		\end{align*}
Here, for  $l \in \left[L\right]$, $\mathrm{CZ}_{l}$ denotes the $l$th entanglement layer composed of arbitrarily many $\mathrm{CZ}$ gates, and $R_{2l}(\boldsymbol{\theta}_{2l})$ and  $R_{2l-1}(\boldsymbol{\theta}_{2l-1})$ denote the $l$th  $R_{Y}$ and $R_{X}$ rotation layer, respectively.  Denote $\rho_{0} = \left| 0\right\rangle\left\langle 0\right|$, and for $l \in \left[L\right]$, let

\begin{equation*}
    \left\{
\begin{aligned}
	&\rho_{2l} = R_{2l}(\boldsymbol{\theta}_{2l})\rho_{2l-1} R_{2l}^{\dagger}(\boldsymbol{\theta}_{2l}), \\
		&\rho_{2l-1} = R_{2l-1}(\boldsymbol{\theta}_{2l-1})\mathrm{CZ}_{l}\rho_{2l-2}\mathrm{CZ}^{\dagger}_{l} R_{2l-1}^{\dagger}(\boldsymbol{\theta}_{2l-1}).
	\end{aligned}
		\right.
\end{equation*}	
		
Now we evaluate the lower bound of the first summation in Eq.~\eqref{totalsum}.  Note that
		\begin{align}
		\mathop{\mathbb{E}}\limits_{\boldsymbol{\theta}}\left(\frac{\partial C_{i}}{\partial{\theta_{q, n}}}\right)^{2}
		&=\mathop{\mathbb{E}}\limits_{\boldsymbol{\theta}}\left(\frac{\partial }{\partial{\theta_{q, n}}}\mathrm{Tr}\left[ Z_{i}Z_{i+1}U( \boldsymbol{\theta} )\rho_{0} U^{\dagger}( \boldsymbol{\theta} )\right]\right)^{2}\nonumber
		\\
		&=\mathop{\mathbb{E}}\limits_{\boldsymbol{\theta}}\left(\frac{\partial }{\partial{\theta_{q, n}}}\mathrm{Tr}\left[ Z_{i}Z_{i+1}\rho_{2L}\right]\right)^{2}\nonumber
		\\
		&=\mathop{\mathbb{E}}\limits_{\boldsymbol{\theta}_{[2L-1]}}\mathop{\mathbb{E}}\limits_{\boldsymbol{\theta}_{2L}}\left(\frac{\partial }{\partial{\theta_{q, n}}}\mathrm{Tr}\left[ Z_{i}Z_{i+1}R_{2L}(\boldsymbol{\theta}_{2L})\rho_{2L-1} R_{2L}^{\dagger}(\boldsymbol{\theta}_{2L})\right]\right)^{2}\label{partial&Tr}
		\\
		&=\mathop{\mathbb{E}}\limits_{\boldsymbol{\theta}_{[2L-1]}}\mathop{\mathbb{E}}\limits_{\boldsymbol{\theta}_{2L}}{\mathrm{Tr}^{2}\left[ Z_{i}Z_{i+1}R_{2L}(\boldsymbol{\theta}_{2L})\frac{\partial {\rho_{2L-1}} }{\partial{\theta_{q, n}}} R_{2L}^{\dagger}(\boldsymbol{\theta}_{2L})\right]}\label{Tr&partial}
		\\
		&\geq\alpha^{2}\mathop{\mathbb{E}}\limits_{\boldsymbol{\theta}_{[2L-1]}}{\mathrm{Tr}^{2}\left[ Z_{i}Z_{i+1}\frac{\partial{\rho_{2L-1}} }{\partial{\theta_{q, n}}} \right]}\label{R2L}
		\\
		&=\alpha^{2}\mathop{\mathbb{E}}\limits_{\boldsymbol{\theta}_{[2L-2]}}\mathop{\mathbb{E}}\limits_{\boldsymbol{\theta}_{2L-1}}{\mathrm{Tr}^{2}\left[ Z_{i}Z_{i+1}R_{2L-1}(\boldsymbol{\theta}_{2L-1})\mathrm{CZ}_{L}\frac{\partial{\rho_{2L-2}}}{\partial{\theta_{q, n}}}\mathrm{CZ}^{\dagger}_{L} R_{2L-1}^{\dagger}(\boldsymbol{\theta}_{2L-1}) \right]}\nonumber
		\\
		&\geq\alpha^{4}\mathop{\mathbb{E}}\limits_{\boldsymbol{\theta}_{[2L-2]}}\left(\frac{\partial}{\partial{\theta_{q, n}}}\mathrm{Tr}\left[ \mathrm{CZ}^{\dagger}_{L} Z_{i}Z_{i+1}\mathrm{CZ}_{L}{\rho_{2L-2}} \right]\right)^{2}\label{R2L-1}
		\\
		&=\alpha^{4}\mathop{\mathbb{E}}\limits_{\boldsymbol{\theta}_{[2L-2]}}\left(\frac{\partial}{\partial{\theta_{q, n}}}\mathrm{Tr}\left[ Z_{i}Z_{i+1}{\rho_{2L-2}} \right]\right)^{2}\label{CZL}.
		\end{align}
Here,  Eq.~\eqref{Tr&partial} is due to the linearity of the trace operation, which is frequently utilized in the proof. Eq.~\eqref{R2L} is obtained by first employing Eq.~\eqref{secondpowercorollary1} to take the expectations over $\theta_{2L,i}$ and $ \theta_{2L,i+1}$ producing the coefficient $\alpha^2$, and then leveraging  Eq.~\eqref{firstpowercorollary3} to take the expectations over the  other  $N-2$ components in $\boldsymbol{\theta}_{2L}$.  Following a similar analysis,  Eq.~\eqref{R2L-1}  can be derived. Eq.~\eqref{CZL} is due to Lemma~\ref{CZlemma}, that is, the tensor product $Z_{i}Z_{i+1}$ remains unchanged after applying arbitrarily many $\mathrm{CZ}$ gates.
				
We repeat the techniques for deriving Eqs.~\eqref{partial&Tr}-\eqref{CZL} to calculate the expectations over $\boldsymbol{\theta}_{2L-1}, \dots, \boldsymbol{\theta}_{q+1}$ and obtain
		\begin{equation}\label{R2L-1toq}
		\mathop{\mathbb{E}}\limits_{\boldsymbol{\theta}}\left(\frac{\partial C_{i}}{\partial{\theta_{q, n}}}\right)^{2}\geq\alpha^{2(2L-q)}\mathop{\mathbb{E}}\limits_{\boldsymbol{\theta}_{[q]}}\left(\frac{\partial}{\partial{\theta_{q, n}}}\mathrm{Tr}\left[ Z_{i}Z_{i+1}{\rho_{q}} \right]\right)^{2}.
		\end{equation}
		
Next, we proceed upon the index $n$, the ordinal number of the qubit where the trainable parameter $\theta_{q, n}$ lies. In the case where $n\in \left\lbrace i, i+1\right\rbrace $, we have
		\begin{align}
		\mathop{\mathbb{E}}\limits_{\boldsymbol{\theta}}\left(\frac{\partial C_{i}}{\partial{\theta_{q, n}}}\right)^{2}
		&\geq\alpha^{2(2L-q)}\cdot\alpha\beta\mathop{\mathbb{E}}\limits_{\boldsymbol{\theta}_{[q-1]}}{\mathrm{Tr}^{2}\left[ Z_{i}Z_{i+1}{\rho_{q-1}} \right]}\label{Rq}
		\\
		&\geq\alpha^{2(2L-q)+1}\beta\cdot\alpha^{2(q-1)}\,{\mathrm{Tr}^{2}\left[ Z_{i}Z_{i+1}{\rho_{0}} \right]}\label{R1toq-1},
		\end{align}
where $\beta$ in Eq.~\eqref{Rq} is obtained by employing  Eq.~\eqref{secondpowercorollary2} to take the expectation over $\theta_{q,n}$. It is worth pointing out that to persist the observable $Z_{i}Z_{i+1}$ in the following analysis,  we keep the item with coefficient $\beta$ rather than $\alpha$ in the inequality.  As for the other $N-1$ components in $\boldsymbol{\theta}_{q}$, the analysis is the same as that in Eq.~\eqref{R2L} or \eqref{R2L-1}. Following a similar derivation to Eq.~\eqref{R2L-1toq}, Eq.~\eqref{R1toq-1} can be obtained.
		
In the case where $n\notin \left\lbrace i, i+1\right\rbrace $, when taking the expectation over $ \theta_{q, n}$, by Eq.~\eqref{firstpowercorollary4} we have
		\begin{equation}\label{key18}
		\mathop{\mathbb{E}}\limits_{\boldsymbol{\theta}_{[q]}}\left(\frac{\partial}{\partial{\theta_{q, n}}}\mathrm{Tr}\left[ Z_{i}Z_{i+1}{\rho_{q}} \right]\right)^{2} = 0.
		\end{equation}
		
Therefore, combining Eq.~\eqref{R1toq-1}  and Eq.~\eqref{key18}, it can be seen that
		\begin{equation}\label{71}
		\mathop{\mathbb{E}}\limits_{\boldsymbol{\theta}}\left(\frac{\partial C_{i}}{\partial{\theta_{q, n}}}\right)^{2} \geq \left\{
		\begin{aligned}
		& \alpha^{4L-1}\beta, \  n\in \left\lbrace i, i+1\right\rbrace, \\
		& 0, \quad\qquad\,  \mathrm{otherwise}.
		\end{aligned}
		\right.
		\end{equation}
and the first summation in  Eq.~\eqref{totalsum} is lower bounded by
		\begin{equation}\label{quadsum}
		\sum_{i=1}^{N-1} {\sum_{q=1}^{2L} {\sum_{n=1}^{N} {\mathop{\mathbb{E}}\limits_{\boldsymbol{\theta}}\left(\frac{\partial C_{i}}{\partial{\theta_{q, n}}}\right)^{2}}}}\geq (N-1)\times 2L \times 2\alpha^{4L-1}\beta.
		\end{equation}
		
		Then, we analyze the second summation in Eq.~\eqref{totalsum}. Note that
		\begin{equation}\nonumber
		\mathop{\mathbb{E}}\limits_{\boldsymbol{\theta}}\left(\frac{\partial C_{i}}{\partial{\theta_{q, n}}}\frac{\partial C_{j}}{\partial{\theta_{q, n}}}\right)=\mathop{\mathbb{E}}\limits_{\boldsymbol{\theta}}\left(\frac{\partial }{\partial{\theta_{q, n}}}\mathrm{Tr}\left[ Z_{i}Z_{i+1}\rho_{2L}\right]\frac{\partial }{\partial{\theta_{q, n}}}\mathrm{Tr}\left[ Z_{j}Z_{j+1}\rho_{2L}\right]\right).
		\end{equation}
		
We first consider the case where $\left\lbrace i, i+1\right\rbrace \cap \left\lbrace j, j+1\right\rbrace = \emptyset $.
		\begin{align}
		\label{0cross2L}
		\mathop{\mathbb{E}}\limits_{\boldsymbol{\theta}}\left(\frac{\partial C_{i}}{\partial{\theta_{q, n}}}\frac{\partial C_{j}}{\partial{\theta_{q, n}}}\right)&=\gamma^{4(2L-q)}
		\mathop{\mathbb{E}}\limits_{\boldsymbol{\theta}_{[q]}}\left(\frac{\partial }{\partial{\theta_{q, n}}}\mathrm{Tr}\left[ Z_{i}Z_{i+1}\rho_{q}\right]\frac{\partial }{\partial{\theta_{q, n}}}\mathrm{Tr}\left[ Z_{j}Z_{j+1}\rho_{q}\right]\right)
		\\&=0\label{0cross}.	
		\end{align}
In regard to the expectation over $\boldsymbol{\theta}_{2L}$ in Eq.~\eqref{0cross2L}, the coefficient $\gamma$ appears each time when we take the expectation over $\theta_{2L,i}$, $ \theta_{2L,i+1} $, $\theta_{2L,j}$ or $ \theta_{2L,j+1}$ by employing Eq.~\eqref{firstpowercorollary1}. In addition, Eq.~\eqref{firstpowercorollary3} is utilized to calculate the expectations over the other components in $\boldsymbol{\theta}_{2L}$. Similarly, we can take the expectations over $\boldsymbol{\theta}_{2L-1},\dots,\boldsymbol{\theta}_{q+1}$, during which the calculation of each $\mathrm{CZ}_l$ layer is the same as that in Eq.~\eqref{CZL}. When taking the expectation over $ \theta_{q, n} $, Eq.~\eqref{0cross} can be obtained by employing either Eq.~\eqref{firstpowercorollary2} for $n\in \left\lbrace i, i+1, j, j+1\right\rbrace$ or Eq.~\eqref{firstpowercorollary4}, otherwise.
		
Next, we consider the case where  $i=j+1$ or $i+1=j$. Let us consider the case $i+1=j$ in detail, and the other case can be followed in a similar way. Note that
		\begin{align}
		\mathop{\mathbb{E}}\limits_{\boldsymbol{\theta}}\left(\frac{\partial C_{i}}{\partial{\theta_{q, n}}}\frac{\partial C_{i+1}}{\partial{\theta_{q, n}}}\right)=&\mathop{\mathbb{E}}\limits_{\boldsymbol{\theta}_{[2L]}}\left(\frac{\partial }{\partial{\theta_{q, n}}}\mathrm{Tr}\left[ Z_{i}Z_{i+1}\rho_{2L}\right]\frac{\partial }{\partial{\theta_{q, n}}}\mathrm{Tr}\left[ Z_{i+1}Z_{i+2}\rho_{2L}\right]\right)\label{1crossoriginal}
		\\
		\begin{split}\label{1cross2L}
		=&\gamma^{2}\alpha\mathop{\mathbb{E}}\limits_{\boldsymbol{\theta}_{[2L-1]}}\left(\frac{\partial }{\partial{\theta_{q, n}}}\mathrm{Tr}\left[ Z_{i}Z_{i+1}\rho_{2L-1}\right]\frac{\partial }{\partial{\theta_{q, n}}}\mathrm{Tr}\left[ Z_{i+1}Z_{i+2}\rho_{2L-1}\right]\right)
		\\
		&+\gamma^{2}\beta\mathop{\mathbb{E}}\limits_{\boldsymbol{\theta}_{[2L-1]}}\left(\frac{\partial }{\partial{\theta_{q, n}}}\mathrm{Tr}\left[ Z_{i}X_{i+1}\rho_{2L-1}\right]\frac{\partial }{\partial{\theta_{q, n}}}\mathrm{Tr}\left[ X_{i+1}Z_{i+2}\rho_{2L-1}\right]\right)
		\end{split}
		\\
		\begin{split}\label{1cross2L-1}
		=&\gamma^{4}\alpha^{2}\mathop{\mathbb{E}}\limits_{\boldsymbol{\theta}_{[2L-2]}}
		\left(
		\frac{\partial }{\partial{\theta_{q, n}}}\mathrm{Tr}\left[ \mathrm{CZ}^{\dagger}_{L}Z_{i}Z_{i+1}\mathrm{CZ}_{L}\rho_{2L-2}\right]\frac{\partial }{\partial{\theta_{q, n}}}\mathrm{Tr}\left[ \mathrm{CZ}^{\dagger}_{L}Z_{i+1}Z_{i+2}\mathrm{CZ}_{L}\rho_{2L-2}\right]
		\right)
		\\
		&+\gamma^{4}\alpha\beta\mathop{\mathbb{E}}\limits_{\boldsymbol{\theta}_{[2L-2]}}
		\left(
		\frac{\partial }{\partial{\theta_{q, n}}}\mathrm{Tr}\left[ \mathrm{CZ}^{\dagger}_{L}Z_{i}Y_{i+1}\mathrm{CZ}_{L}\rho_{2L-2}\right]\frac{\partial }{\partial{\theta_{q, n}}}\mathrm{Tr}\left[ \mathrm{CZ}^{\dagger}_{L}Y_{i+1}Z_{i+2}\mathrm{CZ}_{L}\rho_{2L-2}\right]
		\right)
		\\
		&+\gamma^{4}\beta\mathop{\mathbb{E}}\limits_{\boldsymbol{\theta}_{[2L-2]}}
		\left(
		\frac{\partial }{\partial{\theta_{q, n}}}\mathrm{Tr}\left[\mathrm{CZ}^{\dagger}_{L} Z_{i}X_{i+1}\mathrm{CZ}_{L}\rho_{2L-2}\right]\frac{\partial }{\partial{\theta_{q, n}}}\mathrm{Tr}\left[\mathrm{CZ}^{\dagger}_{L} X_{i+1}Z_{i+2}\mathrm{CZ}_{L}\rho_{2L-2}\right]
		\right).
		\end{split}
		\end{align}
Here, Eq.~\eqref{1cross2L} is obtained by taking the expectations over $\theta_{2L,i+1}$ by Eq.~\eqref{secondpowerlemma1}, $\theta_{2L,i}$ and $\theta_{2L,i+2}$ by Eq.~\eqref{firstpowercorollary1}, and the other $N-3$ components by Eq.~\eqref{firstpowercorollary3}. Similarly, Eq.~\eqref{1cross2L-1} can be derived.
		
It is worth pointing out that after employing Lemmas~\ref{firstpowerlemma} and \ref{secondpowerlemma} to calculate the expectations over the trainable parameters in rotation layers, except for  the terms related to the original observables $Z_iZ_{i+1}$ and $Z_{i+1}Z_{i+2}$, other terms related to new observables may be generated.  For example, after taking the expectation over $\boldsymbol{\theta}_{2L}$, the original  observable $Z_{i}Z_{i+1}$ in Eq.~\eqref{1crossoriginal}  generates two terms related to the observable $Z_{i}Z_{i+1}$ and $Z_{i}X_{i+1}$, respectively, in Eq.~\eqref{1cross2L}. It is clear that the generated new observables are always $N$-qubit Pauli strings as defined in Eq.~\eqref{paulitensor}, and the corresponding  coefficients are monomials of $\gamma$, $\alpha$ and $\beta$, which are always positive.  Moreover, according to Lemma~\ref{CZlemma}, after applying an entanglement layer composed of $\mathrm{CZ}$ gates, the generated observables are still $N$-qubit Pauli strings, where, for example, the observable $Z_iZ_{i+1}$ remains unchanged. Thus, whenever we finish the calculations with respect to a block, there are always  terms related to the original observables $Z_{i}Z_{i+1}$ and $Z_{i+1}Z_{i+2}$, and some other generated terms related to observables in the form of  $N$-qubit Pauli strings. We point out that all the generated terms are  non-negative, and we come back to its proof later. In the following, the terms with the original observables $Z_{i}Z_{i+1}$ and $Z_{i+1}Z_{i+2}$ are kept only for further analysis. Thus, we have
		\begin{align}
		\label{1cross2L-2}
		\mathop{\mathbb{E}}\limits_{\boldsymbol{\theta}}\left(\frac{\partial C_{i}}{\partial{\theta_{q, n}}}\frac{\partial C_{i+1}}{\partial{\theta_{q, n}}}\right)&\geq\gamma^{4}\alpha^{2}\mathop{\mathbb{E}}\limits_{\boldsymbol{\theta}_{[2L-2]}}\left(
		\frac{\partial }{\partial{\theta_{q, n}}}\mathrm{Tr}\left[ Z_{i}Z_{i+1}\rho_{2L-2}\right]\frac{\partial }{\partial{\theta_{q, n}}}\mathrm{Tr}\left[ Z_{i+1}Z_{i+2}\rho_{2L-2}\right]\right)
		\\
		\label{1crosstoq+1}
		&\geq\gamma^{2(2L-q)}\alpha^{2L-q}\mathop{\mathbb{E}}\limits_{\boldsymbol{\theta}_{[q]}}\left(\frac{\partial }{\partial{\theta_{q, n}}}\mathrm{Tr}\left[ Z_{i}Z_{i+1}\rho_{q}\right]\frac{\partial }{\partial{\theta_{q, n}}}\mathrm{Tr}\left[ Z_{i+1}Z_{i+2}\rho_{q}\right]\right),
		\end{align}
where Eq.~\eqref{1crosstoq+1} is obtained by repeating the same analysis as that in Eq.~\eqref{1cross2L-2}  for  the expectations over $\boldsymbol{\theta}_{2L-2}, \dots, \boldsymbol{\theta}_{q+1}$.
		
Then, we proceed upon the index $n$. When $n=i+1$, we have
		\begin{align}
		&\mathop{\mathbb{E}}\limits_{\boldsymbol{\theta}_{[q]}}\left(\frac{\partial }{\partial{\theta_{q, n}}}\mathrm{Tr}\left[ Z_{i}Z_{i+1}\rho_{q}\right]\frac{\partial }{\partial{\theta_{q, n}}}\mathrm{Tr}\left[ Z_{i+1}Z_{i+2}\rho_{q}\right]\right)\nonumber
		\\
		\geq&\gamma^{2}\beta\mathop{\mathbb{E}}\limits_{\boldsymbol{\theta}_{[q-1]}}\left(\mathrm{Tr}\left[ Z_{i}Z_{i+1}\rho_{q-1}\right]\mathrm{Tr}\left[ Z_{i+1}Z_{i+2}\rho_{q-1}\right]\right)\label{1crossq}
		\\
		\geq&\gamma^{2}\beta\cdot\gamma^{2(q-1)}\alpha^{q-1}\mathrm{Tr}\left[ Z_{i}Z_{i+1}\rho_{0}\right]\mathrm{Tr}\left[ Z_{i+1}Z_{i+2}\rho_{0}\right],\label{1cross0}
		\end{align}
where in Eq.~\eqref{1crossq}, the coefficient $\beta$ comes from Eq.~\eqref{secondpowerlemma2} for the expectation over  $\theta_{q,n}$, and the coefficient $\gamma^{2}$ comes from Eq.~\eqref{firstpowercorollary1} for $\theta_{q,n-1}$ and $\theta_{q,n+1}$. Eq.~\eqref{1cross0}  is obtained in the same way as Eq.~\eqref{1crosstoq+1}.
		
When $n\ne i+1$, by employing either Eq.~\eqref{firstpowercorollary2} for $n\in \left\lbrace i, i+2\right\rbrace$ or Eq.~\eqref{firstpowercorollary4}, otherwise, to take the expectation over $ \theta_{q, n} $, we have
		\begin{equation}\label{1cross00}
		\mathop{\mathbb{E}}\limits_{\boldsymbol{\theta}_{[q]}}\left(\frac{\partial }{\partial{\theta_{q, n}}}\mathrm{Tr}\left[ Z_{i}Z_{i+1}\rho_{q}\right]\frac{\partial }{\partial{\theta_{q, n}}}\mathrm{Tr}\left[ Z_{i+1}Z_{i+2}\rho_{q}\right]\right)=0.
		\end{equation}
		
Now we prove the non-negativity of the generated terms. Without loss of generality, we take the inequality in Eq.~\eqref{1cross0} as an example. The gap between the interested term \eqref{1crossoriginal} and the final term \eqref{1cross0} can be  expressed exactly as a linear combination of all generated terms
		\begin{equation}\label{ZZnon-neg}
		\sum_{k}c_{k}\mathrm{Tr}\left[ \boldsymbol{\sigma}_{\boldsymbol{g}_{i}^{k}}\rho_{0}\right] \mathrm{Tr}\left[ \boldsymbol{\sigma}_{\boldsymbol{g}_{i+1}^{k}}\rho_{0}\right],
		\end{equation}
  where $c_{k}$ is the positive coefficient, $\boldsymbol{\sigma}_{\boldsymbol{g}_{i}^{k}}$ and $\boldsymbol{\sigma}_{\boldsymbol{g}_{i+1}^{k}}$ denote the $N$-qubit Pauli strings generated from $Z_{i}Z_{i+1}$ and $Z_{i+1}Z_{i+2}$, but not themselves, respectively. With the choice of the initial state $\rho_{0}=\left| 0\right\rangle\left\langle 0\right|$, other terms like Eq.~\eqref{ZZnon-neg}  are always non-negative.
		
Therefore, from Eqs.~\eqref{1crosstoq+1}, \eqref{1cross0}, and \eqref{1cross00}, we have
\begin{equation*}
    \mathop{\mathbb{E}}\limits_{\boldsymbol{\theta}}\left(\frac{\partial C_{i}}{\partial{\theta_{q, n}}}\frac{\partial C_{j}}{\partial{\theta_{q, n}}}\right) \geq \left\{
		\begin{aligned}
		& \gamma^{4L}\alpha^{2L-1}\beta, \  n=i=j+1,
		\\
		& \gamma^{4L}\alpha^{2L-1}\beta, \  n=i+1=j,
		\\
		& 0, \qquad\qquad\ \ \mathrm{otherwise},
		\end{aligned}
		\right.
\end{equation*}
and the second summation in Eq.~\eqref{totalsum} is lower bounded by
		\begin{equation}\label{crosssum}
		\sum_{i\ne j=1}^{N-1} {\sum_{q=1}^{2L} {\sum_{n=1}^{N} {\mathop{\mathbb{E}}\limits_{\boldsymbol{\theta}}\left(\frac{\partial C_{i}}{\partial{\theta_{q, n}}}\frac{\partial C_{j}}{\partial{\theta_{q, n}}}\right)}}}\geq 2(N-2)\times 2L\gamma^{4L}\alpha^{2L-1}\beta.
		\end{equation}
		
Combining Eq.~\eqref{quadsum} and Eq.~\eqref{crosssum}, we have
		\begin{equation}
		\mathop{\mathbb{E}}\limits_{\boldsymbol{\theta}}\Vert\nabla_{\boldsymbol{\theta}}{C}\Vert_{2}^{2}\geq 4(N-1)L\alpha^{4L-1}\beta+4(N-2)L\gamma^{4L}\alpha^{2L-1}\beta
		\geq 4(2N-3)L\gamma^{8L-2}\beta\label{TH11proof},
		\end{equation}
where the second inequality is derived by using the inequality
		\begin{equation}\nonumber
		\alpha=\frac{1}{4}\left( 2+\frac{\sin {2a\pi}}{a\pi}\right)\geq \gamma^{2}=\left( \frac{\sin {a\pi}}{a\pi}\right) ^{2},\ a\in\left( 0, 1\right).		
		\end{equation}
		
Finally, we  determine  the hyperparameter $a$. Let $x$ denote $(a\pi)^2$. It can be verified  that  for $a\in \left(0, 1 \right)$,
		\begin{align*}
			&\gamma\geq 1-\frac{1}{6}a^{2}\pi^{2}=1-\frac{1}{6}x,
			\\
			&\beta\geq \frac{1}{3}a^{2}\pi^{2}-\frac{1}{15}a^{4}\pi^{4}=\frac{1}{3}x-\frac{1}{15}x^{2}.
		\end{align*}
In order to ensure the positivity of the right sides of the two inequalities above, we further set $a\in \left(0, \frac{\sqrt{5}}{\pi} \right)$, i.e., $x\in \left(0, 5 \right)$. Then, the lower bound in Eq.~\eqref{TH11}, that is, Eq.~\eqref{TH11proof} can be further bounded by
		\begin{equation}\nonumber
		4\left(2N-3\right)L\left(\frac{1}{3}x-\frac{1}{15}x^{2}\right) {\left(1-\frac{1}{6}x\right)}^{8L-2}
		\triangleq{g}\bm{(}N,L,x\bm{)}.
		\end{equation}
		
It can be verified that $g\bm{(}N,L,x\bm{)}$ attains its maximum when
		\begin{equation}\nonumber
		x=\frac{40L+7-\sqrt{1600L^{2}-400L+49}}{16L}.
		\end{equation}
		Thus, we choose the hyperparameter as
		\begin{equation}\label{keya1}
		a=\frac{1}{4\pi}\sqrt{\frac{40L+7-\sqrt{1600L^{2}-400L+49}}{L}},
		\end{equation}
which has the same scaling as $\frac{1}{\sqrt{L}}$, that is, $a\in \mathrm{\Theta}\left(\frac{1}{\sqrt{L}}\right)$. Moreover, with the selected $a$ in Eq.~\eqref{keya1},
		$g\bm{(}N,L,a\bm{)}$ can be lower bounded by
		\begin{equation}\nonumber
		g\bm{(}N,L,a\bm{)}\geq e^{-1}(2N-3),
		\end{equation}
which is obtained by the fact that the function $\frac{1}{2N-3}g\bm{(}N,L,a\bm{)}$ decreases monotonically with $L$ and converges to $e^{-1}$.
	\end{proof}
	
\section{\label{proofTH2}Proof of Theorem~\ref{TH2}}
We continue to utilize the notation in Appendix~\ref{proofTH1}. In addition, for each term $\boldsymbol{\sigma}_{\boldsymbol{i}}$ in Eq.~\eqref{generalH}, denote $S^{\boldsymbol{i}}$ as the number of nonzero components in the vector $\boldsymbol{i}=\left( i_{1}, i_{2}, \cdots, i_{N}\right) $. According to the definition of $S$-local cost function in the form of Eq.~\eqref{generalH}, we have $S^{\boldsymbol{i}}\le S$. The number of elements being $j$ in vector $\boldsymbol{i} $ is denoted by $S^{\boldsymbol{i}}_{j}$, for $j\in\left\lbrace 1, 2, 3\right\rbrace $. Thus, we have
\begin{equation}\label{Si&S}
	S^{\boldsymbol{i}}=S^{\boldsymbol{i}}_{1}+S^{\boldsymbol{i}}_{2}+S^{\boldsymbol{i}}_{3}\le S.
\end{equation}
	
We introduce $I_{S^{\boldsymbol{i}}}\triangleq\left\lbrace n\,|\,i_{n}\ne 0, n\in\left[N \right]  \right\rbrace$ to denote the set of qubits upon which the term $\boldsymbol{\sigma}_{\boldsymbol{i}}$ acts non-trivially. Similarly, we define $I_{S^{\boldsymbol{i}}_{j}}\triangleq \{ n|i_{n}=j, n\in\left[N \right]\}$ for $j\in \left\lbrace 1, 2, 3\right\rbrace$.
	
Moreover, for $\boldsymbol{i}=( i_{1}, i_{2}, \cdots, i_{N})$, $ i_{j}\in \left\lbrace   0,1,2,3\right\rbrace $, $ j \in [N]$, let  $\boldsymbol{3:i}$ denote the $N$-dimensional vector after replacing all nonzero components in $\boldsymbol{i}$ with $3$, and $\boldsymbol{3:i;1}$ denote the $N$-dimensional vector after replacing all $1s$ in $\boldsymbol{i}$ with $3$.
	
\begin{proof}
For the general Hamiltonian of Eq.~\eqref{generalH}, the expectation of the squared norm of the gradient of the cost function can be described as
		\begin{align}\nonumber
		\mathop{\mathbb{E}}\limits_{\boldsymbol{\theta}}\Vert\nabla_{\boldsymbol{\theta}}{C}\Vert_{2}^{2}
		&=\sum_{q=1}^{2L}{\sum_{n=1}^{N}{ \mathop{\mathbb{E}}\limits_{\boldsymbol{\theta}}\left(\frac{\partial C}{\partial{\theta_{q, n}}}\right)^{2}}}
		\\
		&=\sum_{q=1}^{2L} {\sum_{n=1}^{N} {\mathop{\mathbb{E}}\limits_{\boldsymbol{\theta}}\left(\sum_{\boldsymbol{i}\in \mathcal{N}}{\frac{\partial C_{\boldsymbol{i}}}{\partial{\theta_{q, n}}}}\right)^{2}}}\nonumber
		\\
		\label{totalsumH}
		&=\sum_{\boldsymbol{i}\in \mathcal{N}} {\sum_{q=1}^{2L} {\sum_{n=1}^{N} {\mathop{\mathbb{E}}\limits_{\boldsymbol{\theta}}\left(\frac{\partial C_{\boldsymbol{i}}}{\partial{\theta_{q, n}}}\right)^{2}}}}
		+\sum_{\boldsymbol{i}\ne \boldsymbol{j}\in \mathcal{N}} {\sum_{q=1}^{2L} {\sum_{n=1}^{N} {\mathop{\mathbb{E}}\limits_{\boldsymbol{\theta}}\left(\frac{\partial C_{\boldsymbol{i}}}{\partial{\theta_{q, n}}}\frac{\partial C_{\boldsymbol{j}}}{\partial{\theta_{q, n}}}\right)}}},
		\end{align}
where $C_{\boldsymbol{i}}= \mathrm{Tr}\left[ \boldsymbol{\sigma}_{\boldsymbol{i}}U( \boldsymbol{\theta} )\rho U^{\dagger}( \boldsymbol{\theta} )\right]$. Thus, we turn to analyze the lower bound of Eq.~\eqref{totalsumH}.
		
For the first summation in Eq.~\eqref{totalsumH}, we have
		\begin{equation}\nonumber
		\sum_{\boldsymbol{i}\in \mathcal{N}} {\sum_{q=1}^{2L} {\sum_{n=1}^{N} {\mathop{\mathbb{E}}\limits_{\boldsymbol{\theta}}\left(\frac{\partial C_{\boldsymbol{i}}}{\partial{\theta_{q, n}}}\right)^{2}}}}\geq\sum_{\boldsymbol{i}\in \mathcal{N}} {\sum_{q=1}^{2L-2} {\sum_{n=1}^{N} {\mathop{\mathbb{E}}\limits_{\boldsymbol{\theta}}\left(\frac{\partial C_{\boldsymbol{i}}}{\partial{\theta_{q, n}}}\right)^{2}}}}.
		\end{equation}
Note that, for $q\in\left[2L-2\right]$,
		\begin{align}\nonumber
		\mathop{\mathbb{E}}\limits_{\boldsymbol{\theta}}\left(\frac{\partial C_{\boldsymbol{i}}}{\partial{\theta_{q, n}}}\right)^{2}
		&=\mathop{\mathbb{E}}\limits_{\boldsymbol{\theta}}\left(\frac{\partial }{\partial{\theta_{q, n}}}\mathrm{Tr}\left[ \boldsymbol{\sigma}_{\boldsymbol{i}}\rho_{2L} \right]\right)^{2}
		\\\nonumber
		&=\mathop{\mathbb{E}}\limits_{\boldsymbol{\theta}_{[2L-1]}}\mathop{\mathbb{E}}\limits_{\boldsymbol{\theta}_{2L}}\left(\frac{\partial }{\partial{\theta_{q, n}}}\mathrm{Tr}\left[\boldsymbol{\sigma}_{\boldsymbol{i}}R_{2L}(\boldsymbol{\theta}_{2L})\rho_{2L-1} R_{2L}^{\dagger}(\boldsymbol{\theta}_{2L})\right]\right)^{2}
		\\
		&\geq\alpha^{S^{\boldsymbol{i}}_{3}}\beta^{S^{\boldsymbol{i}}_{1}}\mathop{\mathbb{E}}\limits_{\boldsymbol{\theta}_{[2L-1]}}\left(\frac{\partial }{\partial{\theta_{q, n}}}\mathrm{Tr}\left[\boldsymbol{\sigma}_{\boldsymbol{3:i;1}}\rho_{2L-1}\right]\right)^{2}\label{HR2L}
		\\
		&\geq\alpha^{S^{\boldsymbol{i}}_{1}+2S^{\boldsymbol{i}}_{3}}\beta^{S^{\boldsymbol{i}}_{1}+S^{\boldsymbol{i}}_{2}}\mathop{\mathbb{E}}\limits_{\boldsymbol{\theta}_{[2L-2]}}\left(\frac{\partial}{\partial{\theta_{q, n}}}\mathrm{Tr}\left[ \mathrm{CZ}^{\dagger}_{L}\boldsymbol{\sigma}_{\boldsymbol{3:i}}\mathrm{CZ}_{L}{\rho_{2L-2}}  \right]\right)^{2}\label{HR2L-1}
		\\
		&=\alpha^{S^{\boldsymbol{i}}_{1}+2S^{\boldsymbol{i}}_{3}}\beta^{S^{\boldsymbol{i}}_{1}+S^{\boldsymbol{i}}_{2}}\mathop{\mathbb{E}}\limits_{\boldsymbol{\theta}_{[2L-2]}}\left(\frac{\partial}{\partial{\theta_{q, n}}}\mathrm{Tr}\left[ \boldsymbol{\sigma}_{\boldsymbol{3:i}}{\rho_{2L-2}} \right]\right)^{2}\label{HCZL},
		\end{align}
where  Eq.~\eqref{HR2L} is obtained by using Eq.~\eqref{secondpowercorollary1} up to $S^{\boldsymbol{i}}_{3}+S^{\boldsymbol{i}}_{1}$ times with respect to $\theta_{2L,n}, n\in I_{S^{\boldsymbol{i}}_{3}} \cup I_{S^{\boldsymbol{i}}_{1}}$, and  using Eq.~\eqref{firstpowercorollary3} up to $N-S^{\boldsymbol{i}}_{3}-S^{\boldsymbol{i}}_{1}$ times with respect to the other components in $\boldsymbol{\theta}_{2L}$. The appearance of  the coefficient $\alpha^{S^{\boldsymbol{i}}_{3}}\beta^{S^{\boldsymbol{i}}_{1}}$ is due to the fact that each time when Eq.~\eqref{secondpowercorollary1} is employed we  only keep the item with more Pauli $Z$ tensors in the corresponding $N$-qubit Pauli string observable. Eq.~\eqref{HR2L-1} is derived by using Eq.~\eqref{secondpowercorollary1} up to $S^{\boldsymbol{i}}$ times with respect to $\theta_{2L-1,n}, n\in I_{S^{\boldsymbol{i}}}$ with the appearance of the coefficient $\alpha^{S^{\boldsymbol{i}}_{1}+S^{\boldsymbol{i}}_{3}}\beta^{S^{\boldsymbol{i}}_{2}}$, and  using Eq.~\eqref{firstpowercorollary3} up to $N-S^{\boldsymbol{i}}$ times with respect to the remain components in $\boldsymbol{\theta}_{2L-1}$. Eq.~\eqref{HCZL} follows the same reasoning behind Eq.~\eqref{CZL}.
		
After taking the expectations over $\boldsymbol{\theta}_{2L}$ and $\boldsymbol{\theta}_{2L-1}$,  the observable remained in Eq.~\eqref{HCZL} is  $\boldsymbol{\sigma}_{\boldsymbol{3:i}}$, which is an $N$-qubit tensor product composed of identity and at most $S$ Pauli $Z$ operators. Thus, to bound the expectation term in Eq.~\eqref{HCZL}, we can employ the same technique used for the interested term \eqref{71} related to the original observable $Z_{i}Z_{i+1}$, specifically,
\begin{equation*}
    \mathop{\mathbb{E}}\limits_{\boldsymbol{\theta}_{[2L-2]}}\left(\frac{\partial}{\partial{\theta_{q, n}}}\mathrm{Tr}\left[ \boldsymbol{\sigma}_{\boldsymbol{3:i}}{\rho_{2L-2}} \right]\right)^{2} \geq \left\{
		\begin{aligned}
		& \alpha^{2S^{\boldsymbol{i}}(L-1)-1}\beta,\ n\in I_{S^{\boldsymbol{i}}}, \\
		& 0, \qquad\qquad\quad\ \  \mathrm{otherwise}.
		\end{aligned}
		\right.
\end{equation*}
		
Thus, the first summation in Eq.~\eqref{totalsumH} is lower bounded by
		\begin{align}\nonumber
		\sum_{\boldsymbol{i}\in \mathcal{N}} {\sum_{q=1}^{2L} {\sum_{n=1}^{N} {\mathop{\mathbb{E}}\limits_{\boldsymbol{\theta}}\left(\frac{\partial C_{\boldsymbol{i}}}{\partial{\theta_{q, n}}}\right)^{2}}}}
		&\geq(2L-2)\sum_{\boldsymbol{i}\in \mathcal{N}}{S^{\boldsymbol{i}}\alpha^{2S^{\boldsymbol{i}}(L-1)+S^{\boldsymbol{i}}_{1}+2S^{\boldsymbol{i}}_{3}-1}\beta^{S^{\boldsymbol{i}}_{1}+S^{\boldsymbol{i}}_{2}+1}}
		\\
		&\geq(2L-2)\sum_{\boldsymbol{i}\in \mathcal{N}}{S^{\boldsymbol{i}}\alpha^{2SL-1}\beta^{S+1}}\label{expscale}
		\\
		&\geq|\mathcal{N}|(2L-2)\alpha^{2SL-1}\beta^{S+1}\label{quadsumH}.
		\end{align}
Here, Eq.~\eqref{expscale} is due to Eq.~\eqref{Si&S}, and Eq.~\eqref{quadsumH} is due to  the non-triviality of each term in the general Hamiltonian of Eq.~\eqref{generalH}, that is, $S^{\boldsymbol{i}}\geq 1$, for all $\boldsymbol{i}\in \mathcal{N}$.
		
Next, we calculate the second summation in Eq.~\eqref{totalsumH}. Note that
		\begin{equation}\label{interestedd}
		\mathop{\mathbb{E}}\limits_{\boldsymbol{\theta}}\left(\frac{\partial C_{\boldsymbol{i}}}{\partial{\theta_{q, n}}}\frac{\partial C_{\boldsymbol{j}}}{\partial{\theta_{q, n}}}\right)
		=\mathop{\mathbb{E}}\limits_{\boldsymbol{\theta}}\left(\frac{\partial }{\partial{\theta_{q, n}}}\mathrm{Tr}\left[ \boldsymbol{\sigma}_{\boldsymbol{i}}\rho_{2L}\right]\frac{\partial }{\partial{\theta_{q, n}}}\mathrm{Tr}\left[ \boldsymbol{\sigma}_{\boldsymbol{j}}\rho_{2L}\right]\right).
		\end{equation}
		
We first consider the case where $I_{S^{\boldsymbol{i}}} \cap I_{S^{\boldsymbol{j}}} = \emptyset $. Similar to obtaining Eqs.~\eqref{0cross2L} and \eqref{0cross}, we have
		\begin{equation}
		\mathop{\mathbb{E}}\limits_{\boldsymbol{\theta}}\left(\frac{\partial C_{\boldsymbol{i}}}{\partial{\theta_{q, n}}}\frac{\partial C_{\boldsymbol{j}}}{\partial{\theta_{q, n}}}\right)=0.
		\end{equation}
		
Then we consider the  case where $I_{S^{\boldsymbol{i}}} \cap I_{S^{\boldsymbol{j}}} \ne \emptyset $. Following the technique used for the case $i+1=j$ and the form of Eq.~\eqref{ZZnon-neg} in Appendix~\ref{proofTH1}, Eq.~\eqref{interestedd} can be rewritten as a sum of non-negative terms
		\begin{equation}
		\mathop{\mathbb{E}}\limits_{\boldsymbol{\theta}}\left(\frac{\partial C_{\boldsymbol{i}}}{\partial{\theta_{q, n}}}\frac{\partial C_{\boldsymbol{j}}}{\partial{\theta_{q, n}}}\right)
		=\sum_{k}c_{k}\mathrm{Tr}\left[ \boldsymbol{\sigma}_{\boldsymbol{g}_{\boldsymbol{i}}^{k}}\rho_{0}\right] \mathrm{Tr}\left[ \boldsymbol{\sigma}_{\boldsymbol{g}_{\boldsymbol{j}}^{k}}\rho_{0}\right]
		\geq 0,
		\end{equation}
where $c_{k}$ denotes the positive coefficient, $\boldsymbol{\sigma}_{\boldsymbol{g}_{\boldsymbol{i}}^{k}}$ and $\boldsymbol{\sigma}_{\boldsymbol{g}_{\boldsymbol{j}}^{k}}$ denote the $N$-qubit Pauli strings generated from $\boldsymbol{\sigma}_{\boldsymbol{i}}$ and $\boldsymbol{\sigma}_{\boldsymbol{j}}$, respectively. The non-negativity is owing to the chosen initial state $\rho_{0}=\left| 0\rangle\langle0\right|$.
		
Therefore, the second summation in Eq.~\eqref{totalsumH} is non-negative, i.e.,
		\begin{equation}\label{crosssumH}
		\sum_{\boldsymbol{i}\ne \boldsymbol{j}\in \mathcal{N}} {\sum_{q=1}^{2L} {\sum_{n=1}^{N} {\mathop{\mathbb{E}}\limits_{\boldsymbol{\theta}}\left(\frac{\partial C_{\boldsymbol{i}}}{\partial{\theta_{q, n}}}\frac{\partial C_{\boldsymbol{j}}}{\partial{\theta_{q, n}}}\right)}}}\geq 0.
		\end{equation}
		
Combining Eq.~\eqref{quadsumH} and Eq.~\eqref{crosssumH}, we have  Eq.~\eqref{TH21}.
		
Finally, we determine the hyperparameter $a$ by maximizing the lower bound in Eq.~\eqref{TH21} denoted by $$G\bm{(}S,L, a\bm{)}\triangleq2|\mathcal{N}|(L-1)\alpha^{2SL-1}\beta^{S+1}.$$
		
It can be verified that $G\bm{(}S,L,a\bm{)}$ attains its maximum when $a$ is chosen as the solution of the following equation
		\begin{equation}\nonumber
		\frac{\sin {2a\pi}}{2a\pi}=\frac{S(2L-1)-2}{S(2L+1)},
		\end{equation}
and the solution  $a$ has the same scaling as $\frac{1}{\sqrt{L}}$, that is, $a\in \mathrm{\Theta}(\frac{1}{\sqrt{L}})$. Moreover, with the selected $a$, we have
		\begin{equation}\nonumber
		G\bm{(}S,L, a\bm{)}\geq 2|\mathcal{N}| \left(\frac{S+1}{e S}\right)^{S+1}\frac{L-1}{(2L+1)^{S+1}},
		\end{equation}
which is obtained by the fact that the function ${\left[ \frac{2SL-1}{S(2L+1)}\right] }^{2SL-1}$ decreases monotonically with $L$ and converges to $\frac{1}{e^{S+1}}$, when $S\geq 1$.
\end{proof}

\section{\label{proofTH3}Proof of Lemma~\ref{TH3}}
\begin{proof}
Let us first prove Eq.~\eqref{expectation1} concerning the average of the partial derivative of the cost function with respect to $\theta_{q,n}$.
		
For the Hamiltonian of Eq.~\eqref{toyH},
		\begin{equation}\nonumber
		\mathop{\mathbb{E}}\limits_{\boldsymbol{\theta}}\partial_{\theta_{q, n}}{C}=\sum_{i=1}^{N-1}\mathop{\mathbb{E}}
		\limits_{\boldsymbol{\theta}}\partial_{\theta_{q,n}}C_i,
		\end{equation}
where $C_{i}= \mathrm{Tr}\left[ Z_{i}Z_{i+1}U( \boldsymbol{\theta} )\rho U^{\dagger}( \boldsymbol{\theta} )\right]$. To proceed, assume that $\theta_{q, n}$ lies in the $l$th block. Then, we have
		\begin{align}
		\mathop{\mathbb{E}}\limits_{\boldsymbol{\theta}}\partial_{\theta_{q, n}}{C_{i}}&=\mathop{\mathbb{E}}\limits_{\boldsymbol{\theta}_{[2L]}}\partial_{\theta_{q, n}}{\mathrm{Tr}\left[ Z_{i}Z_{i+1}\rho_{2L}\right]}\nonumber
		\\
		&=\gamma^{2}\mathop{\mathbb{E}}\limits_{\boldsymbol{\theta}_{[2L-1]}}\partial_{\theta_{q, n}}{\mathrm{Tr}\left[ Z_{i}Z_{i+1}\rho_{2L-1}\right]}\label{ER2L}
		\\
		&=\gamma^{4}\mathop{\mathbb{E}}\limits_{\boldsymbol{\theta}_{[2L-2]}}\partial_{\theta_{q, n}}{\mathrm{Tr}\left[ \mathrm{CZ}^{\dagger}_{L}Z_{i}Z_{i+1}\mathrm{CZ}_{L}{\rho_{2L-2}}\right]}\label{ER2L-1}
		\\
		&=\gamma^{4}\mathop{\mathbb{E}}\limits_{\boldsymbol{\theta}_{[2L-2]}}\partial_{\theta_{q, n}}{\mathrm{Tr}\left[ Z_{i}Z_{i+1}\rho_{2L-2}\right]}\label{ER2L-2}
		\\
		&=\gamma^{4(L-l)}\mathop{\mathbb{E}}\limits_{\boldsymbol{\theta}_{[2l]}}\partial_{\theta_{q, n}}{\mathrm{Tr}\left[ Z_{i}Z_{i+1}\rho_{2l}\right]}.\label{ER2l}
		\end{align}
Here, the appearance of $\gamma^2$ in Eq.~\eqref{ER2L} is due to taking the expectations over $\theta_{2L,i}$ and $ \theta_{2L,i+1}$ by employing Eq.~\eqref{Efirstpower1} and taking the expectations over the other $N-2$ parameters in $\boldsymbol{\theta}_{2L}$ by utilizing Eq.~\eqref{Efirstpower3}. Following the same reasoning, Eq.~\eqref{ER2L-1} is derived. In addition, Eq.~\eqref{ER2L-2} is due to Lemma~\ref{CZlemma}.  Eqs. \eqref{ER2L}-\eqref{ER2L-2}  comprise a complete calculation for the $L$th block in Fig.~\ref{fig:setup}(a). For the following $L-l-1$ blocks, we repeat the same procedures and obtain Eq.~\eqref{ER2l}.
		
Now we proceed upon the index $n$. For $n\notin \left\lbrace i, i+1\right\rbrace $, from Eq.~\eqref{Efirstpower4}, we have
		\begin{equation}\nonumber
		\mathop{\mathbb{E}}\limits_{\theta_{q, n}}\partial_{\theta_{q, n}}{\mathrm{Tr}\left[ Z_{i}Z_{i+1}\rho_{q}\right]}=0,
		\end{equation}
which results in $\mathop{\mathbb{E}}\limits_{\boldsymbol{\theta}}\partial_{\theta_{q, n}}{C}=0$.
		
For $n\in \left\lbrace i, i+1\right\rbrace $, we first consider the case where $q=2l$. Note that by employing Eq.~\eqref{Efirstpower2} for $\theta_{q, n}$, we have
		\begin{equation}\label{2l}
		\mathop{\mathbb{E}}\limits_{\boldsymbol{\theta}_{2l}}\partial_{\theta_{2l, n}}{\mathrm{Tr}\left[ Z_{i}Z_{i+1}\rho_{2l}\right]}= \left\{
		\begin{aligned}
		& -\gamma^{2}\mathrm{Tr}\left[ X_{i}Z_{i+1}\rho_{2l-1}\right],\ n=i \\
		& -\gamma^{2}\mathrm{Tr}\left[ Z_{i}X_{i+1}\rho_{2l-1}\right],\ n=i+1.
		\end{aligned}
		\right.
		\end{equation}
Moreover, by employing Eq.~\eqref{Efirstpower1}, we have
		\begin{equation}\label{2l-1}
		\left\{
		\begin{aligned}
		\mathop{\mathbb{E}}\limits_{\boldsymbol{\theta}_{2l-1}}\mathrm{Tr}\left[ X_{i}Z_{i+1}\rho_{2l-1}\right]=& \gamma\mathrm{Tr}\left[\mathrm{CZ}^{\dagger}_{l} X_{i}Z_{i+1}\mathrm{CZ}_{l}{\rho_{2l-2}}\right],\ n=i \\
		\mathop{\mathbb{E}}\limits_{\boldsymbol{\theta}_{2l-1}}\mathrm{Tr}\left[ Z_{i}X_{i+1}\rho_{2l-1}\right]=& \gamma\mathrm{Tr}\left[\mathrm{CZ}^{\dagger}_{l} Z_{i}X_{i+1}\mathrm{CZ}_{l}{\rho_{2l-2}}\right],\ n=i+1.
		\end{aligned}
		\right.
		\end{equation}
Combining Eq.~\eqref{2l} and Eq.~\eqref{2l-1}, we obtain
		\begin{equation}
		\mathop{\mathbb{E}}\limits_{\boldsymbol{\theta}_{2l-1}}\mathop{\mathbb{E}}\limits_{\boldsymbol{\theta}_{2l}}\partial_{\theta_{2l, n}}{\mathrm{Tr}\left[ Z_{i}Z_{i+1}\rho_{2l}\right]}
		=\left\{
		\begin{aligned}
		&-\gamma^{3}\mathrm{Tr}\left[\mathrm{CZ}^{\dagger}_{l} X_{i}Z_{i+1}\mathrm{CZ}_{l}{\rho_{2l-2}}\right],\ n=i \\
		&-\gamma^{3}\mathrm{Tr}\left[\mathrm{CZ}^{\dagger}_{l} Z_{i}X_{i+1}\mathrm{CZ}_{l}{\rho_{2l-2}}\right],\ n=i+1.
		\end{aligned}\label{q=2l}
		\right.
		\end{equation}
		
Similarly, for $n\in \left\lbrace i, i+1\right\rbrace $ and $q=2l-1$, we have
		\begin{align}\nonumber
		\mathop{\mathbb{E}}\limits_{\boldsymbol{\theta}_{2l-1}}\mathop{\mathbb{E}}\limits_{\boldsymbol{\theta}_{2l}}\partial_{\theta_{2l-1, n}}{\mathrm{Tr}\left[ Z_{i}Z_{i+1}\rho_{2l}\right]}
		&=\gamma^{2}\mathop{\mathbb{E}}\limits_{\boldsymbol{\theta}_{2l-1}}\partial_{\theta_{2l-1, n}}{\mathrm{Tr}\left[ Z_{i}Z_{i+1}\rho_{2l-1}\right]}
		\\
		&=\left\{
		\begin{aligned}
		&\gamma^{4}\mathrm{Tr}\left[\mathrm{CZ}^{\dagger}_{l} Y_{i}Z_{i+1}\mathrm{CZ}_{l}{\rho_{2l-2}}\right],\ n=i \\
		&\gamma^{4}\mathrm{Tr}\left[\mathrm{CZ}^{\dagger}_{l} Z_{i}Y_{i+1}\mathrm{CZ}_{l}{\rho_{2l-2}}\right],\ n=i+1.
		\end{aligned}\label{q=2l-1}
		\right.
		\end{align}
		
According to Lemma~\ref{CZlemma}, after applying the $l$th entanglement layer, the generated observables in Eqs.~\eqref{q=2l} and \eqref{q=2l-1} are all  $N$-qubit Pauli strings with the $n$th tensor component unchanged and being either Pauli $X$ or Pauli $Y$.
		
Subsequently, from Eqs.~\eqref{Efirstpower1} and \eqref{Efirstpower3}, taking the expectations over $\boldsymbol{\theta}_{q-1}, \dots, \boldsymbol{\theta}_{1}$ does not change the observables any longer. Moreover, applying the entanglement layers does not change the $n$th tensor component of the observables in Eqs.~\eqref{q=2l} and \eqref{q=2l-1}, which is either Pauli $X$ or Pauli $Y$. Thus, for $n\in \left\lbrace i, i+1\right\rbrace $, we have
		\begin{equation}\label{0expectation}
		\mathop{\mathbb{E}}\limits_{\boldsymbol{\theta}_{[2l]}}\partial_{\theta_{q, n}}{\mathrm{Tr}\left[ Z_{i}Z_{i+1}\rho_{2l}\right]}=c\,\mathrm{Tr}\left[ \boldsymbol{\sigma}_{\boldsymbol{g}_{i}}{\rho_{0}}\right],
		\end{equation}
where $c$ is a positive coefficient, $\boldsymbol{\sigma}_{\boldsymbol{g}_{i}}$ denotes the final Pauli string generated from $Z_{i}Z_{i+1}$, whose $n$th tensor component is Pauli $X$ for $q=2l$ and Pauli $Y$ for $q=2l-1$. Moreover, with the choice of $\rho_{0}=\left| 0\right\rangle\left\langle 0\right|$, Eq.~\eqref{0expectation} is equal to zero, which results in $\mathop{\mathbb{E}}\limits_{\boldsymbol{\theta}}\partial_{\theta_{q, n}}{C}=0$.
		
Therefore, we have proved Eq.~\eqref{expectation1}.  The proof of Eqs.~\eqref{variance}-\eqref{TH38} can be straightforwardly derived from the proof in Appendix~\ref{proofTH1}.
\end{proof}
	
\section{\label{proofTH4}Proof of Theorem~\ref{TH4}}
\begin{proof}
The proof of Eqs.~\eqref{TH41} and \eqref{TH44} can be straightforwardly derived from the proof in Appendix~\ref{proofTH2}. We now focus on proving Eq.~\eqref{TH42} when the number of the elements in vector $\boldsymbol{i}$ belonging to $\left\lbrace 1, 2  \right\rbrace$ is not equal to 1 for all $\boldsymbol{i}\in \mathcal{N}$.
		
Following the notation in Appendix~\ref{proofTH2}, we restate Eq.~\eqref{TH42}. Specifically, if $S^{\boldsymbol{i}}_{1}+S^{\boldsymbol{i}}_{2}\ne 1$ for all $\boldsymbol{i}\in \mathcal{N}$, for an arbitrary trainable parameter $\theta_{q, n}$,
	\begin{equation}
		\mathop{\mathbb{E}}\limits_{\boldsymbol{\theta}}\partial_{\theta_{q, n}}{C}=\sum_{\boldsymbol{i}\in \mathcal{N}}\mathop{\mathbb{E}}\limits_{\boldsymbol{\theta}}\partial_{\theta_{q, n}}{C_{\boldsymbol{i}}}=0,
	\end{equation}
where $C_{\boldsymbol{i}}= \mathrm{Tr}\left[ \boldsymbol{\sigma}_{\boldsymbol{i}}U( \boldsymbol{\theta} )\rho U^{\dagger}( \boldsymbol{\theta} )\right]$. To prove it, we assume that $\theta_{q, n}$ lies in the $l$th block.
		
First, if  $S^{\boldsymbol{i}}_{1}+S^{\boldsymbol{i}}_{2}= 0$, we can directly generalize the proof in Appendix~\ref{proofTH3} to obtain
	\begin{equation}
		\mathop{\mathbb{E}}\limits_{\boldsymbol{\theta}}\partial_{\theta_{q, n}}{C_{\boldsymbol{i}}}=0.
	\end{equation}
		
Next, if $S^{\boldsymbol{i}}_{1}+S^{\boldsymbol{i}}_{2}\geq2$, note that
	\begin{align}
		\mathop{\mathbb{E}}\limits_{\boldsymbol{\theta}}\partial_{\theta_{q, n}}{C_{\boldsymbol{i}}}&=\mathop{\mathbb{E}}\limits_{\boldsymbol{\theta}_{[2L]}}\partial_{\theta_{q, n}}{\mathrm{Tr}\left[ \boldsymbol{\sigma}_{\boldsymbol{i}}\rho_{2L}\right]}\nonumber
		\\
		&=\gamma^{S^{\boldsymbol{i}}_{1}+S^{\boldsymbol{i}}_{3}}\mathop{\mathbb{E}}\limits_{\boldsymbol{\theta}_{[2L-1]}}\partial_{\theta_{q, n}}{\mathrm{Tr}\left[ \boldsymbol{\sigma}_{\boldsymbol{i}}\rho_{2L-1}\right]}\label{EER2L}
		\\
		&=\gamma^{S^{\boldsymbol{i}}+S^{\boldsymbol{i}}_{3}}\mathop{\mathbb{E}}\limits_{\boldsymbol{\theta}_{[2L-2]}}\partial_{\theta_{q, n}}{\mathrm{Tr}\left[ \mathrm{CZ}^{\dagger}_{L}\boldsymbol{\sigma}_{\boldsymbol{i}}\mathrm{CZ}_{L}{\rho_{2L-2}}\right]}\label{EER2L-1}
		\\
		&=\gamma^{(S^{\boldsymbol{i}}+S^{\boldsymbol{i}}_{3})(L-l)}\mathop{\mathbb{E}}\limits_{\boldsymbol{\theta}_{[2l]}}\partial_{\theta_{q, n}}{\mathrm{Tr}\left[ \boldsymbol{\sigma}_{\boldsymbol{g}_{\boldsymbol{i},2l}}\rho_{2l}\right]}\label{EER2l}.
	\end{align}
 Here, Eqs.~\eqref{EER2L} and \eqref{EER2L-1} are derived from the same reasoning as Eqs.~\eqref{ER2L} and \eqref{ER2L-1}, respectively. According to Lemma~\ref{CZlemma}, after applying the entanglement layer in  Eq.~\eqref{EER2L-1}, the position set of the tensor components being X or Y does not change, being equal to $I_{S^{\boldsymbol{i}}_{1}}\cup I_{S^{\boldsymbol{i}}_{2}}$. Then we  repeat the above procedures for the following  $L-l-1$ blocks and obtain  Eq.~\eqref{EER2l}, where $\boldsymbol{\sigma}_{\boldsymbol{g}_{\boldsymbol{i},2l}}$ denotes the generated observable after calculating the last $L-l$ blocks with the invariant set $I_{S^{\boldsymbol{g}_{\boldsymbol{i},2l}}_{1}}\cup I_{S^{\boldsymbol{g}_{\boldsymbol{i},2l}}_{2}}=I_{S^{\boldsymbol{i}}_{1}}\cup I_{S^{\boldsymbol{i}}_{2}}$. By now, $S^{\boldsymbol{g_{i,2l}}}_1+S^{\boldsymbol{g_{i,2l}}}_2=S^{\boldsymbol{i}}_{1}+S^{\boldsymbol{i}}_{2}\geq2$.
		
For the $l$th block where $\theta_{q,n}$ lies, from Lemma~\ref{firstpowerlemma}, when taking the expectations over $\boldsymbol{\theta}_{2l}$ and $\boldsymbol{\theta}_{2l-1}$, the number of the tensor components being $X$ or $Y$ changes, which either increases or decreases by at most 1. In fact, the possible change only happens when employing Eq.~\eqref{Efirstpower2}. Moreover, according to Lemma~\ref{CZlemma}, after calculating the $l$th block, we have
\begin{equation*}
     S^{\boldsymbol{g_{i,2(l-1)}}}_1+S^{\boldsymbol{g_{i,2(l-1)}}}_2\geq S^{\boldsymbol{i}}_{1}+S^{\boldsymbol{i}}_{2}-1\geq1.
\end{equation*}	
		
During the calculations of the expectations over $\boldsymbol{\theta}_{[2(l-1)]}$, from Lemmas~\ref{firstpowerlemma} and \ref{CZlemma}, the number of the tensor components being $X$ or $Y$ does not change any longer. Since
\begin{equation*}  \mathop{\mathbb{E}}\limits_{\boldsymbol{\theta}}\partial_{\theta_{q, n}}{C_{\boldsymbol{i}}}\propto \mathrm{Tr}\left[\boldsymbol{\sigma_{g_{i,0}}}\rho_0\right],
\end{equation*}
and at least one of the tensor components in $\boldsymbol{\sigma_{g_{i,0}}}$ is $X$ or $Y$, with the selected $\rho_0=|0\rangle\langle0|$, we have
		\begin{equation}\label{key}
		\mathop{\mathbb{E}}\limits_{\boldsymbol{\theta}}\partial_{\theta_{q, n}}{C_{\boldsymbol{i}}}=0.
		\end{equation}
		
Thus, Eq.~\eqref{TH42} is proved and so is Eq.~\eqref{TH43}.
\end{proof}

\end{document}